\documentclass[11pt]{scrartcl}
\usepackage[english]{babel}
\usepackage[T1]{fontenc}
\usepackage[utf8]{inputenc}
\usepackage{color,amssymb,amsfonts,amsthm,mathrsfs,graphicx,setspace,amsmath}
\usepackage{caption}
\usepackage[left=3cm,right=3cm,top=2.5cm,bottom=3cm]{geometry}
\usepackage{thmtools}
\usepackage[mathscr]{eucal}
\usepackage{fullpage}
\usepackage{xcolor}
\usepackage{tikz}
\usetikzlibrary{shapes.geometric,hobby,calc} \usetikzlibrary{decorations.pathmorphing}
\usetikzlibrary{decorations.text}
\usetikzlibrary{shapes.misc}
\usetikzlibrary{calc}
\usetikzlibrary{arrows.meta}
\usetikzlibrary{decorations,shapes}
\usepackage{hyperref}
\usepackage{nameref,cleveref}
\usepackage[square, numbers, comma, sort&compress]{natbib}  \usepackage{enumitem}
\setlist{nosep}

\usepackage{macros}
\usepackage{tikzmacros}
\usepackage{dsfont}
\usepackage{soul}
\usepackage{xspace}
\usepackage[prependcaption,colorinlistoftodos]{todonotes}
\usepackage{titling}
\usepackage{marginnote}
\usepackage{makeidx}
\usepackage{subcaption}
\usepackage{float}
\usepackage{thm-restate}

\newcommand{\Symbol}[1]{}

\hypersetup{
	colorlinks=true,
	linkcolor=myBlue,
	citecolor=myBlue,
	urlcolor=myLightBlue,
	bookmarksopen=true,
	bookmarksnumbered,
	bookmarksopenlevel=2,
	bookmarksdepth=3,
	hypertexnames=false }

\usetikzlibrary[positioning]
\usetikzlibrary{patterns}
\usetikzlibrary{shapes}
\usetikzlibrary{decorations.markings}
\usetikzlibrary{arrows}

\graphicspath{{./}{../Figures/}}

\title{Cycles of Well-Linked Sets \rom{2}: an Elementary Bound for the Directed Grid Theorem\thanks{The results in this manuscript were also presented in Milani's PhD thesis~\cite{Milani2024}. A preliminary version of this paper was published at \emph{FOCS 2024}\cite{hkmm24cows}.}}

\predate{}
\date{}
\postdate{}

\preauthor{}

\DeclareRobustCommand{\authorthing}{
	\begin{center}
	    Meike Hatzel\thanks{Partially supported by the Institute for Basic Science (IBS-R029-C1).}
            \smallskip\\[0em] Technical University Darmstadt, Germany\\[0em]
	    \href{mailto:research@meikehatzel.com}{research@meikehatzel.com}\\[0em]
	    \bigskip
	    Stephan Kreutzer\smallskip\\[0em]
        Technische Universität Berlin\\[0em]
        \href{mailto:stephan.kreutzer@tu-berlin.de}{stephan.kreutzer@tu-berlin.de}
	    \\[0em]
	    \bigskip
	    Marcelo Garlet Milani\smallskip
			\\[0em]
 			National Institute of Informatics, Tokyo, Japan
			\\[0em]
			\href{mailto:research@mgarletmilani.com}{research@mgarletmilani.com}
	    \\[0em]
	    \bigskip
	    Irene Muzi\smallskip\\[0em]
            Universität Hamburg\\[0em]
            \href{mailto:irene.muzi@gmail.com}{irene.muzi@gmail.com}
          \end{center}}
\author{\authorthing}
\postauthor{}

\newenvironment{cenv}{\begin{list}{}{            \setlength{\labelwidth}{1.5em}      \setlength{\leftmargin}{\labelwidth}      \addtolength{\leftmargin}{\labelsep}      \setlength{\listparindent}{0em}      \setlength{\topsep}{10pt}      \setlength{\itemsep}{5pt}      \setlength{\parsep}{0pt}    }
  }{
  \end{list}
}

\newcounter{claimcounter}

\newcommand{\DrawPicture}[1]{}
\newcommand{\DeferTask}[1]{}
\renewcommand{\DeferTask}[1]{}
\renewcommand{\DrawPicture}[1]{}
\renewcommand*{\Task}[2][mediumpriority]{}
\renewcommand*{\TaskPerson}[3][]{}

\begin{document}
\setlength{\abovedisplayskip}{2pt}
\setlength{\belowdisplayskip}{2pt}
\setlength{\abovedisplayshortskip}{0pt}
\setlength{\belowdisplayshortskip}{0pt}
\maketitle

\begin{abstract}

		In 2015, Kawarabayashi and Kreutzer proved the Directed Grid Theorem --- the generalisation of the well-known Excluded Grid Theorem to directed graphs --- confirming a conjecture by Reed, Johnson, Robertson, Seymour and Thomas from the mid-nineties.
		The theorem states that there is a function $f$ such that every digraph of directed \treewidth $f(k)$ contains a cylindrical grid of order $k$ as a butterfly minor.
		However, the given function grows faster than any non-elementary function of the size of the grid minor.
    More precisely, it is larger than a power tower whose height depends on the size of the grid.
		
		In this paper, we present an alternative proof of the Directed Grid Theorem which is conceptually much simpler, more modular in composition and improves the upper bound for the function \(f\) to a power tower of height 22.
  
        A key concept of our proof is a new structure called \emph{cycles of well-linked sets (CWS)}.
        We show that any digraph of large directed \treewidth contains a large CWS, which in turn contains a large cylindrical grid.

				  		\end{abstract}

\clearpage

\newpage

\tableofcontents

\newpage

\section{Introduction}

	The Excluded Grid Theorem by Robertson and Seymour is a central result in the study of graph minors and is the first major building block of their Graph Minors project~\cite{GM-series}.
    Additionally, the theorem has found numerous applications beyond its original scope, for instance in the theory of graph algorithms (see, for example,~\cite {cygan2015parameterized}).

	Based on a conjecture by Reed, Johnson, Robertson, Seymour and Thomas from the mid-nineties~\citep{johnson2001directed}, Kawarabayashi and Kreutzer proved in 2015~\cite{kawarabayashi2015directed} an Excluded Grid Theorem for directed graphs, that is, the existence of a function $f$ such that every digraph of directed \treewidth $f(k)$ contains a \emph{cylindrical grid} of order $k$ as a butterfly minor.
    In addition, they provided an XP algorithm which produces a directed tree-decomposition of width at most $f(k)$ or finds a cylindrical grid of order $k$ as a butterfly minor.
	Campos et al.~\cite{campos2019adapting} improved their result from XP to FPT.

	The Directed Grid Theorem has been used to prove advanced results in digraph structure theory~\cite{giannopoulou2020directed,giannopoulou2020canonical,2022directedtangle}, Erd\H os-P\'osa/cycle packing~\cite{masavrik2022packing,amiri2016erdos,zbMATH07088645,kawarabayashi2020half} and matching theory~\cite{matchinggrid2019,GiannopoulouKW24} as well as for algorithmic results~\cite{zbMATH06600957,edwards2017half}.
	Intuitively, the Directed Grid Theorem yields a \emph{win-win} situation.
	If a digraph contains a large cylindrical grid as a butterfly minor, the grid's uniform structure can be useful in solving the task at hand.
	Otherwise, if the digraph does not contain a large cylindrical grid as a butterfly minor, we know it has bounded directed \treewidth and we can use its directed tree-decomposition.

	As the original function \(f\) by Kawarabayashi and Kreutzer is very large (specifically, it grows faster than a tower whose height is dependent on the size of the grid), improving it has a major impact on all results which depend on it.
	It is thus natural to ask what the best bound for \(f\) is.
	Indeed, already determining whether it is an \emph{elementary} function	\footnote{In rough terms, a function is elementary if it can be constructed using a finite sequence of arithmetic operations, exponentials, logarithms and trigonometric functions, starting from constants and variables.}
	would vastly improve upon the bound given by~\cite{kawarabayashi2015directed}.
	Further, improving the bounds of the Directed Grid Theorem may require the development of powerful frameworks whose usefulness extends beyond their initial scope, as was observed with the path of sets-system framework introduced by~\cite{chekuri2016polynomial,chuzhoy2021tighterbounds} in order to obtain polynomial bounds for the (undirected) Grid Theorem.

	Our proof starts with a result from~\cite{COSSI}, where we introduced the concept of \emph{paths of order-linked sets, paths of well-linked sets} and \emph{cycles of well-linked sets} and obtained an elementary bound for Younger's conjecture, improving upon the non-elementary bound given by~\cite{reed1996packing}.

	On a high level, a path of well-linked sets consists of an ordered set of clusters, each in turn consisting of two sets of vertices $A$ and $B$, where $A$ is well-linked to $B$ such that each cluster is connected to the next by a linkage and all the clusters and linkages are pairwise internally vertex-disjoint.
	For the path of well-linked sets we obtain in~\cite{COSSI}, we have the additional property that the last cluster is well-linked back to the first one.
	However, the linkages connecting these two clusters may arbitrarily intersect the path of well-linked sets.

	In~\cite{COSSI}, we also introduced a framework based on temporal digraphs in order to construct a path of well-linked sets and a path of order-linked sets in digraphs of large directed \treewidth.
	We then showed that we can obtain \emph{fences} and \emph{acyclic grids} from these two objects.

	This paper builds upon the abovementioned work and proves an elementary upper bound for the Directed Grid Theorem.
	Our starting point is the following statement (\(\PowerTower{h}{n}\) denotes a power tower of height \(h\) with base 2 and an \(n\) \say{on top}, and \(\Polynomial{d}{x_1, \ldots, x_k}\) denotes a polynomial of degree \(d\) on the variables \(x_1, \ldots, x_k\)).

	\begin{restatable}[{\cite[Theorem 9.9]{COSSI}}]{theorem}{thmHighDTWToPOWL}
		\label{thm:high_dtw_to_POSS_plus_back-linkage}
		There exists a function \(\bound{thm:high_dtw_to_POSS_plus_back-linkage}{t}{w,\ell} \in \PowerTower{7}{\Polynomial{25}{w, \ell}}\) such that every digraph $D$ with $\dtw{D} \geq \bound{thm:high_dtw_to_POSS_plus_back-linkage}{t}{w, \ell}$ contains a path of well-linked sets $(\mathcal{S} = ( S_0, S_1, \dots,$ $S_{\ell}), \mathscr{P})$ of width $w$ and length $\ell$ and, additionally, $B(S_\ell)$ is well-linked to $A(S_0)$ in $D$.
	\end{restatable}

	Our first technical contribution is to prove that a digraph containing a path of well-linked sets as described above also contains a cycle of well-linked sets.
	A cycle of well-linked sets consists of a path of well-linked sets together with a linkage \(\mathcal{L}\) from the last to the first cluster, where \(\mathcal{L}\) is internally disjoint from the path of well-linked sets.
	The functions
	\(\bound{thm:POSS_plus_back-linkage_to_COSS}{w'}{}\),
	\(\bound{thm:POSS_plus_back-linkage_to_COSS}{r}{}\) and
	\(\bound{thm:POSS_plus_back-linkage_to_COSS}{\ell'}{}\)
	used below are defined later, but we observe here that
	\begin{align*}
		\bound{thm:POSS_plus_back-linkage_to_COSS}{w'}{w,\ell} & \in \PowerTower{2}{\Polynomial{97}{\ell, w}},
		\\[0em]
		\bound{thm:POSS_plus_back-linkage_to_COSS}{r}{w,\ell} & \in \PowerTower{13}{\Polynomial{97}{\ell, w}} \text{ and }
		\\[0em]
		\bound{thm:POSS_plus_back-linkage_to_COSS}{\ell'}{w,\ell} & \in \PowerTower{14}{\Polynomial{97}{\ell, w}}.
	\end{align*}
	\begin{restatable}{theorem}{thmPOWLtoCOWS}
  	\label{thm:POSS_plus_back-linkage_to_COSS}
		Let $w, \ell$ be integers, let $\Brace{\mathcal{S} = \Brace{ S_0, S_1, \dots, S_{\ell'}},\mathscr{P}}$ be a strict path of well-linked sets of width $w'$ and length $\ell'$ and let $\mathcal{R}$ be a $B(S_{\ell'})$-$A(S_0)$-linkage of order $r$.
		If $w' \geq \bound{thm:POSS_plus_back-linkage_to_COSS}{w'}{w,\ell}$, $r \geq \bound{thm:POSS_plus_back-linkage_to_COSS}{r}{w,\ell}$ and $\ell' \geq \bound{thm:POSS_plus_back-linkage_to_COSS}{\ell'}{w, \ell}$, then $\ToDigraph{\Brace{\mathcal{S},\mathscr{P}} \cup \mathcal{R}}$ contains a cycle of well-linked sets of width $w$ and length $\ell$.
	\end{restatable}

	While at first it may look like the path of well-linked sets provided by~\cref{thm:high_dtw_to_POSS_plus_back-linkage} is very close to a cycle of well-linked sets, ensuring that the linkage connecting the last cluster back to the first one is internally disjoint from the path of well-linked sets is far from being a simple task, as the intersections may yield mostly short cycles which are not very helpful in building a large cylindrical grid.

	Constructing such a \emph{back-linkage} is a step in the proof of the Directed Grid Theorem which finds no parallel in the Undirected Grid Theorem, as there is no need to \say{close cycles} in the undirected setting.
	More than just a mere technical step, a better understanding of this part of the proof can reveal further insights into the differences between the Directed and Undirected Grid Theorems.
	In particular, if one could obtain linear upper bounds for this part of the proof,
	then this could be evidence that the functions for the Directed and for the Undirected Grid Theorems are asymptotically the same.
	Conversely, if one could obtain super-linear lower bounds, then this would be evidence that these functions might be distinct.

	In order to obtain the back-linkage, we introduce the concept of \(c\)-\emph{horizontal webs}, which are a special case of the webs used in~\cite{COSSI,kawarabayashi2015directed}, and apply our framework from~\cite{COSSI} in order to both construct a cylindrical grid from \(2\)-horizontal webs as well as obtaining 2-horizontal webs from the path of well-linked sets described above.
	
	Moreover, our modular approach leads us to two questions regarding obtaining 2-horizontal webs and constructing a cylindrical grid from them (see~\cref{sec:conclusion} for an in-depth discussion), which are not only natural questions on their own, but are also both necessary and sufficient for obtaining better bounds for the Directed Grid Theorem.
	In this way, we effectively reduce the complexity of further improvements of the theorem, as each question can be answered independently.

	As in~\cite{COSSI}, we first look for an intermediate structure with similar connectivity properties as the cylindrical grid, which we call \emph{cycle of well-linked sets}.
	With the tools developed in~\cite{COSSI}, we can easily show that every large cycle of well-linked sets contains a large cylindrical grid, allowing us to deduce our main result from~\cref{thm:POSS_plus_back-linkage_to_COSS}.
	The exact function \(\bound{thm:high-treewidth-implies-cylindrical-grid}{dtw}{}\) is defined later, but we note here that \(\bound{thm:high-treewidth-implies-cylindrical-grid}{dtw}{k} \in \PowerTower{22}{\Polynomial{9}{k}}\).
	
	\begin{restatable}{theorem}{thmCylindricalGrid}
		\label{thm:high-treewidth-implies-cylindrical-grid}
		Every digraph $D$ with $\dtw{D} \geq \bound{thm:high-treewidth-implies-cylindrical-grid}{dtw}{k}$ contains a cylindrical grid of order $k$ as a butterfly minor.
	\end{restatable}

	The statement above has several implications for results using the Directed Grid Theorem.
	In~\cref{sec:erdos-posa}, we analyse one such application, namely the \emph{Erd\H os-Pósa property} of digraphs, and deduce that our results also improve the non-elementary bounds previously obtained by~\cite{amiri2016erdos} to elementary ones.

	The remainder of the paper is organised as follows.
	Basic notation and definitions are listed in~\cref{sec:preliminaries}, necessary definitions and statements from~\cite{COSSI} are repeated in~\cref{sec:powls-framework},~\cref{sec:constructing cows} contains our main proof and concluding remarks are given in~\cref{sec:conclusion}.

\subsection{Erd\H os-P\'osa property for directed graphs}
\label{sec:erdos-posa}

One of the applications of the Directed Grid Theorem is proving the Erd\H os-P\'osa property of digraphs.
We say that a digraph $H$ has the Erd\H os-P\'osa property if there is a function $\ell_H$ such that, in any digraph $D$, we can find $n$ disjoint $H$-butterfly minors or $\ell_H(n)$ vertices covering all $H$-butterfly minors.
Amiri, Kawarabayashi, Kreutzer and Wollan~\cite{amiri2016erdos} proved that the Erd\H os-P\'osa property holds for strongly connected digraphs precisely when they are minors of a cylindrical grid. 

\begin{theorem}\cite[Theorem 4.1]{amiri2016erdos}\label{thm:akkw}
	Let $H$ be a strongly connected digraph.
	There is a function \(\ell_H\) such that the following holds.
	$H$ has the Erd\H os-P\'osa property for butterfly (topological) minors if, and only if, there is a cylindrical grid (wall) of order $c$ of which $H$ is a butterfly (topological) minor.
	Furthermore, for every fixed strongly connected digraph $H$ satisfying these conditions and every $k$ there is a polynomial time algorithm which, given a digraph $D$ as input, either computes $k$ disjoint (butterfly or topological) models of $H$ in $D$ or a set $S$ of $\leq \ell_H(k)$ vertices such that $D-S$ does not contain a model of $H$.
\end{theorem}

Their proof uses the following lemma, which we restate to make the bounds explicit.

\begin{lemma}[{\cite[Lemma 4.2]{amiri2016erdos}}]
	\label{lem:akkw2}
	Let $G$ be a directed graph with $\dtw{G} \leq w$. For each strongly connected directed graph $H$, the graph $G$ has either $k$ disjoint copies of $H$ as a topological (butterfly) minor, or contains a set $T$ of at most $k \cdot (w+1)$ vertices such that $H$ is not a topological (butterfly) minor of $G-T$. 
\end{lemma}

The previous two results and our improved bound for the Directed Grid Theorem yield the following.

\newcommand{\Butterflies}[1]{\textcolor{myGreen}{#1}}
\begin{theorem}\label{thm:ep-ele}
	Let \(H\) be a strongly connected digraph that is a topological \Butterflies{(butterfly)} minor of the cylindrical wall \Butterflies{(grid)} of order \(c\).
	Then for any digraph $D$ and any natural number $k$ either $D$ contains $k$ disjoint copies of $H$ as a topological \Butterflies{(butterfly)} minor or a set $S$ of at most \(k(\bound{statement:high-dtw-to-wall}{dtw}{kc} + 1)\) \Butterflies{($k(\bound{thm:high-treewidth-implies-cylindrical-grid}{dtw}{kc} + 1)$)} vertices such that \(H\) is not a topological \Butterflies{(butterfly)} minor of \(D - S\).
\end{theorem}
\begin{proof}
	We consider the case where \(H\) is a butterfly minor of a cylindrical grid.
	The other case follows analogously by applying \cref{statement:high-dtw-to-wall} instead of \cref{thm:high-treewidth-implies-cylindrical-grid}.

	If $\dtw{D}\geq \bound{thm:high-treewidth-implies-cylindrical-grid}{dtw}{kc}$, then by~\cref{thm:high-treewidth-implies-cylindrical-grid} $D$ contains a cylindrical grid of order $kc$ and hence $k$ disjoint copies of $H$ as a butterfly minor. Otherwise, we can apply~\cref{lem:akkw2}.
\end{proof}

When $H$ is a cycle of length two, this result is equivalent to Younger's conjecture, which states that for all directed graphs $D$, there exists a function $f$ such that if $D$ does not contain $k$ disjoint cycles, $D$ contains a feedback vertex set of size bounded by $f(k)$.
This conjecture was proven true in 1996 by Reed, Robertson, Seymour and Thomas, and~\cref{thm:ep-ele} is an improvement of their non-elementary bound.
This bound is not the best we can achieve, and in fact, in~\cite{COSSI} we prove that the function $f$ is at most a power tower of height $8$.

\section{Preliminaries}
\label{sec:preliminaries}

In this section, we establish our notation and recall standard concepts and results from the literature used throughout the paper.

\paragraph{Sequences, sets and functions.}
Given sequences $S_1 \coloneqq (x_1, x_2, \dots, x_{j})$ and $S_2 \coloneqq (y_1, y_2, \dots$, $y_{k})$, we write $S_1 \cdot S_2$ for the sequence $S_3 \coloneqq (x_1, x_2, \dots, x_{j}, y_1, y_2, \dots, y_{k})$. 
We say that $S_1 \cdot S_2$ is a \emph{decomposition} of $S_3$.
The following is a well-known theorem about sequences of numbers due to Erd\H{o}s and Szekeres.
\begin{theorem}[\cite{erdosszekeres1935}]
	\label{thm:erdos_szekeres}
	Let $r,s \in\N$.
	Every sequence of distinct numbers of length at least $\Brace{r-1}\Brace{s-1}+1$ contains a monotonically increasing subsequence of length $r$ or a monotonically decreasing subsequence of length $s$.
\end{theorem}

\paragraph{Power towers and polynomials.}
Let \(d\) be an integer and \(V = \{x_1, \ldots, x_k\}\) a set of variables.
A \emph{polynomial} of degree \(d\) over \(V\) is a function \(p(x_{1}, x_{2}, \ldots, x_{k})\) of the form \(p(x_{1}, x_{2}, \ldots, x_{k}) = \sum_{i=1}^n (c_i \prod_{j=1}^k x_j^{e_{j,i}})\), where for each \(1 \leq i \leq n\) and each \(1 \leq j \leq k\) we have that \(c_i \in \Reals\), \(e_{j,i} \in \Naturals\) and \(\sum_{j=1}^k e_{j,i} \leq d\).
We write \(\Polynomial{d}{x_{1}, x_{2}, \ldots, x_{k}}\) for the set of all functions \(f\) for which there is a polynomial \(p\) of degree \(d\) over the variable set \(x_{1}, x_{2}, \ldots, x_{k}\) such that \(f(x_{1}, x_{2}, \ldots, x_{k}) \in \Oh(p(x_{1}, x_{2}, \ldots, x_{k}))\).

We define \emph{power towers} as follows.
Given an integer \(h\) and a set of functions \(F\) over a set of variables \(V\), we define a set of functions \(\PowerTower{h}{F}\) recursively as follows.
We set \(\PowerTower{0}{F} = F\) and define \(\PowerTower{h}{F}\) as  \(\Set{f : \Reals^{\Abs{V}} \to \Reals \mid f \in \Oh(2^{g(V)}), g \in \PowerTower{h-1}{F}}\) for \(h > 1\).
If \(F = \Polynomial{d}{V}\), we say that a function \(f \in \PowerTower{h}{F}\) is a \emph{power tower} of height \(h\).

\paragraph{Graphs and digraphs.} 
We denote by $E(G)$ the edge/arc set of a graph $G$, directed or not, and by $V(G)$ its vertex set.
We often use $G$ for undirected graphs and $D$ for directed graphs (also called \emph{digraphs}). 

Let $D$ be a digraph. 
Given a set of vertices $X \subseteq \V{D}$, we write $D - X$ for the digraph $(Y \coloneqq \V{D} \setminus X$, $\A{D} \cap (Y \times Y))$. 
Similarly, given a set of arcs $F \subseteq \A{D}$, we write $D - F$ for the digraph $(\V{D}, \A{D} \setminus F)$.

If $D$ is a digraph and $v \in V(D)$, then  $\InN{D}{v} \coloneqq \{ u \in V \mid (u,v) \in \A{D}\}$ is the set of \emph{in-neighbours} and $\OutN{D}{v} \coloneqq \{ u \in V \mid (v,u) \in \A{D}\}$ the set of \emph{out-neighbours} of $v$.
By $\Indeg[D]{v} \coloneqq \Abs{\InN{}{v}}$ we denote the \emph{in-degree} of $v$ and by $\Outdeg[D]{v} \coloneqq \Abs{\OutN{}{v}}$  its out-degree.
When working with a set or another structure $X$ containing digraphs, we write $\ToDigraph{X}$ to mean the digraph obtained by taking the union of all digraphs in $X$.

\paragraph{Paths and walks.}
A \emph{walk} of length $\ell$ in a digraph $D$ is a sequence of vertices \(W \coloneqq \Brace{v_0, v_1, \dots, v_{\ell}}\) such that $\Brace{v_i, v_{i+1}} \in \A{D}$, for all $0 \leq i < \ell$.
We write $\Start{W}$ for $v_0$ and $\End{W}$ for $v_\ell$ and say that $W$ is a $v_0$-$v_\ell$-walk.

A walk $W \coloneqq \Brace{v_0, v_1, \dots, v_{\ell}}$ is called a \emph{path} if no vertex appears twice in it and it is called a \emph{cycle} if $v_0 = v_\ell$ and $v_i \neq v_j$ for all $0 \leq i < j < \ell$. 

We often identify a walk $W$ in $D$ with the corresponding subgraph and write $V(W)$ and $E(W)$ for the set of vertices and arcs appearing on it.

Given two walks $W_1 \coloneqq (x_1, x_2, \dots, x_{j})$ and $W_2 \coloneqq (y_1, y_2, \dots, y_{k})$ with $\End{W_1} = \Start{W_2}$, we make use of the concatenation notation for sequences and write $W_1 \cdot W_2$ for the walk $W_3 \coloneqq (x_1, x_2, \dots, x_{j}, y_2, y_3, \dots, y_{k})$.
We say that $W_1 \cdot W_2$ is a decomposition of $W_3$.
If $W_1$ or $W_2$ is an empty sequence, then the result of $W_1 \cdot W_2$ is the other walk (or the empty sequence if both walks are empty).

\paragraph{Specific digraphs.}
We denote the digraph of a path on $k$ vertices by $\Pk{k}$.
For the \emph{bidirected path on $k$} vertices, we write $\biPk{k} \coloneqq (\Set{u_1, u_2, \dots, u_{k}}, \{(u_i, u_j) \mid 1 \leq i, j \leq k \text{ and } \Abs{i - j} = 1\})$.
The \emph{cycle on $k$ vertices} is given by $\Ck{k} \coloneqq (\{u_0, u_1, \dots, u_{k - 1}\}, \{(u_{i}, u_{i+1 \mod k}) \mid 0 \leq i < k\})$.

\paragraph{Connectivity.} A digraph $D$ is said to be \emph{strongly connected} if for every $u,v \in V$ there is a $u$-$v$-path \textbf{and} a $v$-$u$-path in $D$.
We say $D$ is \emph{unilateral} if for every $u,v \in V$ there is a $u$-$v$-path \textbf{or} a $v$-$u$-path in $D$.

\paragraph{Linkages and separators.}
Let $A, B \subseteq V(D)$.
An $A$-$B$-walk is a walk $W$ that starts in $A$ and ends in $B$.
A set $X \subseteq \V{D}$ is an \emph{$A$-$B$-separator} if there are no $A$-$B$-paths in $D - X$. 

A \emph{linkage} in $D$ is a set $\LLL$ of pairwise vertex-disjoint paths.
The \emph{order} $|\LLL|$ of $\LLL$ is the number of paths it contains.

Given a linkage \(\mathcal{L}\) and a digraph \(H\), we say that \(\mathcal{L}\) is \emph{internally disjoint} from \(H\) if \(\V{\mathcal{L}} \cap \V{H} \subseteq \Start{\mathcal{L}} \cup \End{\mathcal{L}}\).

An \emph{$A$-$B$-linkage}  of order $k$ is a linkage $\LLL \coloneqq \{L_1, L_2, \dots, L_{k}\}$ such that $\Start{L_i} \in A$ and $\End{L_i} \in B$ for all $1 \leq i \leq k$.
We write $\Start{\mathcal{L}}$ for the set $\{\Start{L_i} \mid L_i \in \mathcal{L}\}$ and $\End{\mathcal{L}}$ for the set $\{\End{L_i} \mid L_i \in \mathcal{L}\}$.
We also extend the notation for path concatenation to linkages.
Given linkages $\mathcal{P} = \Set{P_1, P_2, \dots, P_{k}}$ and $\mathcal{Q} = \Set{Q_1, Q_2, \dots, Q_{k}}$ such that $\End{\mathcal{P}} = \Start{\mathcal{Q}}$, we write $\mathcal{P} \cdot \mathcal{Q}$ for the linkage $\Set{P_a \cdot Q_b \mid P_a \in \mathcal{P}, Q_b \in \mathcal{Q} \text{ and } \End{P_a} = \Start{Q_b}}$. 

It is often convenient to use a linkage $\mathcal{L}$ as a function $\mathcal{L}: \Start{\mathcal{L}} \to \End{\mathcal{L}}$. 
The expression $\Fkt{\mathcal{L}}{a} = b$ then means that $\mathcal{L}$ contains a path starting in $a$ and ending in $b$. 

We frequently use the following classical result by Menger~\cite{menger}.
\begin{theorem}[Menger's Theorem~\cite{menger}]
	\label{thm:menger}
	Let $D$ be a digraph, $A,B \subseteq \V{D}$ with $\Abs{A} = \Abs{B}$.
	There is an $A$-$B$-linkage of size $k$ in $D$ if and only if every $A$-$B$-separator has size at~least~$k$.
\end{theorem}

Let $D$ be a digraph, $A, B \subseteq V(D)$ and $c \in \N$.
An \emph{$A$-$B$-linkage with congestion $c$} in $D$ is a set $\LLL$ of $A$-$B$-paths such that no vertex of $V(D)$ occurs in more than $c$ distinct paths in $\LLL$.
A linkage of congestion $1$ is called \emph{integral}, and a linkage of congestion $2$ is called \emph{half-integral}.
 
A simple application of~\cref{thm:menger} yields the following lemma (see, for example,~\cite{kawarabayashi2015directed} for a proof). 
\begin{lemma}[{\cite{kawarabayashi2015directed}}]
			\label{lemma:half_integral_to_integral_linkage}
	Let $D$ be a digraph, $A,B \subseteq \V{D}$.
	If there is an $A$-$B$-linkage of order $ck$ and congestion $c$ in $D$, then there is an integral $A$-$B$-linkage of order $k$ in $D$. 
\end{lemma}

Throughout the paper, we frequently work with a special kind of linkage which we define next.

\begin{definition}[minimal linkage]
	\label{def:H-minimal}
	Let $D$ be a digraph, $H\subseteq D$ be a subgraph, and $\LLL$ be a linkage of order $k$.
    $\LLL$ is \emph{minimal with respect to $H$}, or \emph{$H$-minimal}, if for all arcs $e \in \bigcup_{L\in \LLL} \E{L} \setminus \E{H}$ there is no $\Start{\LLL}$-$\End{\LLL}$-linkage of order $k$ in the graph $(\LLL \cup H) - e$.
\end{definition}

Given a linkage $\mathcal{L}$ in a digraph $D$ and a subgraph $H \subseteq D$, we can obtain a linkage $\mathcal{L}'$ with the same order and same endpoints as $\mathcal{L}$ which is $H$-minimal.
This can be done by iteratively removing arcs $e \in \A{\mathcal{L}} \setminus \A{H}$ for which a $\Start{\mathcal{L}}$-$\End{\mathcal{L}}$-linkage of order $\Abs{\mathcal{L}}$ avoiding $e$ exists.

Minimal linkages were used extensively in~\cite{kawarabayashi2015directed}.
The idea is that, when constructing paths of an $H$-minimal linkage $\LLL$, we always prefer to use arcs of $H$ over arcs not in $E(H)$.
This implies the following property, which we exploit frequently in our proofs.

\begin{definition}[weak minimality]
	\label{def:weak_minimality}
	A linkage $\mathcal{L}$ in a digraph $D$ is \emph{weakly $k$-minimal} with respect to a subgraph $H$ of $D$ if for every $L_1 \cdot e \cdot L_2 \in \mathcal{L}$ with $e \in \E{\mathcal{L}} \setminus \E{H}$ there is a $\V{L_1}$-$\V{L_2}$-separator of size at most $k-1$ in $\Brace{\mathcal{L} \cup H} - e$. 
\end{definition}

\begin{observation}
	\label{obs:H-minimal-implies-weakly-minimal}
	Let $H$ be a subgraph of a digraph $D$ and let $\mathcal{L}$ be a linkage which is $H$-minimal.
	Then $\mathcal{L}$ is weakly $\Abs{\mathcal{L}}$-minimal with respect to $H$.
\end{observation}
\begin{proof}
	Assume towards a contradiction that there is some $L \in \mathcal{L}$ and some $e \in \E{L} \setminus \E{H}$ such that $L$ can be decomposed into $L_1 \cdot e \cdot L_2$ and there is no $\V{L_1}$-$\V{L_2}$-separator of size less than $\Abs{\mathcal{L}}$ in $\ToDigraph{\mathcal{L} \cup H} - e$.
	By~\cref{thm:menger}, there is a $\V{L_1}$-$\V{L_2}$-linkage $\mathcal{Q}$ of order $\Abs{\mathcal{L}}$ in $\ToDigraph{\mathcal{L} \cup H} - e$.
	
	Let $S$ be a minimum $\Start{\mathcal{L}}$-$\End{\mathcal{L}}$-separator in $\ToDigraph{\mathcal{L} \cup H} - e$.
	Because $\mathcal{L}$ is $H$-minimal, we have that $\Abs{S} < \Abs{\mathcal{L}}$.
	Hence, $S$ must hit every path in $\mathcal{L} \setminus \Set{L}$ and must be disjoint from $L$.
	
	Since $\Abs{\mathcal{Q}} = \Abs{\mathcal{L}}$, there is some $Q \in \mathcal{Q}$ which is not hit by $S$.
	Hence, there is a $\Start{L}$-$\End{L}$-path in $\ToDigraph{\mathcal{L} \cup H} - e - S$, a contradiction to the assumption that $S$ is a separator. 
	Thus, $\mathcal{L}$ is weakly $\Abs{\mathcal{L}}$-minimal with respect to $H$.
\end{proof}

We close this part by recalling the definition of well-linkedness, an important property of a central concept in our proof, the path of well-linked sets.

\begin{definition}
    \label{def:well-linked}
	Let $A, B$ be sets of vertices in a digraph $D$.
    We say that \emph{$A$ is well-linked to $B$ in $D$} if for every $A' \subseteq A$ and every $B' \subseteq B$ with $\Abs{A'} = \Abs{B'}$ there is an $A'$-$B'$-linkage of order $\Abs{A'}$ in $D$.
	If \(A\) is well-linked to \(A\), then we say that \(A\) is a \emph{well-linked set}.
\end{definition}

\paragraph{Minors.}
Given a digraph $D$ and an arc $(u,v) \in \A{D}$, we say that $(u,v)$ is \emph{butterfly contractible} if $\Outdeg{u} = 1$ or $\Indeg{v} = 1$. 
The \emph{butterfly contraction} of $(u,v)$ is the operation which consists in removing $u$ and $v$ from $D$, then adding a new vertex $uv$, together with the arcs $\Set{(w, uv) \mid w \in \InN{D}{u}}$ and $\Set{(uv, w) \mid w \in \OutN{D}{v}}$. 
Note that, by definition of digraphs, we \emph{remove} duplicated arcs and loops, that is, arcs of the form $(w,w)$.
If there is a subgraph $D'$ of $D$ such that we can construct another digraph $H$ from $D'$ using butterfly contractions, then we say that $H$ is a \emph{butterfly minor of $D$}, or that \emph{$D$ contains $H$ as a butterfly minor}.

\begin{definition}
	\label{def:butterfly-model}
    Let $H$ and $D$ be directed graphs.
    A \textbf{butterfly-model} of \(H\) in \(D\) is a function \(\mu\) which assigns to every \(x \in V(H) \cup \A{H}\) a subdigraph of \(D\)
		such that:
    \begin{enumerate}
        \item for every pair of distinct vertices \(u, v \in V(H)\), \(\mu(u)\) and \(\mu(v)\) are vertex-disjoint,
        \item for a vertex \(v \in V(H)\) and a non-incident edge \(e \in \A{H}\), \(\mu(v)\) and \(\mu(e)\) are vertex-disjoint,
                \item for every \(v \in V(H)\), \(\mu(v)\) is the union of an in-tree and an out-tree intersecting exactly on their common root, and
        \item for every \((u, v) \in \A{H}\), \(\mu((u, v))\) is a directed path starting at a vertex of the out-tree of \(\mu(u)\) and ending at a vertex of the in-tree of \(\mu(v)\).
    \end{enumerate}
\end{definition}

Given a digraph \(D\) and an arc \((u,v) \in \A{D}\), a \emph{subdivision} of \((u,v)\) is the operation consisting in removing \((u,v)\), adding a new vertex \(w\) and two new arcs
\((u,w), (w,u)\) to \(D\).
A subdivision of \(D\) is any digraph obtained after subdividing any number of arcs of \(D\).
We say that a digraph \(H\) is a \emph{topological minor} of \(D\) if some subdivision of \(H\) is contained as a subgraph in \(D\).

The following statement is well known.
We include a proof here for completeness.

\begin{observation}
	\label{statement:degree-3-topological-butterfly-minor}
	Let \(D, H\) be digraphs.
	If \(\Indeg{u} \leq 2\), \(\Outdeg{u} \leq 2\) and \(\Indeg{v} + \Outdeg{v} \leq 3\) holds for every vertex \(v\) of \(H\)
	and \(D\) contains \(H\) as a butterfly minor,
	then \(D\) also contains \(H\) as a topological minor.
\end{observation}
\begin{proof}
	Let \(u_{1}, u_{2}, \ldots, u_{n}\) be the vertices of \(H\).
	Let \(\mu\) be a butterfly model of \(H\) in \(D\).
	We find distinct vertices \(v_{1}, v_{2}, \ldots, v_{n}\) in \(D\) as follows.

	If \(u_i\) has in-degree two, then it must have out-degree at most one.
	Let \(w_1, w_2\) be the two in-neighbours of \(u_i\) in \(H\),
	let \(L_1 = \mu((w_1, u_i))\) and
	let \(L_2 = \mu((w_2, u_i))\).
	Let \(c_i\) be a vertex of \(\mu(u_i)\)
	which is reachable from all sources of \(\mu(u_i)\) and
	can reach all sinks of \(\mu(u_i)\).
	
	By definition of \(\mu\), there are two paths \(P_1, P_2\) from
	\(\End{L_1}\) and \(\End{L_2}\) to \(c_i\), respectively.
	Choose \(v_i\) as the first vertex in the intersection between \(P_1\) and \(P_2\).

	If \(u_i\) has out-degree two, we proceed in an analogous fashion as above,
	taking \(w_1, w_2\) as the out-neighbours of \(u_i\) instead.

	By definition of \(\mu\), there are two paths \(P_1, P_2\) from
	\(c_i\) to \(\Start{L_1}\) and \(\Start{L_2}\), respectively.
	Choose \(v_i\) as the last vertex in the intersection between \(P_1\) and \(P_2\).

	If \(u_i\) has both in-degree and out-degree at most one, we choose \(v_i = c_i\).

	We now construct a subdivision \(H'\) of \(H\).
	For every arc \((u_i, u_j) ∈ \A{H}\),
	let \(P_{i,j}\) be a path from \(v_i\) to \(v_j\) inside the union of
	\(\mu(u_i)\), \(\mu(u_j)\) and \(\mu((u_i, u_j))\)
	containing \(\mu((u_i, u_j))\) as a subpath.
	We subdivide the arc \((u_i, u_j)\) \(\Abs{\V{P_{i,j}}} - 2\) many times.
	By choice of \(v_i\) and because \(u_i\) has in-degree or out-degree at most 2, any two distinct paths \(P_{i,j}, P_{x,y}\) are internally disjoint.
	Hence, \(D\) contains \(H'\) as a subdigraph, and, thus, \(H\) as a topological minor.
\end{proof}

\subsection{Directed \treewidth and cylindrical grids}
\label{sec:grids}

In this \namecref{sec:grids}, we recall the definition of directed \treewidth and the concepts of webs, cylindrical grids and cylindrical walls.

Directed \treewidth was originally introduced by Reed~\cite{reed1999introducing} and by Johnson, Robertson, Seymour and Thomas~\cite{johnson2001directed} (see also~\cite{JohnsonRST2001}).
Adler~\cite{adler2007directed} showed that the original definition in~\cite{johnson2001directed} of directed \treewidth is not closed under butterfly minors.
We, therefore, use the variant of directed \treewidth defined in \cite{kawarabayashi2022directed}, which is closed under taking butterfly minors~\cite{butterflyminors2025}.

However, we only include this definition for the sake of completeness, as our proofs do not immediately use it.

An \emph{arborescence} $T$ is an acyclic directed graph obtained from an undirected rooted tree by orienting all edges away from the root.
That is, $T$ has a vertex $r_0$, called the root of $T$, with the property that for every  $r \in V(T)$ there is a unique directed path from $r_0$ to $r$ in $T$.
For each $r \in V(T)$, we denote the subarborescence of $T$ induced by the set of vertices in $T$ reachable from $r$ by $T_r$.
In particular, $r$ is the root of $T_r$.

\begin{definition}[{\cite[Definition 3.1]{kawarabayashi2022directed}}]
	\label{def:directed-tree-width}
	A \emph{directed \treedecomposition} of a digraph $D$ is a triple $(T, \beta, \gamma)$, where $\beta: V(T) \to 2^{V(D)}$ and $\gamma: E(T) \to 2^{V(D)}$ are functions and $T$ is an arborescence such that
	\begin{enamerate}{T}{item:directed-tree-width:last}
		\item \label{item:directed-tree-width:beta}
		      $\Set{\beta(t) : t \in V(T)}$ is a partition of $V(D)$ into (possibly empty) sets and
		\item \label{item:directed-tree-width:guard}
		      for every $e = (s, t) \in E(T)$, there is no closed directed walk in $D - \gamma(e)$ containing a vertex in $A$ and a vertex in $B$, where $A = \bigcup\Set{\beta(t ): t \in V(T_t)}$ and $B = V(D) \setminus A$.
		\label{item:directed-tree-width:last}
	\end{enamerate}
	For $t \in V(T)$ we define $\Gamma(t) \coloneqq \beta(t) \cup \bigcup\Set{\gamma(e) : e \sim t}$, where $e \sim t$ if $e$ is incident to $t$, and we define $\beta(T_t) \coloneqq \bigcup\Set{\beta(t) : t \in V(T_t)}$.
	The \emph{width} of $(T, \beta, \gamma)$ is the smallest integer $w$ such that $\Abs{\Gamma(t)} \leq w + 1$ for all $t \in V(T)$.
	The \emph{directed \treewidth} of $D$ is the smallest integer $w$ such that $D$ has a directed \treedecomposition of width $w$.
	The sets $\beta(t)$ are called the bags and the sets $\gamma(e)$ are called the guards of the directed \treedecomposition.
\end{definition}

We now define another obstruction to directed \treewidth called \emph{cylindrical grids},
which are illustrated in \cref{fig:cylindrical-grid}.
Since we are interested in grids in the context of butterfly minors, we define grids by linkages instead of giving explicit vertex and arc sets. 
\begin{definition} \label{def:cylindrical-grid}
	A \emph{cylindrical grid} of order $k$ is a digraph $G_k$ consisting of $k$ pairwise disjoint directed cycles $C_1, C_2, \dots, C_{k}$ of length $2k$, together with a set of $2k$ pairwise vertex-disjoint paths $P_1, P_2, \dots, P_{2k}$ of length $k - 1$ such that 
	\begin{itemize}
		\item each path $P_i$ has exactly one vertex in common with each cycle $C_j$ and both endpoints of $P_i$ are in $\V{C_1} \cup \V{C_k}$, 
		\item the paths $P_1, P_2, \dots, P_{2k}$ appear on each $C_i$ in this order, and
		\item for each $1 \leq i \leq 2k$, if $i$ is odd, then the cycles $C_1, C_2, \dots, C_{k}$ occur on $P_i$ in this order and, if $i$ is even, then the cycles occur in the reverse order $C_k, C_{k-1}, \dots, C_{1}$. 
	\end{itemize}
\end{definition}

\begin{figure}[!ht]
	\centering
	\begin{minipage}{0.4\textwidth}
		\centering
		\includegraphics{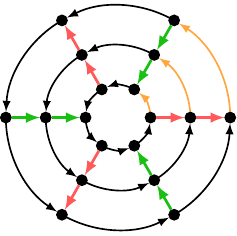}
	\end{minipage}
	\begin{minipage}{0.4\textwidth}
		\centering
		\includegraphics{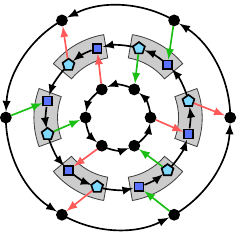}
	\end{minipage}
	\caption{A cylindrical grid of order 3 on the left and the cylindrical wall of order 3 on the right.
	The split vertices on the wall have a grey background,
	where the cyan pentagon is the vertex \(v^{\text{out}}\) and
	the blue rectangle is the vertex \(v^{\text{in}}\).}
	\label{fig:cylindrical-grid}
\end{figure}

Working with butterfly minors can sometimes require technical proofs due to the complexity of butterfly minor models.
For this reason, stating our results in terms of topological minors is also desirable.
To do so, we need a slightly different definition of a cylindrical grid, called a \emph{cylindrical wall}.
In the context of topological minors, it is more convenient to define the cylindrical wall of order \(k\) as a single digraph and not as a family of digraphs.

A \emph{cylindrical wall of order \(k\)} is obtained from the cylindrical grid of order \(k\) with a minimum number of vertices by splitting every vertex of in-degree and out-degree two in the following way.
Replace a vertex \(v\) by two vertices \(v^{\text{in}}, v^{\text{out}}\), add the arc \((v^{\text{in}}, v^{\text{out}})\), replace every arc \((u, v)\) with the arc \((u, v^{\text{in}})\) and every arc \((v, u)\) with the arc \((v^{\text{out}}, u)\).

It is a simple exercise to see that every cylindrical grid of order \(2k - 2\)
contains the cylindrical wall of order \(k\).

\begin{figure}[H]
		\centering
		\includegraphics{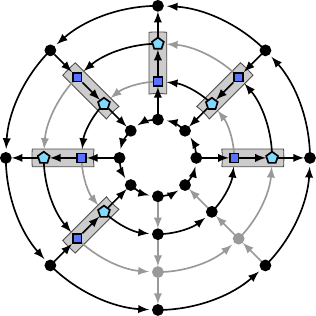}
		\caption{How a subdivision of the cylindrical wall of order 3 can be constructed inside a cylindrical grid of order 4.
		Unused arcs and vertices of the grid are coloured grey.
		The split vertices of the wall have a grey background,
		where the blue square vertex represents \(v^{\text{in}}\) and
		the cyan pentagon vertex represents \(v^{\text{out}}\).
			\label{fig:wall-inside-grid}}
\end{figure}

\begin{observation}
	\label{statement:grid-contains-wall}
	Every cylindrical grid of order \(2k - 2\) contains the cylindrical wall of order \(k\) as a topological minor.
\end{observation}
\begin{proof}
	Let \(P_{1}, P_{2}, \ldots, P_{4k - 4}\) be the paths and
	\(C_{1}, C_{2}, \ldots, C_{2k - 2}\) be the cycles of the given cylindrical grid
	as described in \cref{def:cylindrical-grid}.
	We construct a subdivision of the cylindrical wall of order \(k\) as follows.

	Define \(C'_1 = C_1\) and \(C'_k = C_{2k - 2}\).
	For each \(1 \leq j \leq k\),
	let \(v^{\text{out}}_{1,2j - 1}\) be the vertex of \(P_{2j - 1}\) in \(C_1\),
	let \(v^{\text{in}}_{1,2j}\)      be the vertex of \(P_{2j}\)     in \(C_1\),
	let \(v^{\text{in}}_{k,2j - 1}\)  be the vertex of \(P_{2j - 1}\) in \(C_{2k - 2}\) and
	let \(v^{\text{out}}_{k,2j}\)     be the vertex of \(P_{2j}\)     in \(C_{2k - 2}\).

	For each \(2 \leq i \leq k - 1\), construct \(C'_i\) as follows.
	For each \(1 \leq j \leq 2k\),
	let \(\variablestyle{a} = \variablestyle{in}, \variablestyle{b} = \variablestyle{out}\) if \(j\) is odd and
	let \(\variablestyle{a} = \variablestyle{out}, \variablestyle{b} = \variablestyle{in}\) if \(j\) is even.
	Let \(v^{\variablestyle{a}}_{i,j}\) be the vertex of \(P_j\) on \(C_{2(i - 1)}\) and
	let \(v^{\variablestyle{b}}_{i,j}\) be the vertex of \(P_j\) on \(C_{2(i - 1) + 1}\).
	Let \(L_{i,j}\) be the
	\(v^{\variablestyle{out}}_{i,j}\)-\(v^{\variablestyle{in}}_{i,j+1}\)-subpath
	of \(C_{2(i - 1)}\) (if \(j\) is even) or
	\(C_{2(i - 1) + 1}\) (if \(j\) is odd).
	Let \(R_i\) be the \(v^{\variablestyle{out}}_{i,2k}\)-\(v^{\variablestyle{in}}_{i,1}\)-subpath
	of \(C_{2(i-1)}\).
	Define \(C'_i = L_{i,1} \cdot L_{i,2} \cdot \ldots \cdot L_{i,2k} \cdot R_{i}\).

	For each \(1 \leq i \leq k - 1\) and
	each \(1 \leq j \leq 2k\),
	define the paths \(A_{i,j}, B_{i,j}\) as follows.
	\(B_{i,j}\) is the \(v^{\variablestyle{in}}_{i,j}\)-\(v^{\variablestyle{out}}_{i,j}\)-subpath of \(P_j\).
	If \(j\) is even,
	define \(A_{i,j}\) as the
	\(v^{\variablestyle{out}}_{i,j}\)-\(v^{\variablestyle{in}}_{i+1,j}\)-subpath of \(P_j\).
	If \(j\) is odd,
	define \(A_{i,j}\) as the
	\(v^{\variablestyle{out}}_{i+1,j}\)-\(v^{\variablestyle{in}}_{i,j}\)-subpath of \(P_j\).

	One can confirm that the union of all \(C'_i, A_{i,j}, B_{i,j}\) defined above
	is a subdivision of the cylindrical wall of order \(k\), completing the proof.
\end{proof}

An essential difference between cylindrical grids and grids in undirected graphs is that cylindrical grids are \emph{locally acyclic} in the following sense. 
Suppose we delete in each cycle $C_i$ the arc $e_i$ whose head is on the path $P_1$.
These arcs are drawn in \textcolor{myOrange}{orange} in~\cref{fig:cylindrical-grid}.
The resulting digraph is acyclic and consists of two linkages: the linkage  $\{P_1, \dots, P_{2k}\}$ and the linkage $\{C_1 - e_1, \dots, C_k - e_k\}$ which contains for each cycle $C_i$ the path that remains once the arc $e_i$ is deleted.
Digraphs of this form are called \emph{fences}.
See~\cref{fig:cyl-grid-from-fence} for a drawing of cylindrical grids illustrating how they are constructed from a fence with additional arcs closing the cycles.

\begin{definition}
	\label{def:fence}
	A \emph{$(p,q)$-fence} is a tuple $(\mathcal{P}, \mathcal{Q})$ such that
	\begin{itemize}
		\item \label{fence:P}
		$\mathcal{P} = \Brace{P_1, P_2, \dots, P_{2p}}$ and $\mathcal{Q} = \Brace{Q_1, Q_2, \dots, Q_{q}}$ are linkages,
		\item \label{fence:P-intersection-Q}
		for each $1 \leq i \leq 2p$ and each $1 \leq j \leq q$, the digraph $P_i \cap Q_j$ is a path (and therefore non-empty),
		\item \label{fence:P-order-Q}
		for each $1 \leq j \leq q$, the paths $P_1 \cap Q_j, P_2 \cap Q_j, \dots, P_{2p} \cap Q_j$ appear in this order along $Q_j$, and
		\item \label{fence:Q-order-P}
		for each $1 \leq i \leq 2p$, if $i$ is odd then the paths $P_i \cap Q_1, P_i \cap Q_2, \dots, P_i \cap Q_{q}$ appear in this order along $P_i$, and if $i$ is even instead, then the paths $P_i \cap Q_q, P_i \cap Q_{q-1}, \dots, P_i \cap Q_{1}$ appear in this order along $P_i$.
	\end{itemize}
\end{definition}

See~\cref{fig:fence-1} for an illustration.
The \say{horizontal} paths, or \emph{rows}, constitute the linkage $\QQQ$ and the \emph{columns} form the linkage $\PPP$.

	\begin{figure}[ht]
		\centering
		\begin{subfigure}[b]{0.49\textwidth}
			\centering
			\includegraphics{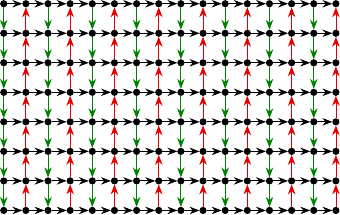}
			\caption{An $(8,8)$-fence.}
			\label{fig:fence-1}
		\end{subfigure}
		\hfill
		\begin{subfigure}[b]{0.49\textwidth}
			\centering
			\includegraphics{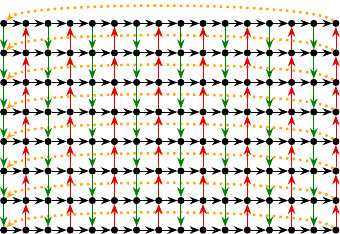}
			\caption{A cylindrical grid of order \(8\).}
			\label{fig:cyl-grid-from-fence}
		\end{subfigure}
		\caption{A fence and, on the right, an illustration of how a cylindrical grid is obtained from a fence.
			The \textcolor{myOrange}{dotted orange} paths symbolise the arcs $e_i$ that close the cycles drawn solid in~\cref{fig:cylindrical-grid}.}
		\label{fig:acyclic-grid-and-fence}
	\end{figure}

Further decomposing the fence constructed from the cylindrical grid yields an even simpler form of directed grid.
In a fence, we can only route from \say{left to right} and we can route \say{upwards} as well as \say{downwards}. 
An even simpler form of a directed grid is obtained if we remove the \say{upwards} paths from a fence, that is, every second column.
The resulting digraph is called an \emph{acyclic grid}, illustrated in~\cref{fig:acyclic-grid}.

\begin{figure}[!ht]
	\centering
	\begin{subfigure}[b]{0.49\textwidth}
        \centering
		\includegraphics{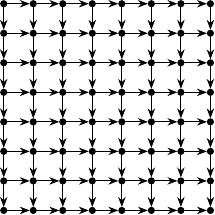}
        \caption{An acyclic $(8,8)$-grid.}
        \label{fig:acyclic-grid}
    \end{subfigure}
    \hfill
    \begin{subfigure}[b]{0.49\textwidth}
        \centering
        \includegraphics{fig-fence.pdf}
        \caption{An $(8,8)$-fence.}
        \label{fig:fence}
    \end{subfigure}
	\caption{An acyclic grid and a fence.}
	\label{fig:acyclic-grid-and-fence}
\end{figure}

\begin{definition}
	\label{def:acyclic-grid}
	An \emph{acyclic $(p,q)$-grid} is a pair $(\mathcal{P}, \mathcal{Q})$ such that
	\begin{itemize}
		\item \label{acyclic-grid:P}
		      $\mathcal{P} = \Brace{P_1, P_2, \dots, P_{p}}$ and $\mathcal{Q} = \Brace{Q_1, Q_2, \dots, Q_{q}}$ are linkages,
		\item \label{acyclic-grid:P-intersection-Q}
		      for each $1 \leq i \leq p$ and each $1 \leq j \leq q$, the digraph $P_i \cap Q_j$ is a path (and therefore non-empty),
		\item \label{acyclic-grid:P-order-Q}
		      for each $1 \leq j \leq q$, the paths $P_1 \cap Q_j, P_2 \cap Q_j, \dots, P_{p} \cap Q_j$ appear in this order along $Q_j$, and
		\item \label{acyclic-grid:Q-order-P}
		      for each $1 \leq i \leq p$, the paths $P_i \cap Q_1, P_i \cap Q_2, \dots, P_i \cap Q_{q}$ appear in this order along $P_i$.
	\end{itemize}
\end{definition}

Another grid-like structure we define is the \emph{web}, initially introduced by Reed et al.~in~\cite{reed1996packing}.
Webs are essential in the proof of the Directed Grid Theorem in~\cite{kawarabayashi2015directed} and are equally crucial for our results.

\begin{definition}\label{def:web}
	Let $D$ be a digraph.
	Two linkages $\mathcal{H}$ and $\mathcal{V}$ in $D$ constitute an $\Brace{\Abs{\mathcal{H}},\Abs{\mathcal{V}}}$-web $\Brace{\mathcal{H},\mathcal{V}}$ if every path in $\mathcal{V}$ intersects every path in $\mathcal{H}$.
	
	The set $\Start{\mathcal{V}}$ is called the \emph{top} of the web, while the set $\End{\mathcal{V}}$ is called the \emph{bottom} of the web.
	Finally, $\Brace{\mathcal{H},\mathcal{V}}$ is \emph{well-linked} if $\End{\mathcal{V}}$ is well-linked linked to $\Start{\mathcal{V}}$ in $D$.
\end{definition}

Next to these webs, we use the following two more restrictive definitions.

\begin{definition}
	\label{def:ordered_web}
	Let $\Brace{\mathcal{H}, \mathcal{V}}$ be an $\Brace{h,v}$-web.
	We say that $\Brace{\mathcal{H}, \mathcal{V}}$ is an \emph{ordered web} if there is an ordering of $\mathcal{V} = \Brace{V_1, V_2, \dots, V_{v}}$ for which each path $H \in \mathcal{H}$ can be decomposed into $H = H_1 \cdot H_2 \cdot \ldots \cdot H_{v}$ such that $H_i$ intersects $V_j$ if and only if $i = j$.
\end{definition}
	
\begin{definition}
	\label{def:folded-web}
	An $(h,v)$-web $\Brace{\mathcal{H}, \mathcal{V}}$ is a \emph{folded web} if every $V_i \in \mathcal{V}$ can be split as $V_i^a \cdot V_i^b \coloneqq V_i$ such that both $V_i^a$ and $V_i^b$ intersect all paths of $\mathcal{H}$.
\end{definition}

\section{Path of sets and temporal digraphs}
\label{sec:powls-framework}

In~\cite{COSSI}, we introduced the concept of paths of well-linked and order-linked sets.
Here, we repeat some of the definitions and statements we use.
First, we need a concept similar to well-linkedness, but which describes the connectivity given by an acyclic grid.
This property is called \emph{order-linkedness}, defined below.

\textbf{Shifts and order-linkedness.}
Let $A = (a_1,\dots, a_n)$ and $B = (b_1,\dots, b_m)$ be ordered sets.
Let $r \in \N,$ let $A'$ be an ordered subset of $A$ and let $B'$ be an ordered subset of $B$ such that $\Abs{A'} = \Abs{B'}.$
We say that $B'$ is an \emph{$r$-shift of $A'$} if there is a bijection $\pi : A' \rightarrow B'$ such that 
\begin{itemize}[noitemsep,topsep=0pt]
	\item for all $a_i \in A'$ we have that $\pi(a_i) = b_j$ implies $i \leq j$;
	\item there are at most $r$ vertices $a_i \in A'$ with $\pi(a_i) \neq b_i$; and
	\item for all $a_i, a_j \in A',$ if $a_i \leq_A a_j,$ then $\pi(a_i) \leq_B \pi(a_j).$
\end{itemize}

Let $H$ be a digraph, $A = \Brace{a_1, \dots, a_n}, B = \Brace{b_1, \dots, b_m} \subseteq \V{H}$ be ordered sets and let $r \in \N.$
We say that $A$ is $r$-\emph{order-linked} to $B$ in $H$ if for every $A' \subseteq A$ and every $B' \subseteq B$ with $\Abs{A'} = \Abs{B'}$ where $B'$ is an $r$-shift of $A'$ witnessed by the bijection $\pi$ there is an $A'$-$B'$-linkage $\mathcal{L}$ in $H$ satisfying $\Fkt{\pi}{a} = \Fkt{\mathcal{L}}{a}$ for all $a \in A'.$

\newcommand*{\powls}{path of well-linked sets\xspace}
\newcommand*{\powlss}{paths of well-linked sets\xspace}
\newcommand*{\pools}[1]{	\ifx\relax #1\relax
	path of order-linked sets\xspace
	\else
	path of $#1$-order-linked sets\xspace
	\fi}
\newcommand*{\poolss}[1]{	\ifx\relax #1\relax
	paths of order-linked sets\xspace
	\else
	paths of $#1$-order-linked sets\xspace
	\fi}
\newcommand*{\cowls}{cycle of well-linked sets\xspace}
\newcommand*{\cowlss}{cycles of well-linked sets\xspace}

We now define two central objects in our proofs: the paths of well-linked sets and the paths of order-linked sets.

\newcommand{\OLtext}[1]{\textcolor{myGreen}{#1}}
\newcommand{\WLtext}[1]{\textcolor{myLightBlue}{#1}}
\begin{definition}[path of \OLtext{$r$-order-linked}/\WLtext{well-linked} sets~\cite{COSSI}]
	\label{def:path-of-order-linked-sets}
	\label{def:path-of-well-linked-sets}
	A \emph{path of \OLtext{$r$-order-linked}\slash \WLtext{well-linked} sets} of width $w$ and length $\ell$ is a tuple $\Brace{\mathcal{S},\mathscr{P}}$ such that
	\begin{itemize}
		\item $\mathcal{S}$ is a sequence of $\ell + 1$ pairwise disjoint subgraphs $\Brace{S_0,\dots,S_{\ell}},$ which are called \emph{clusters},
		
		\item for every $0 \leq i \leq \ell$ there are disjoint \OLtext{ordered} sets $A(S_i),B(S_i) \subseteq \V{S_i}$ of size $w$ such that $A(S_i)$ is \OLtext{$r$-order-linked}\slash \WLtext{well-linked} to $B(S_i)$ in $S_i,$
		
		\item $\mathscr{P}$ is a sequence of $\ell$ pairwise disjoint linkages $\Brace{\mathcal{P}_0, \mathcal{P}_1, \dots, \mathcal{P}_{\ell - 1}}$ such that, for every $0 \leq i < \ell,$
		$\mathcal{P}_i$ is a $B(S_i)$-$A(S_{i+1})$-linkage of order $w$ which is internally disjoint from $S_i$ and $S_{i+1}$ and disjoint from every $S \in \mathcal{S} \setminus \Set{S_i, S_{i+1}}.$
	\end{itemize}
	
	Further, a \OLtext{\pools{r}} $\Brace{\mathcal{S}, \mathscr{P}}$ is called \emph{uniform} if for all $0\leq i < \ell$ and for all $b_1,b_2 \in B(S_i)$ we have that $b_1 \leq_{B(S_i)} b_2$ implies $\mathcal{P}_i(b_1) \leq_{A(S_{i+1})} \mathcal{P}_i(b_2).$
	A \WLtext{\powls} is called \emph{strict} if every vertex in $S_i$ lies on an $A(S_i)$-$B(S_i)$-path.
\end{definition}

Just as we can find a large fence inside a large acyclic grid, we can also construct a path of well-linked sets from a path of order-linked sets.

\begin{lemma}[{\cite[Lemma 8.3]{COSSI}}]
	\label{proposition:order-linked to path of well-linked sets}
	\label[prop-abbr]{prop-abbr:order-linked to path of well-linked sets}
	Let $\bound{proposition:order-linked to path of well-linked sets}{w}{w, \ell} \coloneqq w(\ell + 1)$.
	\boundDef{prop-abbr:order-linked to path of well-linked sets}{w}{w, \ell} 
	Every path of $w$-order-linked sets $(\mathcal{S} = \Brace{S_0, S_1, \dots, S_{\ell}}$, $\mathscr{P} = \Brace{\mathcal{P}_0, \mathcal{P}_1, \dots, \mathcal{P}_{\ell - 1}})$ of width at least $\bound{proposition:order-linked to path of well-linked sets}{w}{w, \ell}$ and length at least $\ell$ contains a \powls $\Brace{\mathcal{S}' = \Brace{S'_0, S'_1, \dots, S'_{\ell}}, \mathscr{P}' = \Brace{\mathcal{P}'_0, \mathcal{P}'_1, \dots, \mathcal{P}'_{\ell - 1}}}$ of width $w$ and length $\ell$.
	Further, for every $0 \leq i \leq \ell$ we have $A(S_i') \subseteq A(S_i)$, $B(S_i') \subseteq B(S_i)$, $S_i' \subseteq S_i$ and for every $0 \leq i < \ell$ we have $\mathcal{P}_i' \subseteq \mathcal{P}_i$.
\end{lemma}

One can find a \powls in any large enough ordered web.

\begin{corollary}[{\cite[Corollary 9.3]{COSSI}}]
	\label{lemma:ordered-web-to-path-of-well-linked-sets-with-side-linkage}
	There exist two functions
	\boundDef{lemma:ordered-web-to-path-of-well-linked-sets-with-side-linkage}{h}{w, \ell}
    \(\bound{lemma:ordered-web-to-path-of-well-linked-sets-with-side-linkage}{h}{w, \ell} \in \Oh(w^{2} \ell^{2})\), and
	\boundDef{lemma:ordered-web-to-path-of-well-linked-sets-with-side-linkage}{v}{w, \ell}
	\(\bound{lemma:ordered-web-to-path-of-well-linked-sets-with-side-linkage}{v}{w, \ell} \in \PowerTower{1}{\Polynomial{25}{w, \ell}}\)
    such that for every ordered $(h,v)$-web $\Brace{\mathcal{H},\mathcal{V}}$ with $h \geq \bound{lemma:ordered-web-to-path-of-well-linked-sets-with-side-linkage}{h}{w, \ell}$ and $v \geq \bound{lemma:ordered-web-to-path-of-well-linked-sets-with-side-linkage}{v}{w, \ell}$ there is a path of well-linked sets $\Brace{\mathcal{S} = \Brace{ S_0, S_1, \dots, S_{\ell} }, \mathscr{P}}$ of width $w$ and length $\ell$ in $\ToDigraph{\mathcal{H} \cup \mathcal{V}}$ such that $B(S_{\ell}) \subseteq \End{\mathcal{H}}$.
	Additionally, there is a linkage $\mathcal{X} \subseteq \mathcal{V}$ of order $\ell + 1$ and a bijection $\pi : \mathcal{S} \rightarrow \mathcal{X}$ such that $A(S_i) \subseteq \V{\pi(S_i)}$ and $\V{\pi(S_i)} \cap \V{(\mathcal{S}, \mathscr{P})} \subseteq \V{S_i}$ for each $0 \leq i \leq \ell$.
\end{corollary}

While 1-order-linked sets offer very little connectivity,
one can increase the order of linkedness of a \pools{} by merging several clusters. 
\begin{lemma}[{\cite[Lemma 7.7]{COSSI}}]
	\label{lemma:increase-order-linkedness}
	Let $r,c,w$ be integers.
	Let $D = \Brace{\mathcal{S} = \Brace{S_0, S_1, \dots, S_{\ell}}, \mathscr{P}}$ be a uniform path of $r$-order-linked sets of width $w$ and length at least $c-1$. 
	Then $A(S_0)$ is $cr$-order-linked to $B(S_\ell)$ in $D$.
\end{lemma}

\begin{theorem}[{\cite[Theorem 7.8]{COSSI}}]
	\label{lemma:merging path of order-linked sets}
	Every uniform path of $r$-order-linked sets \(D = (\mathcal{S} = (S_0, S_1, \dots,\allowbreak S_{\ell}), \mathscr{P} = (\mathcal{P}_0, \mathcal{P}_1, \dots, \mathcal{P}_{\ell - 1}))\) of length at least $c\ell$ and width $w$ contains a uniform path of $cr$-order-linked sets \(\Brace{\mathcal{S}' = \Brace{ S'_0, S'_1, \dots, S'_{\ell}}, \mathscr{P}' = \Brace{ \mathcal{P}'_0, \mathcal{P}'_1, \dots, \mathcal{P}'_{\ell - 1}}}\) of length $\ell$ and width $w$. Additionally, for every $0 \leq i \leq \ell$ we have $S_i' \subseteq \SubPOSS{D}{ci}{c(i+1)-1}$, $A(S_i') \subseteq A(S_{ci})$ and $B(S_i') \subseteq B(S_{c(i + 1) - 1})$, and for $0 \leq i < \ell$ we have $\mathcal{P}'_i \subseteq \mathcal{P}_{(c-1)(i+1)}$.
\end{theorem}

One of the main results of \cite{COSSI} is that every digraph of large enough directed \treewidth does contain a \powls that has the additional property of the $B$-set in its last cluster being well-linked to the $A$-set in the first cluster.

\thmHighDTWToPOWL*

We often need to find linkages inside paths of well-linked or order-linked sets.
We collect in the two lemmas and the observation below all the different cases we need,
proving that the required linkages exist.

\begin{lemma}[{\cite[Lemma 8.8]{COSSI}}]
	\label{lem:linkage_inside_poss}
	Let \(\Brace{\mathcal{S} = \Brace{S_0, S_1, \dots, S_{\ell}}, \mathscr{P} = \Brace{\mathcal{P}_0, \mathcal{P}_1, \dots, \mathcal{P}_{\ell - 1}}}\) be a path of well-linked sets of width  $w$ and length $\ell$.
	Let $X,Y \subseteq \V{(\mathcal{S}, \mathscr{P})}$ such that $\Abs{X} = \Abs{Y} = k$.
	Let $f : X \cup Y \rightarrow \N$ be a function such that $v \in S_{\Fkt{f}{v}} \cup \mathcal{P}_{\Fkt{f}{v}}$ for all $v \in X \cup Y$.
	There is an $X$-$Y$-linkage $\mathcal{L}$ in $(\mathcal{S}, \mathscr{P})$ if $\Fkt{f}{x} \leq \Fkt{f}{y} - 2$ for all $x \in X$ and all $y \in Y$ and at least one of the following is true:
	\renewcommand{\labelenumi}{\textbf{\theenumi}} 
	\renewcommand{\theenumi}{(L\arabic{enumi})}
	\begin{enumerate}[labelindent=0pt,labelwidth=\widthof{\ref{last-item-linkage-inside-POSS}},itemindent=1em]
		\item there are $0 \leq i < j \leq \ell$ such that $X \subseteq B(S_{i})$ and $Y \subseteq A(S_{j})$,
		\label{case:linkage_inside_poss_B_A}
		\item $\Abs{\Fkt{f}{x_1} - \Fkt{f}{x_2}} \geq 2$ for all $x_1, x_2 \in X$ with $x_1 \neq x_2$ and there is some $0 \leq i \leq \ell$ such that $Y \subseteq A(S_{i})$,
		      \label{case:linkage_inside_poss_scattered_A}
		\item $\Abs{\Fkt{f}{y_1} - \Fkt{f}{y_2}} \geq 2$ for all $y_1, y_2 \in Y$ with $y_1 \neq y_2$ and there is some $0 \leq i \leq \ell$ such that $X \subseteq B(S_{i})$, or
		      \label{case:linkage_inside_poss_B_scattered}
		\item $\Abs{\Fkt{f}{x_1} - \Fkt{f}{x_2}} \geq 2$ for all $x_1, x_2 \in X$ with $x_1 \neq x_2$ and $\Abs{\Fkt{f}{y_1} - \Fkt{f}{y_2}} \geq 2$ for all $y_1, y_2 \in Y$ with $y_1 \neq y_2$.
		      \label{case:linkage_inside_poss_scattered_scattered}
		      \label{last-item-linkage-inside-POSS}
	\end{enumerate}
	Furthermore, choose $i$ minimal with $S_i$ containing a vertex from $X$ and $j$ maximal with $S_j$ containing a vertex from $Y$.
    Then, $\mathcal{L}$ is contained inside $\SubPOSS{\Brace{\mathcal{S}, \mathscr{P}}}{i}{j}$.
\end{lemma}

\begin{observation}[{\cite[Observation 7.4]{COSSI}}]
	\label{lemma:linkage-inside-pools}
	Let $D = (\mathcal{S} = \Brace{S_0, S_1, \dots, S_{\ell}}, \mathscr{P})$ be a path of $0$-order-linked sets of width $w$. 
	For every $0 \leq i < j \leq \ell$, every $A' \in \Set{A(S_i), B(S_i)}$, and every $B' \in \Set{A(S_j), B(S_j)}$ there is an $A'$-$B'$-linkage $\mathcal{L}$ of order $w$ in $D$. 
	Furthermore, for all $i < k < j$ every path in $\mathcal{L}$ must intersect $A(S_k)$ and $B(S_k)$.
\end{observation}

\begin{lemma}[{\cite[Lemma 8.4]{COSSI}}]
	\label{lem:poss-simple-routing}
	Let $(\SSS \coloneqq (S_0, \dots, S_{\ell}), \PPPP \coloneqq (\mathcal{P}_0, \dots, \mathcal{P}_{\ell-1}))$ be a path of well-linked sets of width $w$ and length  $\ell$ and let $0 \leq i < j \leq \ell$.	Then for every $0 \leq i < j \leq \ell$, for every $A' \in \Set{A(S_i), B(S_i)}$ and for every $B' \in \Set{B(S_j), A(S_j)}$ we have that $A'$ is well-linked to $B'$ in $\SubPOSS{\Brace{\mathcal{S}, \mathscr{P}}}{i}{j}$. 
\end{lemma}

Finally, we sometimes need to restrict the size of a path of well-linked sets without losing width or length.
This is possible in the following case.

\begin{observation}[{\cite[Observation 8.9]{COSSI}}]
	\label{obs:restricting_width_poss}
	\label[lem-abbr]{lem-abbr:restricting_width_poss}
	Let $\Brace{\mathcal{S} = \Brace{S_0, S_1, \dots, S_{\ell}}, \mathscr{P} = \Brace{\mathcal{P}_0, \mathcal{P}_1, \dots, \mathcal{P}_{\ell - 1}}}$ be a path of well-linked sets of width at least $w$ and length $\ell$.
	Let $A_{0} \subseteq A(S_0)$ and $B_{\ell} \subseteq B(S_{\ell})$ with $\Abs{A(S_{0})} = \Abs{B_\ell} = w$.
	Then, $\Brace{\mathcal{S}, \mathscr{P}}$ contains a path of well-linked sets $(\mathcal{S}' = \Brace{S'_0, S'_1, \dots, S'_{\ell}}, \mathscr{P}' = \Brace{\mathcal{P}'_0, \mathcal{P}'_1, \dots, \mathcal{P}'_{\ell - 1}})$ of width $w$ and length $\ell$ such that $B(S_\ell') = B_\ell$, $A(S_0') = A_0$, $S'_i \subseteq S_i$ for all $0 \leq i \leq \ell$ and $\mathcal{P}'_i \subseteq \mathcal{P}_i$ for all $0 \leq i < \ell$.
\end{observation}

\subsection{Temporal digraphs and routings}
\label{sec:temporal_digraphs}

\cite{COSSI} introduced a framework based on temporal digraphs
in order to construct paths of order-linked and of well-linked sets.
This framework is particularly useful when we are working with some
linkage \(\mathcal{L}\) intersecting a sequence \(H_{1}, H_{2}, \ldots, H_{t}\) of disjoint subgraphs in order.
Since there might be no path from \(H_i\) to \(H_j\) if \(i > j\),
each \(H_i\) behaves somewhat like a layer of a temporal digraph,
allowing us to reroute some path going through \(\mathcal{L}\).

In order to better understand the connectivity given by the \(H_i\) described above,
we need very precise ways of analysing the connectivity available in a temporal digraph.
Towards this end, \cite{COSSI} defines the concept of \emph{\(H\)-routings}, which we recall here as well.

	A \emph{temporal digraph} is a pair $T = (V, \AAA)$ consisting of a vertex set $V$ and sequence of arc sets $\mathcal{A} = \Brace{A_1, A_2, \dots A_{\ell}}$ such that $\Layer{T}{t} \coloneqq \Brace{V, A_t}$ is a digraph for all $1 \leq t \leq \ell$. 
	We also refer to $\Layer{T}{t}$ as \emph{layer $t$} of $T$ and call $t$ a \emph{time step}.
	The \emph{lifetime} of $D$ is given by $\Lifetime{D} \coloneqq \ell$.
	
	\begin{definition}
		\label{def:h-routings}
		Let $H$ be a digraph, $D$ be a digraph or temporal digraph, and let $S\subseteq \V{D}$.
		An \emph{$H$-routing (over $S$)} is a bijection $\varphi : \V{H} \to S$ such that for each $u$-$v$-path $P$ in $H$ we can find a $\varphi(u)$-$\varphi(v)$-path (or temporal path, resp.) in $D$ which is disjoint from $S \setminus \Fkt{\varphi}{\V{P}}$. 
	\end{definition}

	We apply the framework on temporal digraphs in our setting by constructing a \emph{routing temporal digraph},
	defined below and illustrated in \cref{fig:routing_temporal_digraph}.
	
	\begin{definition}
		\label{def:routing-temporal-digraph}
		Let $\mathcal{P}$ be a linkage and let $\mathcal{Q} = \Set{Q_1, Q_2, \dots, Q_{q}}$ be a set of pairwise disjoint digraphs such that each path $P_i \in \mathcal{P}$ can be partitioned as $P_i^1 \cdot P_i^2 \cdot \ldots \cdot P_i^{q} = P_i$ such that  $\V{P_i^j} \cap \V{\mathcal{Q}} \subseteq \V{Q_j}$ for all $1 \leq j \leq q$.
		
		The \emph{routing temporal digraph $\Brace{V, \mathcal{A}}$ of $\mathcal{P}$ through $\mathcal{Q}$} is constructed as follows.
			We set $V = \mathcal{P}$ and for each $1 \leq j \leq q$ we define $A_j = \{(P_a, P_b) \mid P_a, P_b \in \mathcal{P}$ and there is a path from \(\V{P_a}\) to \(\V{P_b}\) inside \(Q_j\) which is internally disjoint from \(\mathcal{P}\)\(\}\).
	\end{definition}

	\begin{figure*}[ht!]
		\centering
		\resizebox{!}{8cm}{		\includegraphics{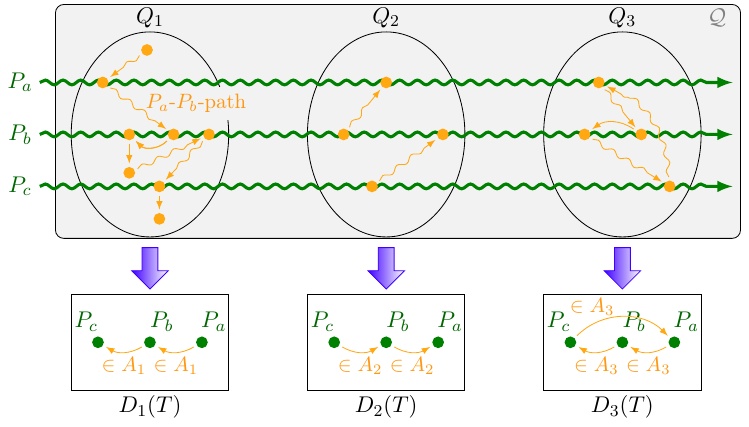}}
		\caption{The layers $D_j(T)$ of the temporal graph $T \coloneqq \Brace{V = \Set{a,b,c}, \mathcal{A} = \Set{A_1,A_2,A_3}}$ constructed from the graphs $Q_j$ displayed above as defined in~\cref{def:routing-temporal-digraph}.}
		\label{fig:routing_temporal_digraph}
	\end{figure*}

	A routing of a path in a routing temporal digraph can be used to obtain order-linkedness between two sets.

	\begin{lemma}[{\cite[Lemma 7.6]{COSSI}}]
		\label{lemma:P_k-routing-implies-1-order-linked}
		\label[lem-abbr]{lem-abbr:P_k-routing-implies-1-order-linked}
		Let $h,k$ be integers.
		Let $T$ be the routing temporal digraph of some linkage $\mathcal{L}$ through a sequence $\Brace{H_1, H_2, \dots, H_{h}}$ of disjoint digraphs. 
		Let $\mathcal{L}' \subseteq \mathcal{L}$ be a linkage of order at most $k$.
		If $T$ contains a $\Pk{k}$-routing on the paths $L_1, L_2, \dots, L_{k} \in \mathcal{L}'$, ordered according to their occurrence on the $\Pk{k}$-routing, then $A$ is $1$-order-linked to $B$ in $\ToDigraph{\mathcal{L} \cup \bigcup_{i=1}^{h} H_i}$, where $A = \Set{a_i \mid a_i \text{is the first vertex of \(L_i\) on \(H_1\)}}$ and $B = \Set{b_i \mid b_i \text{ is the last vertex of \(L_i\) on \(H_h\)}}$. 
	\end{lemma}

	As every temporal digraph with enough unilateral layers contains the routing of a path,	we can use~\cref{lemma:P_k-routing-implies-1-order-linked} in order to construct order-linked sets in our digraph.
	
	\begin{theorem}[{\cite[Theorem 6.16]{COSSI}}]
		\label{theorem:one-way connected temporal digraph contains P_k routing}
		\label[thm-abbr]{thm-abbr:one-way connected temporal digraph contains P_k routing}
		There is a function
		\boundDef{theorem:one-way connected temporal digraph contains P_k routing}{\Lifetime{}}{n, k}
		\(\bound{theorem:one-way connected temporal digraph contains P_k routing}{\Lifetime{}}{n,k} \in \PowerTower{1}{\Polynomial{6}{k, n}}\)
		such that for every temporal digraph $T$ where each layer is unilateral with $\Lifetime{T} \geq \bound{theorem:one-way connected temporal digraph contains P_k routing}{\Lifetime{}}{n, k}$ and $n \coloneqq \Abs{\V{T}} \geq \frac{k^2}{2}$ there is a set $S \subseteq \V{T}$ such that $T$ contains a $\Pk{k}$-routing over $S$. 
	\end{theorem}

	The following \namecref{theorem:strongly connected temporal digraph contains H routing} states that every temporal digraph of sufficiently long lifetime contains
	an $H$-routing for some $H \in \Set{ \biPk{k}, \Ck{k} }$ if all its layers are strongly connected.
	
	\begin{theorem}[{\cite[Theorem 6.20]{COSSI}}]
		\label{theorem:strongly connected temporal digraph contains H routing}
		\label[thm-abbr]{thm-abbr:strongly connected temporal digraph contains H routing}
		There exist functions
		\boundDef{thm-abbr:strongly connected temporal digraph contains H routing}{s}{k}
		\(\bound{theorem:strongly connected temporal digraph contains H routing}{s}{k} \in \Oh(k^{11})\) and
		\boundDef{thm-abbr:strongly connected temporal digraph contains H routing}{\Lifetime{}}{k}
		\(\bound{theorem:strongly connected temporal digraph contains H routing}{\Lifetime{}}{n,k} \in \Oh(k^{11} + k^{5} n)\)
        such that every temporal digraph $T$ in that $\Layer{T}{i}$ is strongly connected for all $1 \leq i \leq \Lifetime{T}$ fulfils the following.
		If $\Lifetime{T} \geq \bound{theorem:strongly connected temporal digraph contains H routing}{\Lifetime{}}{\Abs{\V{T}}, k}$, then for every set $S \subseteq \V{T}$ with $\Abs{S} \geq \bound{theorem:strongly connected temporal digraph contains H routing}{s}{k}$ there is a subset $S' \subseteq S$ with $\Abs{S'} = k$ such that $D$ contains an $H$-routing over $S'$ for some $H \in \Set{ \Ck{k},  \biPk{k}}$.
	\end{theorem}

	\subsection{Cycles of well-linked sets and cylindrical grids}
	\label{sec:COSS_to_cyl_grid}
	
	In this section, we introduce an abstraction of cylindrical grids.
	
	\begin{definition}[cycle of well-linked sets]
		\label{def:cycle-of-well-linked-sets}
		A \emph{cycle of well-linked sets} of width $w$ and length $\ell$ is a tuple $\Brace{\mathcal{S},\mathscr{P}}$ such that
		\begin{itemize}             
			\item $\mathcal{S}$ is a sequence of $\ell$ pairwise disjoint subgraphs $\Brace{S_0,\dots,S_{\ell - 1}}$, which are called \emph{clusters},
			\item for every $0 \leq i < \ell$ there are disjoint sets $A(S_i),B(S_i) \subseteq \V{S_i}$ of size $w$ such that $A(S_i)$ is well-linked to $B(S_i)$ in $S_i$, 
			\item $\mathscr{P}$ is a sequence of $\ell$ pairwise disjoint linkages $\Brace{\mathcal{P}_0, \mathcal{P}_1, \dots, \mathcal{P}_{\ell - 1}}$ such that, for every $0 \leq i < \ell$, $\mathcal{P}_i$ is a $B(S_i)$-$A(S_{(i+1 \operatorname{mod} \ell)})$-linkage of order $w$ which is internally disjoint from $S_i$ and $S_{i+1}$ and is disjoint from every $S \in \mathcal{S} \setminus \Set{S_i, S_{i+1}}$. 
		\end{itemize}
		As before, we call $(\SSS, \PPPP)$ \emph{strict} if in every cluster $S_i$ every vertex $v \in V(S_i)$ is contained in an $A(S_i)$-$B(S_i)$-path. 
	\end{definition}
	Note that a \emph{cycle of well-linked sets} of width $w$ and length $\ell$ is a pair \((\mathcal{S}, \mathscr{P} \cup \Set{\mathcal{P}_{\ell}})\) where \((\mathcal{S}, \mathscr{P})\) is a path of well-linked sets of width $w$ and length $\ell - 1$, and \(\mathcal{P}_\ell\) is a linkage from the $B$-set of the last cluster to the $A$-set of the first cluster that is internally disjoint from \((\mathcal{S}, \mathscr{P})\),
	similarly to how a cylindrical grid is a fence plus a linkage from the right side of the fence to its left side which is internally disjoint from the fence.

	A useful property of a $(p, q)$-fence $(\PPP, \QQQ)$ is that if $A \subseteq \Start{\PPP}$ and $B \subseteq \End{\PPP}$ are sets with $|A| = |B| \leq p$ then there
	is an $A{-}B$-linkage $\LLL$ of order $|A|$ in the graph $\bigcup \PPP \cup \bigcup \QQQ$.
	In terms of connectivity, a path of well-linked sets displays essentially the same properties as a fence.
	The following theorem formalises this intuition, showing that we can always find a large fence inside a large
	path of well-linked sets.

	\begin{theorem}[{\cite[Theorem 8.6]{COSSI}}]
		\label{thm:poss-to-fence}
		\label[thm-abbr]{thm-abbr:poss-to-fence}
		There exist two functions \boundDef{thm-abbr:poss-to-fence}{w}{p,q} \(\bound{thm:poss-to-fence}{w}{p,q} \in \Oh(p^{5} q^{5})\), and \boundDef{thm-abbr:poss-to-fence}{\ell}{p,q} \(\bound{thm:poss-to-fence}{\ell}{p,q} \in \PowerTower{1}{\Polynomial{5}{p, q}}\) such that the following holds.
		Every path of a well-linked sets $\Brace{\mathcal{S} = (S_{0}, S_{1}, \ldots,\allowbreak S_{\ell}), \mathscr{P}}$ of width at least $\bound{thm-abbr:poss-to-fence}{w}{p,q}$ and length $\ell \geq \bound{thm-abbr:poss-to-fence}{\ell}{p,q}$ contains a $(p,q)$-fence \((\mathcal{P}, \mathcal{Q})\).
		Moreover, \(\Start{\mathcal{Q}} \subseteq A(S_0)\) and \(\End{\mathcal{Q}} \subseteq B(S_\ell)\).
	\end{theorem}

	In the same way as a path of well-linked sets can be constructed from a fence~\cite{COSSI}, a cycle of well-linked sets can easily be constructed from a cylindrical grid.
	We now turn to the converse operation, that is, how one can construct a cylindrical grid from a cycle of well-linked sets.
	We first need the following lemma from~\cite{kawarabayashi2015directed}.
	
	\begin{lemma}[{\cite[Lemma 6.3]{kawarabayashi2015directed}}]
		\label{lemma:fence-plus-disjoint-back-linkage-implies-cylindrical-grid}
		Let $t$ be an integer, let $(\mathcal{P}, \mathcal{Q})$ be a $(q,q)$-fence where $q \geq \bound{lemma:fence-plus-disjoint-back-linkage-implies-cylindrical-grid}{q}{t} \coloneqq (t-1)(2t-1)+1$ and let $\mathcal{R}$ be an $\End{\mathcal{P}}$-$\Start{\mathcal{P}}$-linkage of order $q$ which is internally disjoint from $(\mathcal{P}, \mathcal{Q})$. Then $(\mathcal{P}, \mathcal{Q})$ contains a cylindrical grid of order $t$ as a butterfly minor.
	\end{lemma}

	Combining \cref{lemma:fence-plus-disjoint-back-linkage-implies-cylindrical-grid,thm:poss-to-fence}, we obtain the following.

	\begin{theorem}
		\label{thm:coss contains grid}
		\label[thm-abbr]{thm-abbr:coss contains grid}
		There exist functions
		\boundDef{thm:coss contains grid}{w}{k}
		\(\bound{thm:coss contains grid}{w}{k} \in \Oh({k^{20}})\) and
		\boundDef{thm:coss contains grid}{\ell}{k}
		\(\bound{thm:coss contains grid}{\ell}{k} \in \PowerTower{1}{\Polynomial{9}{k}}\) such that the following holds.
		Every cycle of well-linked sets of width $w \geq \bound{thm-abbr:coss contains grid}{w}{k}$ and length $\ell \geq \bound{thm-abbr:coss contains grid}{\ell}{k}$ contains a cylindrical grid of order $k$.
	\end{theorem}
	\begin{proof}
		Let $k_1 = \bound{lemma:fence-plus-disjoint-back-linkage-implies-cylindrical-grid}{q}{k}$.
		Let $\ell_1 = \bound{thm:poss-to-fence}{len}{k_1, k_1}$.
		Note that $w \geq \bound{thm:poss-to-fence}{w}{k_1, k_1}$ and $\ell \geq \ell_1 + 1$.
		
		Let $\Brace{\mathcal{S} = \Brace{S_0, S_1, \dots, S_{\ell}}, \mathscr{P} = \Brace{\mathcal{P}_0, \mathcal{P}_1, \dots, \mathcal{P}_{\ell}}}$ be a cycle of well-linked sets of width $w$ and length $\ell$.
		Note that $D_1 \coloneqq \Brace{\mathcal{S}' \coloneqq \Brace{S_0, S_1, \dots, S_{\ell - 1}}, \mathscr{P}' \coloneqq \Brace{\mathcal{P}_0, \mathcal{P}_1, \dots, \mathcal{P}_{\ell - 2}}}$ is a path of well-linked sets of width $w$ and length at least $\ell_1$.
		
		By~\cref{thm:poss-to-fence}, $D_1$ contains a $(k_1,k_1)$-fence $(\mathcal{P}^1, \mathcal{Q}^1)$ such that $\Start{\mathcal{Q}^1} \subseteq A(S_0)$ and $\End{\mathcal{Q}^1} \subseteq B(S_{\ell - 1})$.
		Let $\mathcal{R}^1 \subseteq \mathcal{P}_\ell$ be the set of paths satisfying $\End{\mathcal{R}^1} = \Start{\mathcal{Q}^1}$.
		Let $\mathcal{R}^2$ be an $\End{\mathcal{Q}^1}$-$\Start{\mathcal{R}^1}$-linkage of order $k_1$ in $\SubPOSS{\Brace{\mathcal{S}, \mathscr{P}}}{\ell_1}{\ell}$.
		By~\cref{lem:linkage_inside_poss}\cref{case:linkage_inside_poss_B_A}, such a linkage $\mathcal{R}^2$ exists.
		Further, $\mathcal{R}^2$ is internally disjoint from $(\mathcal{P}^1, \mathcal{Q}^1)$.
		By~\cref{lemma:fence-plus-disjoint-back-linkage-implies-cylindrical-grid}, $(\mathcal{P}^1, \mathcal{Q}^1)$ and $\mathcal{R}^2$ together contain a cylindrical grid of order $k$ as a butterfly minor.
	\end{proof}
	
	With the results of this section, we now have suitable abstractions of fences and cylindrical grids at hand.
	We have	also seen how to obtain a cylindrical grid from a cycle of well-linked sets.
    	What remains to show is how to find a cycle of well-linked sets in a given digraph of large directed \treewidth.
    We address this problem in the remainder of the paper. 
	
	\section{The 2-horizontal web}
	\label{sec:constructing cows}
	
We construct a cycle of well-linked sets from
a path of well-linked sets where the last cluster
is well-linked to the first cluster.
Despite the similarities of the statements, we cannot directly use many of the results of \cite{kawarabayashi2015directed}
since several of their steps depend on the functions of \cite[Corollary 5.19]{kawarabayashi2015directed},
which is a power tower whose height depends on the order of the cylindrical grid desired,
and is hence too large for the bounds we wish to obtain here.

Nevertheless, the overall structure of our proof is still based on \cite{kawarabayashi2015directed}.
In \cref{subsection:back-linkage-cluster-by-cluster}, we obtain a path of well-linked sets with a back-linkage intersecting it \emph{cluster-by-cluster} (or we obtain a cycle of well-linked sets).
Then, in \cref{subsection:2-horizontal-web}
we obtain \emph{horizontal webs}, which mimic the cluster-by-cluster property above.
	
	\subsection{Back-linkage intersecting cluster-by-cluster}
	\label{subsection:back-linkage-cluster-by-cluster}
	
	We start with the case where the back-linkage is mostly disjoint from the
	path of well-linked sets.
	With the following simple observation, we can
	construct a shorter path of well-linked sets and still have a
	back-linkage for it.
	This will be then used in \cref{lemma:back-linkage avoids many clusters}
	in order to construct a cycle of well-linked sets.
	
	\begin{observation}
		\label{obs:restricting_length_poss}
		Let $\bound{obs:restricting_length_poss}{w}{w} \coloneqq 2w$.\boundDef{obs:restricting_length_poss}{w}{w}
        Let $(\mathcal{S} = \Brace{S_0, S_1, \dots, S_{\ell}}, \mathscr{P})$ be a path of well-linked sets of width at least $\bound{obs:restricting_length_poss}{w}{w}$ and length $\ell$ in a digraph $D$.
		Let $\mathcal{R}$ be a $B(S_{\ell})$-$A(S_0)$-linkage of order $\bound{obs:restricting_length_poss}{w}{w}$.
		Let $0 \leq i \leq j \leq \ell$.
		Then there is a $B(S_{j})$-$A(S_i)$-linkage $\mathcal{R}'$ of order $w$ such that $\Subgraph{D}{\mathcal{R}'} \cap \SubPOSS{(\mathcal{S}, \mathscr{P})}{i}{j} \subseteq \Subgraph{D}{\mathcal{R} \cup \Start{\mathcal{R}'} \cup \End{\mathcal{R}'}}$.
	\end{observation}
	\begin{proof}
		If $j < \ell$, then by~\cref{lem:poss-simple-routing} there is a linkage $\LLL_B$ from $B(S_j)$ to $B(S_\ell)$ in $\SubPOSS{\Brace{\mathcal{S}, \mathscr{P}}}{j}{\ell}$ which is internally disjoint from $S_j$.
		If $j = \ell$, we set $\LLL_B$ as the linkage containing only the vertices of $B(S_{\ell})$ and no arcs.
		
		Similarly, if $i = 0$, we set $\LLL_A$ as the linkage containing only the vertices of $A(S_{0})$ and no arcs.
		Otherwise, we set $\LLL_A$ as an $A(S_0)$-$A(S_{i})$-linkage of order $2w$ in $\SubPOSS{\Brace{\mathcal{S}, \mathscr{P}}}{0}{i}$, which exists by~\cref{lem:poss-simple-routing}.
		
		Let $\mathcal{L} = \mathcal{L}_B \cdot \mathcal{R} \cdot \mathcal{L}_A$.
		As $\mathcal{L}_B$ and $\mathcal{L}_A$ are internally disjoint, we have that $\mathcal{L}$ is a half-integral linkage from $B(S_{j})$ to $A(S_{i})$.
		By~\cref{lemma:half_integral_to_integral_linkage}, $\Subgraph{D}{\mathcal{L}}$ contains a $B(S_k)$-$A(S_{i})$-linkage $\RRR'$ of order $w$.
		
		As both $\mathcal{L}_A$ and $\mathcal{L}_B$ are internally disjoint from $\SubPOSS{(\mathcal{S}, \mathscr{P})}{i}{j}$, we have that $\Subgraph{D}{\mathcal{R}'} \cap \SubPOSS{(\mathcal{S}, \mathscr{P})}{i}{j} \subseteq \Subgraph{D}{\mathcal{R} \cup \Start{\mathcal{R}'} \cup \End{\mathcal{R}'}}$.
	\end{proof}
	
	If a back-linkage $\mathcal{R}$ is completely disjoint
	from a contiguous \emph{part} of a path of well-linked sets,
	then we can use \cref{obs:restricting_length_poss,obs:restricting_width_poss}
	in order to obtain a cycle of well-linked sets, whereas the original back-linkage
	becomes essentially one of the linkages between the clusters.
	
	We define
	\begin{align*}
		\boundDefAlign{lemma:back-linkage avoids many clusters}{\ell'}{\ell}
		\bound{lemma:back-linkage avoids many clusters}{\ell'}{\ell} & \coloneqq \ell - 1,\\[0em]
		\boundDefAlign{lemma:back-linkage avoids many clusters}{r}{w}
		\bound{lemma:back-linkage avoids many clusters}{w'}{w} & \coloneqq 2w,\\[0em]
		\boundDefAlign{lemma:back-linkage avoids many clusters}{w'}{w}
		\bound{lemma:back-linkage avoids many clusters}{r}{w} & \coloneqq 2w.
	\end{align*}
	
	\begin{lemma}
		\label{lemma:back-linkage avoids many clusters}
		Let $w, \ell$ be integers, let $\Brace{\mathcal{S}, \mathscr{P}}$ be a path of well-linked sets of length $\ell' \geq \bound{lemma:back-linkage avoids many clusters}{\ell'}{\ell}$ and width $w' \geq \bound{lemma:back-linkage avoids many clusters}{w'}{w}$ with a partial back-linkage $\mathcal{R}$ of order $r \geq \bound{lemma:back-linkage avoids many clusters}{r}{w}$ in a digraph $D$. 
		If there is a $0 \leq i \leq \ell' - \ell + 1$ such that $\mathcal{R}$ is internally disjoint from $\SubPOSS{\Brace{\mathcal{S}, \mathscr{P}}}{i}{i + \ell - 1}$, then $D$ contains a cycle of well-linked sets of length $\ell$ and width $w$ as a subgraph.
	\end{lemma}
	\begin{proof}
		Let $D' \coloneqq \SubPOSS{\Brace{\mathcal{S}, \mathscr{P}}}{i}{i + \ell - 1}$, $\Brace{ S_0, S_1, \dots, S_{\ell'}	} \coloneqq \mathcal{S}$ and $\Brace{ \mathcal{P}_0, \mathcal{P}_1, \dots, \mathcal{P}_{\ell' - 1}} \coloneqq \mathscr{P}$.
		There is a $B(S_{i + \ell - 1})$-$A(S_i)$-linkage $\mathcal{R}'$ of order $r / 2 = w$ such that $\V{\mathcal{R}'} \cap \V{D'} \subseteq \V{\mathcal{R}} \cup \Start{\mathcal{R}'} \cup \End{\mathcal{R}'}$ due to~\cref{obs:restricting_length_poss}.
		As $\mathcal{R}$ is internally disjoint from $D'$, the linkage $\mathcal{R}'$ is a partial back-linkage for $D'$ of order $w$ which is also internally disjoint from $D'$.
		By~\cref{obs:restricting_width_poss}, $D'$ contains a path of well-linked sets $(\mathcal{S}' = \Brace{ S'_0, S'_1, \dots, S'_{\ell - 1}}, \mathscr{P}' = \Brace{ \mathcal{P}'_0, \mathcal{P}'_1, \dots, \mathcal{P}'_{\ell - 2}})$ of width $w$ and length $\ell - 1$ such that $S_j' \subseteq S_{i+j}$ for all $0 \leq j \leq \ell - 1$ and $\mathcal{P}_j' \subseteq \mathcal{P}_{i+j}$ for all $0 \leq j \leq \ell - 2$.
		Additionally, $A(S_0') = \End{\mathcal{R}'}$ and $B(S_{\ell -1}') = \Start{\mathcal{R}'}$.
		Hence, by definition, $\Brace{\mathcal{S}', \Brace{ \mathcal{P}'_0, \mathcal{P}'_1, \dots, \mathcal{P}'_{\ell - 1}, \mathcal{R}'}}$ is a cycle of well-linked sets of width $w$ and length $\ell$.
	\end{proof}

	If the back-linkage intersects most of the clusters of the path of well-linked sets, we need to control how these intersections happen.
	In particular, we show that we can obtain a path of well-linked sets and a back-linkage which intersects the clusters from the end to the beginning in a (mostly) ordered fashion.
	To describe this property precisely, we use the notion of \emph{jumps}, defined below.

	\begin{definition}
		\label{def:jumps}
		Let $\Brace{\mathcal{S},\mathscr{P}}$ be a path of well-linked sets of length $\ell$.
		A \emph{jump} of length $k$ over $\Brace{\mathcal{S},\mathscr{P}}$ is a path $R$ with $\Start{R} \in \V{S_i} \cup \V{\mathcal{P}_i}$ and $\End{R} \subseteq \V{S_j} \cup \V{\mathcal{P}_j}$ (if $j = \ell$, we require $\End{R} \subseteq \V{S_j}$ instead) such that $\Abs{j - i} = k$.
		If $i < j$, then $R$ is a \emph{forward} jump.
		If $i \geq j$ and $R$ is internally disjoint from $\Brace{\mathcal{S}, \mathscr{P}}$, then $R$ is a \emph{backward} jump, see~\cref{fig:jumps} for an illustration.
	\end{definition}
	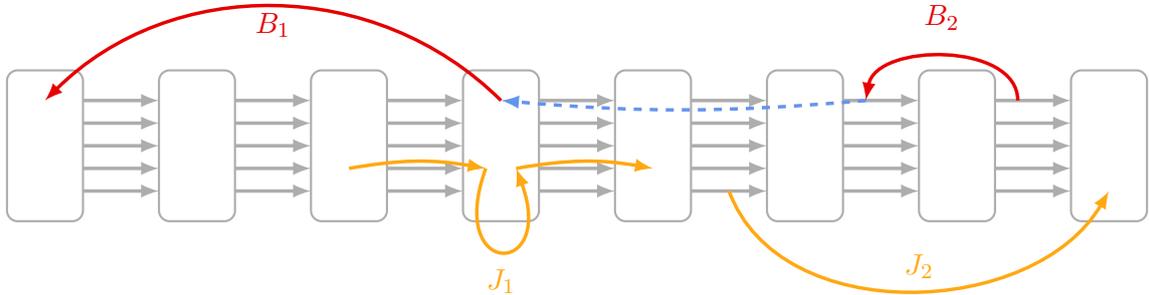
\begin{figure}[!ht]
		
		\begin{center}
		\begin{tikzpicture}
			\tikzstyle{cluster} = [draw, rectangle, minimum height = 2cm, minimum width = 1cm, thick, gray!80!white, rounded corners]
			\newcommand*{\linkage}[2]{
								\node (aux1) at ($(#1)+(0.5cm,0)$) {};
				\node (aux2) at ($(#2)+(-0.5cm,0)$) {};
				
				\draw[directededge,gray!80!white] ($(aux1)+(0,0)$) to ($(aux2)+(0,0)$);
				\draw[directededge,gray!80!white] ($(aux1)+(0,0.3)$) to ($(aux2)+(0,0.3)$);
				\draw[directededge,gray!80!white] ($(aux1)+(0,0.6)$) to ($(aux2)+(0,0.6)$);
				\draw[directededge,gray!80!white] ($(aux1)+(0,-0.3)$) to ($(aux2)+(0,-0.3)$);
				\draw[directededge,gray!80!white] ($(aux1)+(0,-0.6)$) to ($(aux2)+(0,-0.6)$);
				
			}
			\node (center) at (0,0) {};
			\def\clusterDist{2}
			
			\node[cluster] (S1) at ($(center)+(-3.5*\clusterDist,0)$) {};
			\node[cluster] (S2) at ($(S1)+(\clusterDist,0)$) {};
			\node[cluster] (S3) at ($(S2)+(\clusterDist,0)$) {};
			\node[cluster] (S4) at ($(S3)+(\clusterDist,0)$) {};
			\node[cluster] (S5) at ($(S4)+(\clusterDist,0)$) {};
			\node[cluster] (S6) at ($(S5)+(\clusterDist,0)$) {};
			\node[cluster] (S7) at ($(S6)+(\clusterDist,0)$) {};
			\node[cluster] (S8) at ($(S7)+(\clusterDist,0)$) {};
			
			\linkage{S1}{S2}
			\linkage{S2}{S3}
			\linkage{S3}{S4}
			\linkage{S4}{S5}
			\linkage{S5}{S6}
			\linkage{S6}{S7}
			\linkage{S7}{S8}
			
			\node (B1-start) at ($(S4)+(0,0.6)$) {};
			\node (B1-end) at ($(S1)+(0,0.6)$) {};
			\node (B1-label) at ($(B1-start)!0.5!(B1-end)+(0,1)$) {\textcolor{myRed}{$B_1$}};
			\node (B2-start) at ($(S7)+(0.4*\clusterDist,0.6)$) {};
			\node (B2-end) at ($(S6)+(0.4*\clusterDist,0.6)$) {};
			\node (B2-label) at ($(B2-start)!0.5!(B2-end)+(0,1.1)$) {\textcolor{myRed}{$B_2$}};
			
			\node (J1-start) at ($(S3)+(0,-0.3)$) {};
			\node (J1-m1) at ($(S4)+(-0.2,-0.3)$) {};
			\node (J1-m2) at ($(S4)+(0.2,-0.3)$) {};
			\node (J1-end) at ($(S5)+(0,-0.3)$) {};
			\node (J1-label) at ($(J1-start)!0.5!(J1-end)+(0,-1.5)$) {\textcolor{myOrange}{$J_1$}};
			
			\node (J2-start) at ($(S5)+(0.5*\clusterDist,-0.6)$) {};
			\node (J2-end) at ($(S8)+(0,-0.6)$) {};
			\node (J2-label) at ($(J2-start)!0.5!(J2-end)+(0,-1)$) {\textcolor{myOrange}{$J_2$}};
			
			\draw[directededge,myRed,out=135,in=45] (B1-start.center) to (B1-end.center);
			\draw[directededge,myRed,out=90,in=80] (B2-start.center) to (B2-end.center);
			
			\draw[directededge, myLightBlue, dashed,out=185,in=355] (B2-end.center) to (B1-start.center);
			
			\draw[directededge, myOrange,out=10,in=170] (J1-start.center) to (J1-m1.center);
			\draw[directededge, myOrange, out=250, in=290, looseness=10] (J1-m1.center) to (J1-m2.center);
			\draw[directededge, myOrange,out=10,in=170] (J1-m2.center) to (J1-end.center);
			
			\draw[directededge,myOrange,out=290,in=240] (J2-start.center) to (J2-end.center);
		\end{tikzpicture}
		\end{center}
		\caption{A path of well-linked sets with two backward jumps $B_1$ and $B_2$ and two forward jumps $J_1$ and $J_2$.
		Note that even if the \textcolor{myLightBlue}{dashed blue path} exists, the two jumps cannot be combined into one, as the result would not be internally disjoint from the path of well-linked sets.
		This is not required for forward jumps, as $J_1$ shows.}
		\label{fig:jumps}
	\end{figure}
	
		Note that while a backward jump must be internally disjoint from the path of well-linked sets, we do not require this from a forward jump.
    In fact, a forward jump could simply be a subpath of a path $R \in \RRR$ which is fully contained inside the path of well-linked sets $(\SSS, \PPPP)$.

	By excluding forward jumps, we guarantee that the back-linkage must intersect the clusters of a path of well-linked sets in an ordered fashion.
	Additionally excluding backward jumps of length greater than one ensures that the back-linkage intersects all clusters.
	These properties will be useful later when we construct a web using the back-linkage and the path of well-linked sets.

	\begin{definition}
		\label{def:cluster_by_cluster}
		Let $\Brace{\mathcal{S},\mathscr{P}}$ be a path of well-linked sets and let $\mathcal{R}$ be a partial back-linkage for $\Brace{\mathcal{S}, \mathscr{P}}$.
		We say that $\mathcal{R}$ intersects $\Brace{\mathcal{S},\mathscr{P}}$ \emph{cluster-by-cluster} if $\mathcal{R}$ does not contain any forward or backward jump of length greater than one over $\Brace{\mathcal{S},\mathscr{P}}$.
	\end{definition}
	
	Note that even if $\RRR$ intersects $(\SSS, \PPPP)$ cluster-by-cluster, this does not imply that the paths in $\RRR$ visit the clusters strictly in reverse order $S_\ell, S_{\ell-1}, \dots, S_0$.
	It is still possible that a path in $\RRR$ intersects a cluster $S_i$, then goes back to $S_{i+1}$, and then intersects $S_i$ again.
	So the paths in $\RRR$ can go back and forth between two consecutive clusters numerous times.
	However, once $R$ hits a vertex in $S_{i-1}$, it can no longer go back to $S_{i+1}$.

	To obtain a back-linkage intersecting our path of well-linked sets cluster-by-cluster, we start with a back-linkage $\mathcal{R}$ which is weakly minimal with respect to the path of well-linked sets.

	The weak-minimality property allows us to exclude long forward jumps in the case where we do not have long backward jumps, since we can use the well-linkedness of the clusters to construct a linkage	which would contradict the weak-minimality property.

  We show how to eliminate backward jumps in the second step, proved in~\cref{lemma:coss-or-back-linkage-intersects-cluster-by-cluster} below.
			
	\begin{lemma}
		\label{lemma:back-linkage intersects cluster by cluster}
		Let $\Brace{\mathcal{S}, \mathscr{P}}$ be a strict		\footnote{In order to properly use the intersections of the back-linkage with the path of well-linked sets, we need the path of well-linked sets to be \emph{strict} as defined in~\cref{def:path-of-well-linked-sets}.
		Since we can always take a path of well-linked sets to be strict without losing width or height, we often implicitly assume that the path of well-linked sets we construct is strict.}
		path of well-linked sets of length \boundDef{lemma:back-linkage intersects cluster by cluster}{\ell}{j,\ell,m} $\ell' \geq \bound{lemma:back-linkage intersects cluster by cluster}{\ell}{j,\ell,m} \coloneqq 3j\ell m$ and width $w$ in a digraph $D$ and let $\mathcal{R}$ be a partial back-linkage of order at least $w$ for $\Brace{\mathcal{S}, \mathscr{P}}$ which is weakly $m$-minimal with respect to $\Brace{\mathcal{S}, \mathscr{P}}$ and does not induce any backward jumps of length $j$ or more.
		Then, there is a path of well-linked sets $\Brace{\mathcal{S}', \mathscr{P}'}$ of length $\ell$ and width $w$ within $\ToDigraph{\Brace{\mathcal{S}, \mathscr{P}}}$ with a back-linkage $\mathcal{R}' \subseteq \mathcal{R}$ such that $\mathcal{R}'$ intersects $\Brace{\mathcal{S}', \mathscr{P}'}$ cluster-by-cluster.
	\end{lemma}
	\begin{proof}
		Let $\Brace{S_0, S_1, \dots, S_{\ell'}} \coloneqq \mathcal{S}$ and $\Brace{\mathcal{P}_0, \mathcal{P}_1, \dots, \mathcal{P}_{\ell' - 1}} \coloneqq \mathscr{P}$.
		First, we prove that $\mathcal{R}$ contains no forward jumps of length more than $3mj$.
		Suppose there were a path $R \in \mathcal{R}$ containing a forward jump $J$ of length more than $3mj$ with $\Start{J} \in \V{S_{s}} \cup \V{\mathcal{P}_{s}}$ and $\End{J} \in \V{S_{t}} \cup \V{\mathcal{P}_{t}}$, for some $s$ smaller than $t$.
		Let $R_a \cdot J \cdot R_b \coloneqq R$ be a decomposition of $R$ into subpaths. 
		Let $e_J \in E(R)$ be the arc of $R$ with its tail in $R_a$ and whose head is the first vertex of $J$.         
		Going to a longer forward jump containing $J$ if necessary, we may assume w.l.o.g.~that $e_J$ is not contained in $(\SSS, \PPPP)$. 
		
		For each $1 \leq i \leq 3m$ let $\mathcal{H}_i = \Set{S_{s + (i-1)j + k} \cup \ToDigraph{\mathcal{P}_{s + (i-1)j + k}} \mid 0 \leq k \leq j - 1}$.
		As $\mathcal{R}$ does not contain any backward jumps of length $j$ or more, for every $1 \leq i \leq 3m$ there is a subgraph $H^a_i \in \mathcal{H}_i$ which intersects $R_a$ and there is an $H^b_i \in \mathcal{H}_i$ which intersects $R_b$.
		
		For each $1 \leq i \leq m$, let $v_i^a$ be an arbitrary vertex of $\V{\mathcal{H}^a_{3i - 2}} \cap \V{R_a}$, let $v_i^b$ be an arbitrary vertex of $\V{\mathcal{H}^b_{3i}} \cap \V{R_b}$ and let $L_i$ be a $v_i^a$-$v_i^b$-path inside $\SubPOSS{\Brace{\mathcal{S}, \mathscr{P}}}{3i-2}{3i}$.
		By~\cref{lem:linkage_inside_poss}\cref{case:linkage_inside_poss_scattered_scattered}, such a path $L_i$ exists.
		Note that $L_i$ is disjoint from $L_j$ for all $1 \leq i,j \leq m$ where $i \neq j$.
		Thus, $\mathcal{L} = \Set{L_i \mid 1 \leq i \leq m}$ is a $\V{R_a}$-$\V{R_b}$-linkage of order $m$ in $\Brace{\mathcal{S}, \mathscr{P}}$ which does not contain the arc $e_J$ defined above, a contradiction to the assumption that $\mathcal{R}$ is weakly $m$-minimal with respect to $\Brace{\mathcal{S}, \mathscr{P}}$.
		Thus, $\mathcal{R}$ does not contain any forward jumps of length greater than $3mj$.
		
		Second, we construct the desired path of well-linked sets.
		For each $0 \leq k < \ell$, let $S'_k = S_{3kmj}$ and let $\mathcal{P}_k'$ be a $B(S_{3kmj})$-$A(S_{3(k+1)mj})$-linkage of order $m$ inside the path of well-linked sets $\SubPOSS{\Brace{\mathcal{S}, \mathscr{P}}}{3kmj}{3(k+1)mj}$.
		Further, let $S'_\ell = S_{\ell'}$ and $\mathcal{P}'_\ell$ be a $B(S_{3(\ell - 1)mj})$-$A(S_{\ell'})$-linkage of order $w$ inside $\SubPOSS{\Brace{\mathcal{S}, \mathscr{P}}}{3(\ell - 1)mj}{\ell'}$.
		By~\cref{lem:linkage_inside_poss}\cref{case:linkage_inside_poss_B_A}, such linkages $\mathcal{P}_k'$ exist.
		
		Let $\mathcal{S}' = \Brace{S'_0, S'_1, \ldots, S'_\ell}$ and let $\mathscr{P}' = \Brace{\mathcal{P}'_0, \mathcal{P}'_{3mj}, \ldots \mathcal{P}'_{(\ell - 1)3mj}}$.
		Note that $\Start{\mathcal{R}} \subseteq \V{S'_\ell}$ and $\End{\mathcal{R}} \subseteq \V{S'_0}$.
		By construction, $\Brace{\mathcal{S}',\mathscr{P}'}$ is a path of well-linked sets of width $w$ and length $\ell$.
		Furthermore, every jump over $\Brace{\mathcal{S}',\mathscr{P}'}$ of length $j'$, for some $j'$, in $\mathcal{R}$ is a jump over $\Brace{\mathcal{S}',\mathscr{P}'}$ of length $3mj'$.
		Hence, $\mathcal{R}$ does not contain any forward jumps or backward jumps of length greater than one over $\Brace{\mathcal{S}',\mathscr{P}'}$.
		Finally, the linkage $\mathcal{R}$ is a back-linkage for $\Brace{\mathcal{S}',\mathscr{P}'}$ and $\mathcal{R}$ intersects $\Brace{\mathcal{S}', \mathscr{P}'}$ cluster-by-cluster.
	\end{proof}

	In order to exclude long backward jumps so that we can apply
	\cref{lemma:back-linkage intersects cluster by cluster},
	we try to construct a cycle of well-linked sets from the backward jumps.
	If we have many very long backward jumps, then these jumps essentially act like a back-linkage
	in the context of \cref{lemma:back-linkage avoids many clusters}.
	
	We define
	\begin{align*}
		\boundDefAlign{lemma:coss-or-back-linkage-intersects-cluster-by-cluster}{\ell'}{w_1, \ell_1, \ell_2, m}\bound{lemma:coss-or-back-linkage-intersects-cluster-by-cluster}{\ell'}{w_1,
			\ell_1, \ell_2, m} & \coloneqq 3 \ell_2 m ((\ell_1 + 3) (3\ell_2 m)^{w_1} + 6 \frac{(3\ell_2 m)^{w_1} - 1}{w_1 - 1}), \\[0em]
		\boundDefAlign{lemma:coss-or-back-linkage-intersects-cluster-by-cluster}{w'}{w_1, w_2}\bound{lemma:coss-or-back-linkage-intersects-cluster-by-cluster}{w'}{w_1,
			w_2} & \coloneqq 2w_2 + w_1.
	\end{align*}
	Observe that \(\bound{lemma:coss-or-back-linkage-intersects-cluster-by-cluster}{\ell'}{w_1,\ell_1,\ell_2,m} \in \PowerTower{1}{\Polynomial{2}{w_1,\ell_1,\ell_2,m}}\).
	\begin{lemma}
		\label{lemma:coss-or-back-linkage-intersects-cluster-by-cluster}
		Let $\ell_1, w_1, \ell_2, w_2$ be integers, let \((\mathcal{S} = \Brace{S_0, S_1, \dots S_{\ell'}}, \mathscr{P} = (\mathcal{P}_0, \mathcal{P}_1, \dots, \mathcal{P}_{\ell' - 1}))\) be a strict path of well-linked sets of length $\ell' \geq \bound{lemma:coss-or-back-linkage-intersects-cluster-by-cluster}{\ell'}{w_1, \ell_1, \ell_2, m}$ and width $w' \geq \bound{lemma:coss-or-back-linkage-intersects-cluster-by-cluster}{w'}{w_1, w_2}$ with a partial back-linkage $\mathcal{R}$ of order at least $w_2$ in a digraph $D$ such that $\mathcal{R}$ is weakly $m$-minimal with respect to $\Brace{\mathcal{S}, \mathscr{P}}$.
		Then $D$ contains at least one of the following:
		\begin{enamerate}{C}{item:coss-or-back-linkage-cluster-by-cluster:back-linkage}
			\item 
			\label{item:coss-or-back-linkage-cluster-by-cluster:coss}
			a cycle of well-linked sets of length $\ell_1$ and width $w_1$, or
			\item 
			\label{item:coss-or-back-linkage-cluster-by-cluster:back-linkage}
			a path of well-linked sets of length $\ell_2$ and width $w_2$ together with a partial back-linkage $\mathcal{R}' \subseteq \mathcal{R}$ of order $w_2$ intersecting it cluster-by-cluster.
		\end{enamerate}
	\end{lemma}
	\begin{proof}
		We recursively define $d_i$ by $d_1 = \ell_1 + 3$ and $d_i = 3\ell_2 m d_{i - 1} + 6$.
		Solving the recurrence relation, we obtain that $d_{w_1} = ((\ell_1 + 3) (3\ell_2 m)^{w_1} + 6 \frac{(3\ell_2 m)^{w_1} - 1}{w_1 - 1})$ and thus $\ell' \geq 3\ell_2 m d_{w_1}$.
		
		Let $J_1, J_2$ be two backward jumps over $(\mathcal{S}, \mathscr{P})$ and let $x_1,x_2,y_1,y_2$ be such that \(\Start{J_1} \subseteq \V{S_{y_{1}}} \cup \V{\mathcal{P}_{y_1}}, \End{J_1} \subseteq \V{S_{x_1}} \cup \V{\mathcal{P}_{x_1}}, \Start{J_2} \subseteq \V{S_{y_2}} \cup \V{\mathcal{P}_{y_2}}\) and \(\End{J_2} \subseteq \V{S_{x_2}} \cup \V{\mathcal{P}_{x_2}}.\)
		We say that $J_1$ jumps over $J_2$ if $y_1 \geq y_2 + 2$ and $x_1 + 2 \leq x_2$.
		See~\cref{fig:jumps2} for an illustration.
		
		\begin{figure}[!ht]
			\begin{center}
			\begin{tikzpicture}
				\tikzstyle{cluster} = [draw, rectangle, minimum height = 2cm, minimum width = 1cm, thick, gray!80!white, rounded corners]
				\newcommand*{\linkage}[3]{
										\node (aux1) at ($(#1)+(0.5cm,0)$) {};
					\node (aux2) at ($(#2)+(-0.5cm,0)$) {};
					
					\draw[directededge,#3] ($(aux1)+(0,0)$) to ($(aux2)+(0,0)$);
					\draw[directededge,#3] ($(aux1)+(0,0.3)$) to ($(aux2)+(0,0.3)$);
					\draw[directededge,#3] ($(aux1)+(0,0.6)$) to ($(aux2)+(0,0.6)$);
					\draw[directededge,#3] ($(aux1)+(0,-0.3)$) to ($(aux2)+(0,-0.3)$);
					\draw[directededge,#3] ($(aux1)+(0,-0.6)$) to ($(aux2)+(0,-0.6)$);
					
																			}
				\node (center) at (0,0) {};
				\def\clusterDist{2}
				
				\node[cluster] (S1) at ($(center)+(-3.5*\clusterDist,0)$) {};
				\node[cluster,myGreen] (S2) at ($(S1)+(\clusterDist,0)$) {};
				\node (X1-label) at ($(S2)+(0,-1.5)$) {\textcolor{myGreen}{$X_1$}};
				\node[cluster] (S3) at ($(S2)+(\clusterDist,0)$) {};
				\node[cluster,myBlue] (S4) at ($(S3)+(\clusterDist,0)$) {};
				\node (X2-label) at ($(S4)+(0,-1.5)$) {\textcolor{myBlue}{$X_2$}};
				\node[cluster] (S5) at ($(S4)+(\clusterDist,0)$) {};
				\node[cluster,myBlue] (S6) at ($(S5)+(\clusterDist,0)$) {};
				\node (Y2-label) at ($(S6)+(0,-1.5)$) {\textcolor{myBlue}{$Y_2$}};
				\node[cluster] (S7) at ($(S6)+(\clusterDist,0)$) {};
				\node[cluster,myGreen] (S8) at ($(S7)+(\clusterDist,0)$) {};
				\node (Y1-label) at ($(S8)+(0,-1.5)$) {\textcolor{myGreen}{$Y_1$}};
				
				\draw[latex-latex,myOrange,thick] ($(S2.north east) + (0.1,0.1)$) -- ($(S4.north west) + (-0.1,0.1)$) node[midway,above] {$\geq 2$};
				\draw[latex-latex,myOrange,thick] ($(S6.north east) + (0.1,0.1)$) -- ($(S8.north west) + (-0.1,0.1)$) node[midway,above] {$\geq 2$};
				
				\linkage{S1}{S2}{gray!80!white}
				\linkage{S2}{S3}{myGreen}
				\linkage{S3}{S4}{gray!80!white}
				\linkage{S4}{S5}{myBlue}
				\linkage{S5}{S6}{gray!80!white}
				\linkage{S6}{S7}{myBlue}
				\linkage{S7}{S8}{gray!80!white}
				
				\node (J1-label) at ($(S8.north)!0.8!(S2.north)+(0,1.5)$) {\textcolor{myGreen}{$J_1$}};
				\node (J2-label) at ($(S6.north)!0.6!(S4.north)+(0,0.9)$) {\textcolor{myBlue}{$J_2$}};
				
				\draw[directededge,myGreen,out=150,in=30] (S8.north) to (S2.north);
				\draw[directededge,myBlue,out=150,in=30] (S6.north) to (S4.north);
			\end{tikzpicture}
			\end{center}
			\caption{The backward jump $J_1$ jumps over the backward jump $J_2$.}
			\label{fig:jumps2}
		\end{figure}
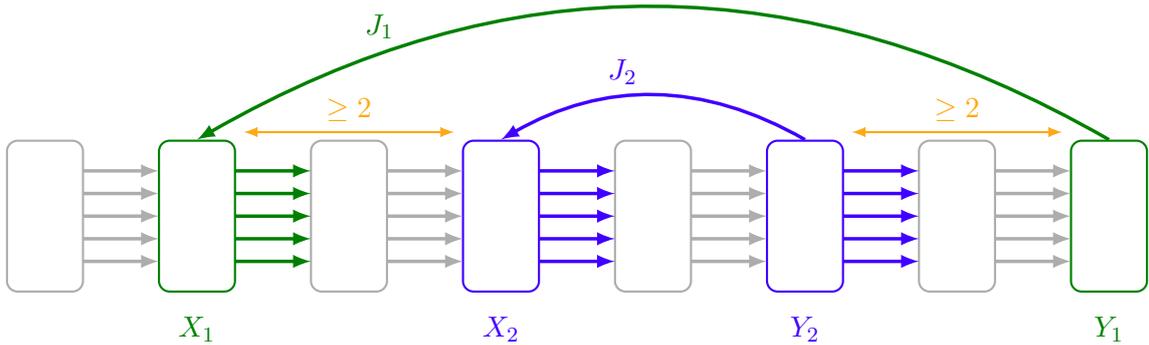
		
		If $\mathcal{R}$ does not contain any jump of length at least $d_{w_1}$ over $\Brace{\mathcal{S}, \mathscr{P}}$, then by~\cref{lemma:back-linkage intersects cluster by cluster} there is a path of well-linked sets $\Brace{\mathcal{S}' = \Brace{ S'_0, S'_1, \dots, S'_{\ell_2}}, \mathscr{P}'}$ of width $w_2$ and length $\ell_2$ together with a $B(S'_{\ell_2})$-$A(S'_0)$-linkage $\mathcal{R}'$ of order $w_2$ intersecting $\Brace{\mathcal{S}', \mathscr{P}'}$ cluster-by-cluster, satisfying~\cref{item:coss-or-back-linkage-cluster-by-cluster:back-linkage}. 
		
		Otherwise, let $r \in \Set{0, \ldots, w_1 - 1}$ be the highest number for which a set $\mathcal{J} = \Set{J_{w_1 -  r}, \ldots, J_{w_1}}$ of backward jumps over $(\mathcal{S}, \mathscr{P})$ exists such that for every $i \in \Set{w_1 - r, \ldots, w_1}$ and every $i+1 \leq j \leq w_1$, $J_i$ is a backward jump of length at least $d_i$ and $J_j$ jumps over $J_i$. 
		We distinguish between two possible cases.
		
		\textbf{Case 1:} $r < w_1 - 1$.
		
		Then $J_{w_1-r}$ is a backward jump of length at least $d_{w_1 - r}$.
		Let $i,j$ be such that \(\Start{J_{w_1 - r}} \subseteq \V{S_j} \cup \V{\mathcal{P}_j}\) and \(\End{J_{w_1 - r}} \subseteq \V{S_i} \cup \V{\mathcal{P}_i}.\)
		By~\cref{obs:restricting_length_poss}, there is a $B(S_{j-2})$-$A(S_{i+2})$-linkage $\mathcal{R}'$ of order $w_2$ such that $\V{\mathcal{R}'} \cap \V{\SubPOSS{(\mathcal{S}, \mathscr{P})}{i+2}{j-2}} \subseteq \V{\mathcal{R}} \cup \Start{\mathcal{R}'} \cup \End{\mathcal{R}'}$.
		Hence, any backward jump over $\SubPOSS{(\mathcal{S}, \mathscr{P})}{i+3}{j-3}$ contained in $\mathcal{R}'$ is also contained in $\mathcal{R}$.
		Finally, $\mathcal{R}'$ is also weakly $m$-minimal with respect to $\SubPOSS{(\mathcal{S}, \mathscr{P})}{i+3}{j-3}$.
		
		By choice of \(i\) and \(j\), if \(\mathcal{R}'\) contains a backward jump \(J'\) over $\SubPOSS{(\mathcal{S}, \mathscr{P})}{i+3}{j-3}$, then every jump in \(\mathcal{J}\) jumps over \(J'\).
		Since $r$ is maximal, there is no backward jump over $\SubPOSS{(\mathcal{S}, \mathscr{P})}{i+3}{j-3}$ of length at least $d_{w_1 - r - 1}$ in $\mathcal{R}'$.
										And because $j - 3 - (i + 3) \geq d_{w_1 - r} - 6 = 3 \ell_2 m d_{w_1 - r - 1}$, by~\cref{lemma:back-linkage intersects cluster by cluster} there is a path of well-linked sets $(\mathcal{S}', \mathscr{P}')$ of width $w_2$ and length $\ell_2$ together with a partial back-linkage $\mathcal{R}'' \subseteq \mathcal{R}'$ of order $w_2$ intersecting $\Brace{\mathcal{S}', \mathscr{P}'}$ cluster-by-cluster, satisfying~\cref{item:coss-or-back-linkage-cluster-by-cluster:back-linkage}.
		
		\textbf{Case 2:} $r = w_1 - 1$, that is, $w_1-r = 1$.
		
		We construct a linkage $\mathcal{R}'$ as follows.
		Let $i,j$ be such that $\Start{J_{w_1 - r}} \subseteq \V{S_j} \cup \V{\mathcal{P}_j}$ and $\End{J_{w_1 - r}} \subseteq \V{S_i} \cup \V{\mathcal{P}_i}$.
		
		For every two distinct jumps $J_x, J_y \in \mathcal{J}$ we have that $J_x$ jumps over $J_y$ or $J_y$ jumps over $J_x$.
		Further, $\mathcal{J}$ is internally disjoint from $\Brace{\mathcal{S}, \mathscr{P}}$.
		Hence, by~\cref{lem:linkage_inside_poss}\cref{case:linkage_inside_poss_B_scattered}, there is a $B(S_{j-2})$-$\Start{\mathcal{J}}$-linkage $\mathcal{X}_1$ of order $w_1$ in $\SubPOSS{(\mathcal{S}, \mathscr{P})}{j-2}{\ell'}$ which is internally disjoint from $S_{j - 2}$ and from $\mathcal{J}$.
		Additionally, by~\cref{lem:linkage_inside_poss}\cref{case:linkage_inside_poss_scattered_A}, there is an $\End{\mathcal{J}}$-$A(S_{i + 2})$-linkage $\mathcal{X}_2$ of order $w_1$ in $\SubPOSS{(\mathcal{S}, \mathscr{P})}{0}{i+2}$ which is internally disjoint from $S_{i + 2}$ and from $\mathcal{J}$.
		Thus, $\mathcal{R}' \coloneqq \mathcal{X}_1 \cdot \mathcal{J} \cdot \mathcal{X}_2$ is a linkage.
		
		By construction, the linkage $\mathcal{R}'$ above has order at least $w_1$ and is internally disjoint from $\SubPOSS{\Brace{\mathcal{S}, \mathscr{P}}}{i+2}{j-2}$, which is a path of well-linked sets of length $j-2 - (i + 2) \geq \ell_1 - 1$ and width $w_1$.
		Thus, by~\cref{obs:restricting_width_poss}, there is a path of well-linked sets $(\mathcal{S}' = \Brace{S'_0, S'_1, \dots, S'_{\ell_1 - 1}}, \mathscr{P}' = \Brace{\mathcal{P}'_0, \mathcal{P}'_1, \dots, \mathcal{P}'_{\ell - 2}})$ of length $\ell_1 - 1$ and width $w_1$ inside $\ToDigraph{\Brace{\mathcal{S}, \mathscr{P}}}$ such that $B(S'_{\ell_1 - 1}) = \Start{\mathcal{R}'}$ and $A(S'_0) = \End{\mathcal{R}'}$.
		By definition, $\Brace{\mathcal{S}', \Brace{ \mathcal{P}'_0, \mathcal{P}'_1, \dots, \mathcal{P}'_{\ell_1 - 2}, \mathcal{R}'}}$ is a cycle of well-linked sets of length $\ell_1$ and width $w_1$, satisfying~\cref{item:coss-or-back-linkage-cluster-by-cluster:coss}.
	\end{proof}
	
	\subsection{Obtaining a 2-horizontal web}
	\label{subsection:2-horizontal-web}
	
	We use the back-linkage intersecting our
	path of well-linked sets cluster-by-cluster
	in order to construct a new web with additional properties
	which we use later in \cref{sec:cows-from-horizontal-web}.
	Intuitively, we are looking for a web $(\mathcal{H}, \mathcal{V})$ which mimics the cluster-by-cluster property.
	That is, we want to be able to decompose the \emph{horizontal} linkage $\mathcal{H}$ into parts such that $\mathcal{V}$ intersects $\mathcal{H}$ in a horizontally sorted fashion.
	This type of web is called a \emph{horizontal web} and is defined below.
	An example is provided in \cref{fig:q-horizontal-web}.
	
	\begin{definition}
		\label{def:horizontal-web}
		Let $\Brace{\mathcal{H}, \mathcal{V}}$ be a web.
		We say that $\Brace{\mathcal{H}, \mathcal{V}}$ is a \emph{$c$-\HorizontalWeb} if every path $H_i \in \mathcal{H}$ can be decomposed into paths $H_i = H_i^1 \cdot H_i^2 \cdot \ldots \cdot H_i^c$ and every path $V_j \in \mathcal{V}$ can be decomposed into paths $V_j = V_j^1 \cdot V_j^2 \cdot \ldots \cdot V_j^c$ such that $V_j^x \cap H_i \subseteq H_i^{c - x} \cup H_i^{c - x + 1}$ and $V_j^x \cap H_i^{c - x + 1} \neq \emptyset$ for all $1 \leq x \leq c$, where for simplicity we define $H_i^{0}$ to be empty.
	\end{definition}

    \begin{figure}
        \centering
        \includegraphics{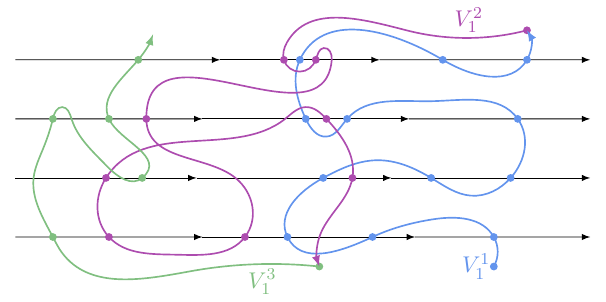}
        \caption{A $(4,1)$-web that is a $3$-horizontal web.
            Four horizontal paths $H_1,H_2,H_3$ and $H_4$ partitioned into three subpaths each and one vertical path $V_1$ partitioned into three subpaths \textcolor{myLightBlue}{$V_1^1$}, \textcolor{myLightViolet}{$V_1^2$} and \textcolor{myLightGreen}{$V_1^3$}.
            \textcolor{myLightBlue}{$V_1^1$} only intersects the latter two subpaths of the horizontal paths.
            \textcolor{myLightViolet}{$V_1^2$} only intersects the first two, note that it always intersects the second subpaths, but not necessarily the first.
            Finally, \textcolor{myLightGreen}{$V_1^3$} only intersects the first subpath of the horizontal paths.
        }
        \label{fig:q-horizontal-web}
    \end{figure}
		
	We start by constructing an ordered web from a back-linkage and a path of well-linked sets.
	
	\begin{lemma}
		\label{lemma:coss_or_ordered_web}
		Let $(\mathcal{S}, \mathscr{P})$ be a strict path of well-linked sets of length $\ell$ and width at least $1$ in a digraph $D$, and let $\mathcal{R}$ be a partial back-linkage of order $r$ intersecting $(\mathcal{S}, \mathscr{P})$ cluster-by-cluster.
		
		If $\ell \geq \bound{lemma:coss_or_ordered_web}{\ell}{r,v} \coloneqq 2(r - 1) + 2r(v - 1)$, then there is a linkage $\mathcal{V} \coloneqq \Brace{ V_1, V_2, \dots, V_{v}}$ of order $v$ inside $\ToDigraph{\Brace{\mathcal{S}, \mathscr{P}}}$ such that $\Brace{\mathcal{R}, \mathcal{V}}$ is an ordered web and for all $1 \leq i \leq v$ there are $0 \leq s_i \leq t_i \leq \ell$ with $V_i \subseteq \SubPOSS{(\mathcal{S}, \mathscr{P})}{s_i}{t_i}$ such that $t_i < s_j$ for all $1 \leq i < j \leq v$.
	\end{lemma}
	\begin{proof}
		Let $\Brace{ S_0, S_1, \dots, S_{\ell}} \coloneqq \mathcal{S}$ and $\Brace{ \mathcal{P}_0, \mathcal{P}_1, \dots, \mathcal{P}_{\ell - 1}} \coloneqq \mathscr{P}$.
		To simplify notation, we set $\mathcal{P}_\ell \coloneqq \emptyset$.
		Since $\mathcal{R}$ intersects $\Brace{\mathcal{S}, \mathscr{P}}$ cluster-by-cluster, every $R \in \mathcal{R}$ intersects some vertex of $\ToDigraph{S_i \cup \mathcal{P}_i}$ for every $0 \leq i < \ell$.
		Let $\mathcal{R} = \Brace{R_1, R_2, \dots, R_{r}}$ be an arbitrary ordering of the paths in $\mathcal{R}$.
		For each $1 \leq i \leq v$, construct a path $V^i$ intersecting every path in $\mathcal{R}$ as follows.
		
		For each $R_j \in \mathcal{R}$, let $k^i_j = 2(j - 1) + 2r(i - 1)$ and let $u^i_j$ be some vertex in $\V{R_j} \cap \V{S_{k^i_j}} \cup \V{\mathcal{P}_{k^i_j}}$.
		Let $Q^i_j$ be a path visiting $u^i_j$ with  $y^i_j \coloneqq \Start{Q^i_j} \in A(S_{k^i_j})$ and  $z^i_j\coloneqq \End{Q^i_j} \in A(S_{k^i_j + 1})$.
		Since $k^i_{j} - k^i_{j - 1} = 2$ for all $2 \leq j \leq r$, by~\cref{lem:linkage_inside_poss}\cref{case:linkage_inside_poss_scattered_scattered} there is a $z^i_j$-$y^i_{j+1}$-path $V^i_j$ inside $\SubPOSS{\Brace{\mathcal{S}, \mathscr{P}}}{k^i_j}{k^i_{j + 1}}$ for every $1 \leq j < r$.
		Since all $V^i_j$ and all $Q^i_j$ are pairwise internally disjoint and $V^i_j$ intersects $R_j$ at $u^i_j$, the path $V^i = Q^i_1 \cdot V^i_1 \cdot Q^i_2 \cdot V^i_2 \cdot \ldots \cdot Q^i_r$ in $\SubPOSS{\Brace{\mathcal{S}, \mathscr{P}}}{k^i_1}{k^i_r}$ intersects every path in $\mathcal{R}$.
		
		Let $\mathcal{V} = \Set{V^i \mid 1 \leq i \leq v}$.
		Since $k^i_1 - k^{i-1}_r = 2$ for all $2 \leq i \leq v$, all paths in $\mathcal{V}$ are pairwise disjoint.
		Further, such paths exist because $\ell \geq k^v_r = 2(r - 1) + 2r(v - 1)$.
		Finally, because $\mathcal{R}$ intersects $\Brace{\mathcal{S}, \mathscr{P}}$ cluster-by-cluster, $\Brace{\mathcal{R}, \mathcal{V}}$ is an ordered web.
	\end{proof}
	
	Our goal is to obtain a 2-horizontal web $(\mathcal{H}, \mathcal{V})$ where $\mathcal{H}$ is weakly $\Abs{\mathcal{H}}$-minimal with respect to $\mathcal{V}$.
	The idea is that $\HHH$ goes \emph{forwards} through the path of well-linked sets $(\SSS, \PPPP)$ from beginning to end, that is,
	in the same direction as the path of well-linked sets.
	$(\SSS, \PPPP)$ contains a forward linkage from its beginning to its end, and as a result of the way the linkage $\VVV$ is constructed, we can construct $\HHH$ so that it forms a web together with $\VVV$.
	However, if we naively choose $\mathcal{H}$ as a $\mathcal{V}$-minimal linkage,
	then we may lose intersections between $\mathcal{H}$ and $\mathcal{V}$, and the result would
	no longer be a 2-horizontal web.

	To address this issue, we need to be able to construct a cycle of well-linked sets in the case
	where many paths of $\mathcal{H}$ become disjoint from many paths in $\mathcal{V}$.
	In \cref{lemma:well-linked-poss-with-disjoint-forward-linkage-implies-coss}, we show that
	a path of well-linked sets which contains a \emph{forward linkage} which is disjoint from
	the back-linkage also contains a cycle of well-linked sets.

	Towards proving \cref{lemma:well-linked-poss-with-disjoint-forward-linkage-implies-coss},
	we show in \cref{lemma:separate-linkage-from-well-linked-poss} below that a path of well-linked sets which contains a \emph{forward linkage} $\mathcal{L}$
	can be split in such a way that
	some $\mathcal{L}^* \subseteq \mathcal{L}$ is disjoint from a smaller path of well-linked sets.

	We define
	\begin{align*}
	\Fkt{w'}{w, \ell, \ell^* }
		& \coloneqq 
			\ell^{*} + \bound{proposition:order-linked to path of well-linked sets}{w}{w,\ell },
	\\[0em]
		\boundDefAlign{lemma:separate-linkage-from-well-linked-poss}{q}{w, \ell, \ell^* }
	\bound{lemma:separate-linkage-from-well-linked-poss}{q}{w, \ell, \ell^* }
		& \coloneqq 
			\bound{theorem:strongly connected temporal digraph contains H routing}{s}{\ell^{*} + \bound{proposition:order-linked to path of well-linked sets}{w}{w, \ell } },
	\\[0em]
		\boundDefAlign{lemma:separate-linkage-from-well-linked-poss}{\ell'}{w, \ell, \ell^*}
	\bound{lemma:separate-linkage-from-well-linked-poss}{\ell'}{w, \ell, \ell^*}
		& \coloneqq 
			(3 w \ell {\binom{\bound{lemma:separate-linkage-from-well-linked-poss}{q}{w,\ell, \ell^{*} }}{\Fkt{w'}{w, \ell, \ell^{*} }}} (\Fkt{w'}{w, \ell, \ell^{*} })! + 3)\\[0em]
			& \phantom{\coloneqq} \quad\cdot \bound{theorem:strongly connected temporal digraph contains H routing}{\Lifetime{}}{\bound{lemma:separate-linkage-from-well-linked-poss}{q}{w, \ell, \ell^*}, \Fkt{w'}{w, \ell, \ell^*}} - 1.
	\end{align*}
	We note that \(\bound{lemma:separate-linkage-from-well-linked-poss}{q}{w,\ell,\ell^*} \in \Polynomial{22}{w, \ell, {\ell}^{*}}\) and \(\bound{lemma:separate-linkage-from-well-linked-poss}{\ell'}{w,\ell,\ell^*} \in \PowerTower{1}{\Polynomial{243}{w, \ell, {\ell}^{*}}}\).
	
	\begin{lemma}
		\label{lemma:separate-linkage-from-well-linked-poss}
		Let $\ell^*,w$ be integers, let $D = \Brace{\mathcal{S} = \Brace{S_0, S_1, \dots, S_{\ell'}}, \mathscr{P} = (\mathcal{P}_0, \mathcal{P}_1, \dots,\allowbreak \mathcal{P}_{\ell' - 1})}$ be a strict path of well-linked sets of width $w' \geq 1$ and length $\ell' \coloneqq \bound{lemma:separate-linkage-from-well-linked-poss}{\ell'}{w, \ell^*}$.
		Let $\mathcal{L}$ be an $A(S_0)$-$B(S_{\ell'})$-linkage of order at least $\bound{lemma:separate-linkage-from-well-linked-poss}{q}{w, \ell, \ell^*}$ such that every path in $\mathcal{L}$ intersects every $S_i \in \mathcal{S}$.
		Then, there is an $\mathcal{L}^{*} \subseteq \mathcal{L}$ of order $\ell^*$ for which $\Subgraph{D}{\mathcal{S} \cup \mathscr{P}}$ contains a path of well-linked sets $\Brace{\mathcal{S}' = \Brace{S'_0, S'_1, \dots, S'_{\ell}}, \mathscr{P}' = \Brace{\mathcal{P}'_0, \mathcal{P}'_1, \dots, \mathcal{P}'_{\ell - 1}}}$ of width $w$ and length $\ell$ which is disjoint from $\mathcal{L}^{*}$ such that $A(S_0') \subseteq A(S_0)$ and $B(S_\ell') \subseteq B(S_{\ell'})$.
		Further, for every $0 \leq i < j \leq \ell$ there are $0 \leq i_0 < i_1 < j_0 < j_1 \leq \ell$ such that $\mathcal{S}_i' \subseteq \SubPOSS{D}{i_0}{i_1}$ and $\mathcal{S}_j' \subseteq \SubPOSS{D}{j_0}{j_1}$.
		Finally, $\ToDigraph{\mathscr{P}'} \subseteq \mathcal{L} \setminus \mathcal{L}^*$.
	\end{lemma}
	\begin{proof}
		Let
		$w_1 = \bound{proposition:order-linked to path of well-linked sets}{w}{w, \ell}$,
		$w_2 = \ell^* + w_1$,
		$w_3 = \bound{theorem:strongly connected temporal digraph contains H routing}{s}{w_2}$.
		$\ell_1 = w \ell$,
		$\ell_2 = \ell_1 \binom{w_3}{w_2}w_2! + 1$ and
		$\ell_3 = \bound{theorem:strongly connected temporal digraph contains H routing}{\ell}{w_2}$.
		Note that $\ell' \geq 3 \ell_2 \ell_3 - 1$ and that $\Abs{\mathcal{L}} \geq w_3$.
		
		For each $1 \leq i \leq \ell_2$, construct a temporal digraph $T_i$ as follows.
		For each $1 \leq j \leq \ell_3$, let $s_{i,j} = 3(i-1)\ell_3 + 3(j-1)$ and note that $s_{i,1} = 1 + s_{i-1,\ell_3}$ and that $s_{\ell_2, \ell_3} + 2 \leq \ell'$.
		Let $H_j^i = \Subgraph{D}{{S_{s_{i,j}} \cup S_{s_{i,j} + 1} \cup S_{s_{i,j} + 2}}}$.
		
		Let $\mathcal{H}^i = \Brace{H^i_1, H^i_2, \dots, H^i_{\ell_2}}$.
		Let $T_i$ be the routing temporal digraph of $\mathcal{L}$ through $\mathcal{H}^i$ as described in~\cref{def:routing-temporal-digraph}.
		Note that $\Lifetime{T_i} = \ell_3$ for every $i$.
		
		Let $A_i$ be the set containing the first intersection of each path in $\mathcal{L}$ with $H^i_1$ and let $B_i$ be the set containing the last intersection of each path in $\mathcal{L}$ with $H^i_{\ell_2}$.
		Since every path in $\mathcal{L}$ intersects every $S \in \mathcal{S}$, we have that $\Abs{A_i} = \Abs{B_i} = \Abs{\mathcal{L}}$.
		
		We show that every layer of $T_i$ is strongly connected.
		Let $L_a, L_b \in \mathcal{L}$ be two distinct paths. 
		Since every path in $\mathcal{L}$ intersects every cluster in $\mathcal{S}$, there is some $a_0 \in \V{S_{s_{i,j}}}$ and some $b_0 \in \V{S_{s_{i,j} + 2}}$ such that $L_a$ contains $a_0$ and $L_b$ contains $b_0$.
		Further, there is some $a_1 \in B(S_{s_{i,j}})$ and some $b_1 \in A(S_{s_{i,j} + 2})$ such that $a_0$ can reach $a_1$ in $S_{s_{i,j}}$ and $b_1$ can reach $b_0$ in $S_{s_{i,j} + 2}$.
		    				As $A(S_{s_{i,j}})$ is well-linked to $B(S_{s_{i,j} + 1})$, there is some $a_2$-$b_2$-path in $S_{s_{i,j} + 1}$, where $a_2 = \mathcal{P}_{s_{i,j}}(a_1)$ and $b_1 = \mathcal{P}_{s_{i,j}+2}(b_2)$.
		Hence, $a_0$ can reach $b_0$ in $H^i_j$.
		Thus, there is a $\V{L_a}$-$\V{L_b}$-path in $H^i_j$, which implies that $\Layer{T_i}{j}$ is strongly connected.
		(Recall that $\Layer{T_i}{j}$ is the layer $j$ of $T_i$.)
		
		As $\Lifetime{T_i} = \ell_3 = \bound{theorem:strongly connected temporal digraph contains H routing}{\Lifetime{}}{w_3, w_2}$ and $\Abs{\V{T_i}} = w_3 = \bound{theorem:strongly connected temporal digraph contains H routing}{s}{w_2}$, by~\cref{theorem:strongly connected temporal digraph contains H routing} $T_i$ contains an $R_i$-routing $\varphi_i$ for some $R_i \in \Set{\Ck{w_2}, \biPk{w_2}}$.
		Since there are $\ell_2$ temporal digraphs $T_i$, by the pigeon-hole principle there is a set $\mathcal{T} \coloneqq \Set{T_{t_0}, T_{t_1}, \dots, T_{t_{\ell_1}}}$ of size $\ell_1 + 1$ of temporal digraphs such that $R \coloneqq R_i = R_j$ and $\varphi \coloneqq \varphi_i = \varphi_j$ for all $T_i, T_j \in \mathcal{T}$.
		
		Let $R'$ be a path of length $w_2 - 1$ in $R$ and let $u_1, u_2, \dots, u_{w_2}$ be the vertices of $R'$ sorted according to their order on $R'$.
		Let $\mathcal{L}' = \Set{\varphi(u_i) \mid 1 \leq i \leq w_1}$, let $\mathcal{L}^* = \Set{\varphi(u_i) \mid w_1 + 1 \leq i \leq w_2}$ and let $\varphi' = \FktRest{\varphi}{\mathcal{L}'}$.
		Note that $\Abs{\mathcal{L}^*} = \ell^*$.
		
		For each $0 \leq j \leq \ell_1$, we construct a subgraph $S_j'$ of $D$ and a linkage $\mathcal{P}'_j$ as follows.
		Note that $\varphi'$ is a $\Pk{w_1}$-routing in $T_{t_j} - \V{\mathcal{L}^*}$.
		Let $\mathcal{Q}^j$ be the set of digraphs obtained by deleting $\V{\mathcal{L}^*}$ from each digraph in $\mathcal{H}^{t_j}$ and let $T'_j$ be the routing temporal digraph of $\mathcal{L}'$ through $\mathcal{Q}^j$.
		Observe that $\varphi'$ is also a $\Pk{w_1}$-routing in $T'_j$.
		
		Let $A_j' = A_j \cap \V{\mathcal{L}'}$ and let $B_j' = B_j \cap \V{\mathcal{L}'}$.
		By~\cref{lemma:P_k-routing-implies-1-order-linked}, we have that $A_j'$ is 1-order-linked to $B_j'$ in $\Subgraph{D}{\mathcal{Q}^j}$.
		We set $S'_j = \Subgraph{D}{\mathcal{Q}^j}$ and take $\mathcal{P}'_j$ as the $B_j'$-$A_{j+1}'$-linkage of order $w_1$ inside $\mathcal{L}'$ (to simplify notation, we set $A_{\ell_1 + 1}' = \End{\mathcal{L}'}$).
		
		After finishing the construction, let $\mathcal{S}' = \Brace{S'_0, S'_1, \dots, S'_{\ell_1}}$ and let \(\mathscr{P}' = (\mathcal{P}'_0, \mathcal{P}'_1, \dots, \mathcal{P}'_{\ell_1 - 1}).\)
		By construction, $\Brace{\mathcal{S}', \mathscr{P}'}$ is a path of 1-order-linked sets of width $w_1$ and length $\ell_1$.
		Furthermore, $\mathcal{L}^*$ is disjoint from $\Brace{\mathcal{S}', \mathscr{P}'}$.
		Finally, we have that $\ToDigraph{\mathcal{P}'_j} \subseteq \mathcal{L} \setminus \mathcal{L}^*$ and that $S_j' \subseteq \SubPOSS{D}{t_j}{t_{j+1}-1}$.
		
		By~\cref{lemma:merging path of order-linked sets}, $\Brace{\mathcal{S}', \mathscr{P}'}$ contains a path of $w$-order-linked sets $\Brace{\mathcal{S}^2, \mathscr{P}^2}$ of width $w_1$ and length $\ell$.
		By~\cref{proposition:order-linked to path of well-linked sets}, $\Brace{\mathcal{S}^2, \mathscr{P}^2}$ contains a path of well-linked sets $(\mathcal{S}^3, \mathscr{P}^3)$ of width $w$ and length $\ell$ satisfying the requirements of the statement.
	\end{proof}

	Using \cref{lemma:separate-linkage-from-well-linked-poss}, we can construct
	a path of well-linked sets which is disjoint from a \emph{forward linkage} $\mathcal{L}^*$ which
	is in turn disjoint from a back-linkage $\mathcal{R}$.
	The idea is then to apply \cref{lemma:coss_or_ordered_web} to the
	new path of well-linked sets and the back-linkage, obtaining an ordered web.
	With \cref{lemma:ordered-web-to-path-of-well-linked-sets-with-side-linkage}, we obtain
	another path of well-linked sets, this time following the \emph{direction} of $\mathcal{R}$.
	That is, the original \emph{forward linkage} $\mathcal{L}^*$ now behaves like a
	back-linkage for the last path of well-linked sets, and we can then construct
	a cycle of well-linked sets.

	We define
	\begin{align*}
		\bound{lemma:well-linked-poss-with-disjoint-forward-linkage-implies-coss}{r}{w, \ell }
		\boundDefAlign{lemma:well-linked-poss-with-disjoint-forward-linkage-implies-coss}{r}{w, \ell}
		& \coloneqq 
			\bound{lemma:ordered-web-to-path-of-well-linked-sets-with-side-linkage}{h}{2 w,\ell - 1 },
		\\[0em]
		\Fkt{\ell''}{w, \ell}
		& \coloneqq 
			\bound{lemma:coss-or-back-linkage-intersects-cluster-by-cluster}{\ell'}{w,\ell,\bound{lemma:coss_or_ordered_web}{\ell}{\bound{lemma:ordered-web-to-path-of-well-linked-sets-with-side-linkage}{h}{2 w,\ell - 1 },8 w + \bound{lemma:ordered-web-to-path-of-well-linked-sets-with-side-linkage}{v}{2 w,\ell - 1 } + 2 }, \bound{lemma:well-linked-poss-with-disjoint-forward-linkage-implies-coss}{r}{w,\ell} },
		\\[0em]
		\bound{lemma:well-linked-poss-with-disjoint-forward-linkage-implies-coss}{\ell'}{w, \ell}
		\boundDefAlign{lemma:well-linked-poss-with-disjoint-forward-linkage-implies-coss}{\ell'}{w, \ell}
		& \coloneqq 
		\bound{lemma:separate-linkage-from-well-linked-poss}{\ell'}{\bound{lemma:coss-or-back-linkage-intersects-cluster-by-cluster}{w'}{w,\bound{lemma:ordered-web-to-path-of-well-linked-sets-with-side-linkage}{h}{2 w,\ell - 1 } },\Fkt{\ell''}{w,\ell},2 w },
		\\[0em]
		\bound{lemma:well-linked-poss-with-disjoint-forward-linkage-implies-coss}{q}{w, \ell}
		\boundDefAlign{lemma:well-linked-poss-with-disjoint-forward-linkage-implies-coss}{q}{w, \ell}
		& \coloneqq 
		\bound{lemma:separate-linkage-from-well-linked-poss}{q}{w,\Fkt{\ell''}{w,\ell },\bound{lemma:ordered-web-to-path-of-well-linked-sets-with-side-linkage}{h}{2 w,\ell - 1 } }.
	\end{align*}
	Observe that we have
	\(\bound{lemma:well-linked-poss-with-disjoint-forward-linkage-implies-coss}{r}{w,\ell} \in \Oh(w^{2} \ell^{2} )\),
	\(\bound{lemma:well-linked-poss-with-disjoint-forward-linkage-implies-coss}{\ell'}{w,\ell} \in \PowerTower{3}{\Polynomial{25}{w, \ell}},\) and
	\(\bound{lemma:well-linked-poss-with-disjoint-forward-linkage-implies-coss}{q}{w,\ell} \in \PowerTower{2}{\Polynomial{25}{w, \ell}}\).
	
	\begin{lemma}
		\label{lemma:well-linked-poss-with-disjoint-forward-linkage-implies-coss}
		
		Let $\Brace{\mathcal{S} = \Brace{ S_0, S_1, \dots, S_{\ell'}}, \mathscr{P} = \Brace{ \mathcal{P}_0, \mathcal{P}_1, \dots, \mathcal{P}_{\ell' - 1}}}$ be a strict path of well-linked sets of width $w' \geq 1$ and length $\ell' \geq \bound{lemma:well-linked-poss-with-disjoint-forward-linkage-implies-coss}{\ell'}{w, \ell}$ with a partial back-linkage $\mathcal{R}$ of order $r \geq \bound{lemma:well-linked-poss-with-disjoint-forward-linkage-implies-coss}{r}{w, \ell}$ in a digraph $D$.
		Let $\mathcal{L}$ be an $A(S_0)$-$B(S_{\ell'})$ linkage of order $q \geq \bound{lemma:well-linked-poss-with-disjoint-forward-linkage-implies-coss}{q}{w, \ell, m}$ which is internally disjoint from $\mathcal{R}$ such that every $L \in \mathcal{L}$ intersects some vertex of $\ToDigraph{S_i \cup \mathcal{P}_i}$ for every $0 \leq i \leq \ell'$ and, for all \(0 \leq i < j \leq \ell'\), \(\mathcal{L}\) does not intersect \(\ToDigraph{S_i \cup \mathcal{P}_i}\) after intersecting \(\ToDigraph{S_j \cup \mathcal{P}_j}\).
		Then, $\ToDigraph{\mathcal{S} \cup \mathscr{P} \cup \mathcal{R} \cup \mathcal{L}}$ contains a cycle of well-linked sets of width $w$ and length $\ell$.
	\end{lemma}
	\begin{proof}
		Let 
		$\ell_4 = \ell - 1$,
		$\ell_3 = 4w + 1$,
		$v_1 = \bound{lemma:ordered-web-to-path-of-well-linked-sets-with-side-linkage}{v}{2w, \ell_4} + 2 \ell_3$,
		$w_2 = \bound{lemma:ordered-web-to-path-of-well-linked-sets-with-side-linkage}{h}{2w, \ell_4}$,
		$w_1 = \bound{lemma:coss-or-back-linkage-intersects-cluster-by-cluster}{w'}{w,w_2}$,
		$\ell_2 = \bound{lemma:coss_or_ordered_web}{\ell}{w_2, v_1}$ and
		$\ell_1 = \bound{lemma:coss-or-back-linkage-intersects-cluster-by-cluster}{\ell'}{w, \ell, \ell_2, r}$.
		Note that $w_2 \geq 2w$,
		$\ell' \geq \bound{lemma:separate-linkage-from-well-linked-poss}{\ell'}{w_1, \ell_1, 2w}$,
		$r \geq w_2$ and
		$q \geq \bound{lemma:separate-linkage-from-well-linked-poss}{q}{w, \ell_1, w_2}$.
		
		By~\cref{lemma:separate-linkage-from-well-linked-poss} there is some linkage $\mathcal{L}' \subseteq \mathcal{L}$ of order $2w$ and a path of well-linked sets $(\mathcal{S}^1 = ({ S^1_0, S^1_1, \dots, S^1_{\ell_1}}), \mathscr{P}^1 = ({ \mathcal{P}^1_0, \mathcal{P}^1_1, \dots, \mathcal{P}^1_{\ell_1 - 1}}))$ of width $w_1$ and length $\ell_1$ inside $(\mathcal{S}, \mathscr{P})$ with $A(S_0^1) \subseteq A(S_0)$ and $B(S_{\ell_1}^1) \subseteq B(S_{\ell'})$ such that $\mathcal{L}'$ is internally disjoint from $\Brace{\mathcal{S}^1, \mathscr{P}^1}$.
		Further, $\ToDigraph{\mathscr{P}^1} \subseteq \ToDigraph{\mathcal{L} \setminus \mathcal{L}'}$ and $S^1_i \subseteq S_i$ for all $0 \leq i \leq \ell_1$.
		
		Let \(\mathcal{R}^1 \subseteq \ToDigraph{\Brace{\mathcal{S}^1, \mathscr{P}^1}} \cup \ToDigraph{\mathcal{R}}\) be a \(\Brace{\mathcal{S}^1, \mathscr{P}^1}\)-minimal linkage of order \(\Abs{\mathcal{R}}\) such that \(\Start{\mathcal{R}^1} = \Start{\mathcal{R}}\) and \(\End{\mathcal{R}^1} = \End{\mathcal{R}}\).
		By~\cref{obs:H-minimal-implies-weakly-minimal}, \(\mathcal{R}^1\) is weakly \(r\)-minimal with respect to \(\Brace{\mathcal{S}^1, \mathscr{P}^1}\).
		Further, \(\mathcal{R}^1\) is internally disjoint from \(\mathcal{L}'\).
		Applying~\cref{lemma:coss-or-back-linkage-intersects-cluster-by-cluster} to $(\mathcal{S}^1, \mathscr{P}^1)$ and $\mathcal{R}^1$ yields two cases.
		
		If~\cref{item:coss-or-back-linkage-cluster-by-cluster:coss} holds, then we have a cycle of well-linked sets of width $w$ and length $\ell$, as desired.
		Otherwise,~\cref{item:coss-or-back-linkage-cluster-by-cluster:back-linkage} holds.
		That is, $\ToDigraph{\Brace{\mathcal{S}^1, \mathscr{P}^1} \cup \mathcal{R}}$ contains a path of well-linked sets $\Brace{\mathcal{S}^2, \mathscr{P}^2}$ of length $\ell_2$ and width $w_2$ and a linkage $\mathcal{R}^2 \subseteq \mathcal{R}^1$ of order $w_2$ such that $\mathcal{R}^2$ intersects $\Brace{\mathcal{S}^2, \mathscr{P}^2}$ cluster-by-cluster.
		Note that $\mathcal{R}^2$ is weakly $r$-minimal with respect to $\Brace{\mathcal{S}^2, \mathscr{P}^2}$.
		
		By~\cref{lemma:coss_or_ordered_web}, there is some linkage $\mathcal{V} = \Brace{ V_1, V_2, \dots, V_{v_1}}$ of order $v_1$ inside $\Brace{\mathcal{S}^2, \mathscr{P}^2}$ such that $\Brace{\mathcal{R}^2, \mathcal{V}}$ is an ordered web.
		Additionally, for all $1 \leq i < j \leq v_1$ there are $0 \leq s_i \leq t_i < s_j \leq t_j \leq \ell_2$ such that $V_i \subseteq \SubPOSS{(\mathcal{S}^2, \mathscr{P}^2)}{s_i}{t_i}$ and $V_j \subseteq \SubPOSS{(\mathcal{S}^2, \mathscr{P}^2)}{s_j}{t_j}$.
		
		Let $\mathcal{V}' = \Brace{V_{\ell_3 + 1}, V_{\ell_3 + 1}, \ldots , V_{v_1 - \ell_3}}$ and observe that $\Abs{\mathcal{V}'} = \bound{lemma:ordered-web-to-path-of-well-linked-sets-with-side-linkage}{v}{2w, \ell_4}$.
		Decompose $\mathcal{R}^2$ as $\mathcal{R}^2_a \cdot \mathcal{R}^3 \cdot \mathcal{R}^2_b \coloneqq \mathcal{R}^2$ such that $\mathcal{R}^3$ intersects all paths of $\mathcal{V}'$, $\End{\mathcal{R}^2_a} \subseteq \V{V_{v_1 - \ell_3}}$, $\Start{\mathcal{R}^2_b} \subseteq \V{V_{\ell_3 + 1}}$ and $\mathcal{R}^2_a$ and $\mathcal{R}^2_b$ are internally disjoint from $\mathcal{V}'$.
		
		By~\cref{lemma:ordered-web-to-path-of-well-linked-sets-with-side-linkage}, $\Brace{\mathcal{R}^3, \mathcal{V}'}$ contains a path of well-linked sets \((\mathcal{S}^3 = ( S^3_0, S^3_1, \dots, S^3_{\ell_4}), \mathscr{P}^3)\) of width $2w$ and length $\ell_4$ such that $A(S_0^3) \subseteq \Start{\mathcal{R}^3}$ and \(B(S^3_{\ell_4}) \subseteq \End{\mathcal{R}^3}.\)
		Since $\mathcal{L}'$ is internally disjoint from $\Brace{\mathcal{S}^2, \mathscr{P}^2}$ and from $\mathcal{R}$, it is also internally disjoint from $\Brace{\mathcal{S}^3, \mathscr{P}^3}$.
        We construct a partial back-linkage $\mathcal{R}^4$ for $\Brace{\mathcal{S}^3, \mathscr{P}^3}$ as follows.
		
		Choose some $B'_0 \subseteq B(S_0)$ and some $A'_{\ell'} \subseteq A(S_{\ell'})$ of size $2w$.
		Let $\mathcal{X}_1$ be some $\End{\mathcal{R}^2_b}$-$B'_0$-linkage of order $2w$ in $S_0$ and let $\mathcal{X}_2$ be some $A'_{\ell'}$-$\Start{\mathcal{R}^2_a}$-linkage of order $2w$ in $S_{\ell'}$.
		Since $\End{\mathcal{R}^2_b} \subseteq A(S_0)$ and $\Start{\mathcal{R}^2_a} \subseteq B(S_{\ell'})$, the linkages $\mathcal{X}_1$ and $\mathcal{X}_2$ exist.
		
		Fix an arbitrary ordering of $\mathcal{L}' = \Set{L'_1, L'_2, \dots, L'_{2w}}$.
		For each $L'_i \in \mathcal{L}'$ let $k_i = 2i + 1$ and choose some $v_i \in \V{L_i} \cap \V{S_{k_i} \cup \mathcal{P}_{k_i}}$ and some $u_i \in \V{L_i} \cap \V{S_{\ell' - k_i} \cup \mathcal{P}_{\ell' - k_i}}$.
		Let $Y_1 = \Set{v_i \mid 1 \leq i \leq 2w}$ and $Y_2 = \Set{u_i \mid 1 \leq i \leq 2w}$.
		Since $k_i - k_j \geq 2$ and $\ell' - k_i - (\ell' - k_j) \geq 2$ if $i < j$, by~\cref{lem:linkage_inside_poss}\cref{case:linkage_inside_poss_B_scattered} there is a $B''_0$-$Y_1$-linkage $\mathcal{Z}_1$ of order $2w$ inside $\SubPOSS{\Brace{\mathcal{S}, \mathscr{P}}}{0}{k_{2w}}$ and by~\cref{lem:linkage_inside_poss}\cref{case:linkage_inside_poss_scattered_A} there is a $Y_2$-$A''_{\ell'}$-linkage $\mathcal{Z}_2$ of order $2w$ inside $\SubPOSS{\Brace{\mathcal{S}, \mathscr{P}}}{\ell' - k_{2w}}{\ell'}$.
		
		Let $\mathcal{L}''$ be the sublinkage of $\mathcal{L}'$ from $Y_1$ to $Y_2$.
		By construction, $\mathcal{X}_1$, $\mathcal{X}_2$, $\mathcal{Z}_1$ and $\mathcal{Z}_2$ are pairwise internally disjoint.
		Further, $\mathcal{R}^2_a$ and $\mathcal{R}^2_b$ are disjoint since they are both part of the linkage $\mathcal{R}$.
		Hence, $\mathcal{R}'' = \mathcal{R}^2_b \cdot \mathcal{X}_1 \cdot \mathcal{Z}_1 \cdot \mathcal{L}'' \cdot \mathcal{Z}_2 \cdot \mathcal{X}_2 \cdot \mathcal{R}^2_a$ is a half-integral $B(S^3_{\ell_4})$-$A(S^3_0)$ linkage of order $2w$ which is internally disjoint from $\Brace{\mathcal{S}^3, \mathscr{P}^3}$.
		By~\cref{lemma:half_integral_to_integral_linkage}, there is a $\End{\mathcal{R}''}$-$\Start{\mathcal{R}''}$ linkage $\mathcal{R}^4$ of order $w$ inside $\ToDigraph{\mathcal{R}''}$.
		Hence, $\mathcal{R}^4$ is a partial back-linkage of order $w$ for $\Brace{\mathcal{S}^3, \mathscr{P}^3}$ which is internally disjoint from $\Brace{\mathcal{S}^3, \mathscr{P}^3}$.
		
		By~\cref{obs:restricting_width_poss}, $\Brace{\mathcal{S}^3, \mathscr{P}^3}$ contains a path of well-linked sets \((\mathcal{S}^4 = (S^4_0, S^4_1, \dots,\allowbreak S^4_{\ell_4}), \mathscr{P}^4)\) of width $w$ and length $\ell_4$ as a subgraph with $A(S_0^4) = \End{\mathcal{R}^4}$ and $B(S_{\ell_4}^4) = \Start{\mathcal{R}^4}$.
		By definition, $\Brace{\mathcal{S}^4, \Brace{ \mathcal{P}^4_0, \mathcal{P}^4_1, \dots, \mathcal{P}^4_{\ell_4 - 1}, \mathcal{R}^4}}$ is a cycle of well-linked sets of width $w$ and length $\ell$.
	\end{proof}

We consider a relaxation of horizontal webs where we no longer require
that every path in one linkage intersects every path
in the other linkage, which we call a \emph{\(c\)-horizontal semi-web}.
Observe that every \(c\)-horizontal web is also a \(c\)-horizontal semi-web.
	
	\begin{definition}
		\label{def:c-horizontal-semi-web}
		Let $\mathcal{H}, \mathcal{V}$ be two linkages.
		We say that $\Brace{\mathcal{H}, \mathcal{V}}$ is a $c$-\emph{horizontal semi-web} if $\mathcal{H}$ can be decomposed as $\mathcal{H} = \mathcal{H}^1 \cdot \mathcal{H}^2 \cdot \ldots \cdot \mathcal{H}^c$ and $\mathcal{V}$ can be decomposed as $\mathcal{V} = \mathcal{V}^1 \cdot \mathcal{V}^2 \cdot \ldots \cdot \mathcal{V}^c$ such that $\ToDigraph{\mathcal{V}^i} \cap \ToDigraph{\mathcal{H}} \subseteq \ToDigraph{\mathcal{H}^{c - i + 1} \cup \mathcal{V}^{c - i}}$ (we set $\mathcal{V}^0 = \emptyset$ for simplicity).
	\end{definition}

	While choosing the linkage $\mathcal{H}$ to be weakly minimal with respect to $\mathcal{V}$ in a horizontal web may produce a semi-web instead, we can still preserve the \emph{horizontal} property of the semi-web, as shown below.

	\begin{observation}
		\label{lemma:horizontally-minimal-semi-web}
		Let $\Brace{\mathcal{H}, \mathcal{V}}$ be a $3$-horizontal semi-web in a digraph $D$.
		Let \(\mathcal{H}^1 \cdot \mathcal{H}^2 \cdot \mathcal{H}^3 \coloneqq \mathcal{H}\) and
		\(\mathcal{V}^1 \cdot \mathcal{V}^2 \cdot \mathcal{V}^3 \coloneqq \mathcal{V}\)
		be decompositions of these two linkages witnessing that \((\mathcal{H}, \mathcal{V})\) is a 3-horizontal semi-web.
		Let \(\mathcal{U}^2 = \mathcal{V}^2 \cdot \mathcal{V}^3\).
		Then $\ToDigraph{\Brace{\mathcal{H}, \mathcal{V}}}$ contains a linkage $\mathcal{P} = \mathcal{P}^1 \cdot \mathcal{P}^2$ of order $\Abs{\mathcal{H}}$ and
		such that
		\(\mathcal{P}\) is weakly \(\Abs{\mathcal{H}}\)-minimal with respect to \(\mathcal{V}\),
		\(\mathcal{P}^2\) is internally disjoint from \(\mathcal{U}^2\),
		$\Start{\mathcal{P}} = \Start{\mathcal{H}}$,
		$\End{\mathcal{P}} = \End{\mathcal{H}}$ and
		$\Start{\mathcal{P}^2} \subseteq \V{\mathcal{H}^{2}}$.
	\end{observation}
	\begin{proof}
		Let $\mathcal{P}$ be a $\Start{\mathcal{H}}$-$\End{\mathcal{H}}$-linkage of order $\Abs{\mathcal{H}}$ which is $\mathcal{V}$-minimal.
		By~\cref{obs:H-minimal-implies-weakly-minimal}, $\mathcal{P}$ is also weakly $\Abs{\mathcal{H}}$-minimal with respect to $\mathcal{V}$.
		
		Since $\Brace{\mathcal{H}, \mathcal{V}}$ is a 3-horizontal semi-web, there is no path from $\mathcal{H}^1$ to $\mathcal{H}^3$ in $\mathcal{V}$.
		As $\Start{\mathcal{P}} = \Start{\mathcal{H}} = \Start{\mathcal{H}^1}$ and $\End{\mathcal{P}} = \End{\mathcal{H}} = \End{\mathcal{H}^3}$, every path in $\mathcal{P}$ must intersect some vertex of $\mathcal{H}^2$.
		Let $Y$ be the set containing, for each $P \in \mathcal{P}$, the last vertex of $P$ which is also in $\mathcal{H}^2$.
		
		Decompose $\mathcal{P}$ into $\mathcal{P}^1 \cdot \mathcal{P}^2 = \mathcal{P}$ such that $\Start{\mathcal{P}^2} = Y$.
				By construction, we have that $\mathcal{P}^2$ is internally disjoint from \(\mathcal{U}^2\).
    							\end{proof}

	Using \cref{lemma:horizontally-minimal-semi-web} above we can
	show that we can construct a horizontal web $(\mathcal{H}', \mathcal{V}')$
	from a semi-web $(\mathcal{H}, \mathcal{V})$
	such that $\mathcal{H}'$ is weakly minimal with respect to $\mathcal{V}'$, or
	we find large $\mathcal{H}' \subseteq \mathcal{H}$ and $\mathcal{V}' \subseteq \mathcal{V}$
	which are disjoint.
	From the latter case we will construct a cycle of well-linked sets in
	\cref{lemma:coss-or-minimal-2-horizontal-web},
	whereas the former case will be handled in \cref{sec:cows-from-horizontal-web}.
	
	We define
	\begin{align*}
		\boundDefAlign{lemma:c-horizontal-web-minimal-or-disjoint-linkage}{h}{h_1,h_2}
		\bound{lemma:c-horizontal-web-minimal-or-disjoint-linkage}{h}{h_1,h_2} & \coloneqq 
		2(h_2 - 1) + h_1,	\\[0em]
		\boundDefAlign{lemma:c-horizontal-web-minimal-or-disjoint-linkage}{v}{h,h_1,v_1,h_2,v_2}
		\bound{lemma:c-horizontal-web-minimal-or-disjoint-linkage}{v}{h,h_1,v_1,h_2,v_2} & \coloneqq 
		(v_2 - 1) \cdot 2 \binom{h}{h_2} +  (v_1 - 1) \cdot \binom{h}{h_1} - 1.
	\end{align*}
	Note that \(\bound{lemma:c-horizontal-web-minimal-or-disjoint-linkage}{v}{h,h_1,v_1,h_2,v_2} \in \PowerTower{1}{\Polynomial{3}{h, {h}_{1}, {v}_{1}, {h}_{2}, {v}_{2}}}\).

	\begin{lemma}
		\label{lemma:c-horizontal-web-minimal-or-disjoint-linkage}
		Let $\Brace{\mathcal{H}, \mathcal{V}}$ be a $3$-horizontal semi-web such that $h \coloneqq \Abs{\mathcal{H}} \geq \bound{lemma:c-horizontal-web-minimal-or-disjoint-linkage}{h}{h_1,h_2}$ and $v \coloneqq \Abs{\mathcal{V}} \geq \bound{lemma:c-horizontal-web-minimal-or-disjoint-linkage}{v}{h, h_1, v_1, h_2, v_2}$.
		Let \(\mathcal{H}^1 \cdot \mathcal{H}^2 \cdot \mathcal{H}^3 \coloneqq \mathcal{H}\) and
		\(\mathcal{V}^1 \cdot \mathcal{V}^2 \cdot \mathcal{V}^3 \coloneqq \mathcal{V}\) 
		be decompositions of these two linkages witnessing that \((\mathcal{H}, \mathcal{V})\) is a 3-horizontal semi-web.
		Then $\Brace{\mathcal{H}, \mathcal{V}}$ contains one of the following:
		\begin{enamerate}{W}{item:c-horizontal-web:minimal}
			\item \label{item:c-horizontal-web:minimal}
                a $2$-horizontal web $\Brace{\mathcal{H}', \mathcal{V}'}$ where $\ToDigraph{\mathcal{H}'} \subseteq \ToDigraph{\mathcal{H} \cup \mathcal{V}}$, $\mathcal{V}' \subseteq \mathcal{V}$, $\Abs{\mathcal{H}'} \geq h_1$, $\mathcal{H}'$ is weakly $h$-minimal with respect to $\mathcal{V}$ and $\Abs{\mathcal{V}'} \geq v_1$, or
			\item \label{item:c-horizontal-web:linkage}
                some linkage $\mathcal{H}' \subseteq \ToDigraph{\mathcal{H} \cup \mathcal{V}}$	of order $h_2$ and some linkage $\mathcal{V}' \subseteq \mathcal{V}$ of order $v_2$ such that $\mathcal{H}'$ is internally disjoint from~$\mathcal{V}'$.
                Additionally, $\Start{\mathcal{H}'} \subseteq \Start{\mathcal{H}}$ and $\End{\mathcal{H}'} \subseteq \V{\mathcal{H}^2}$, or $\Start{\mathcal{H}'} \subseteq \V{\mathcal{H}^2}$ and $\End{\mathcal{H}'} \subseteq \End{\mathcal{H}}$.
		\end{enamerate}
	\end{lemma}
	\begin{proof}
		By~\cref{lemma:horizontally-minimal-semi-web}, $\Brace{\mathcal{H}, \mathcal{V}}$ contains a linkage $\mathcal{P} = \mathcal{P}^1 \cdot \mathcal{P}^2$ of order $\Abs{\mathcal{H}}$
		such that 
		\(\mathcal{P}\) is weakly \(\Abs{\mathcal{H}}\)-minimal with respect to \(\mathcal{V}\),
		\(\mathcal{P}^2\) is internally disjoint from \(\mathcal{U}^2 \coloneqq \mathcal{V}^2 \cdot \mathcal{V}^3\),
		$\Start{\mathcal{P}} = \Start{\mathcal{H}}$,
		$\End{\mathcal{P}} = \End{\mathcal{H}}$ and
		$\Start{\mathcal{P}^2} \subseteq \V{\mathcal{H}^{2}}$.
		
		For each $V_j \in \mathcal{V}$ and each $1 \leq i \leq 2$ let $\mathcal{X}_j^i \subseteq \mathcal{P}^i$ be the paths of $\mathcal{P}^i$ which internally intersect $V_j$ and let $\mathcal{Y}_j^i \subseteq \mathcal{P}^i$ be the paths of $\mathcal{P}^i$ which are internally disjoint from $V_j$.
		Let $\mathcal{M} \subseteq \mathcal{V}$ be the set of paths $V_j \in \mathcal{V}$ for that some $i$ exists such that $\Abs{\mathcal{Y}_j^i} \geq h_2$.
		Let $\mathcal{N} = \mathcal{V} \setminus \mathcal{M}$.
		
		\textbf{Case 1:} $\Abs{\mathcal{M}} \leq (v_2 - 1) \cdot 2 \binom{\Abs{\mathcal{H}}}{h_2}$.
		
		Hence, $\Abs{\mathcal{N}} \geq \Abs{\mathcal{V}} - (v_2 - 1) \cdot 2 \binom{\Abs{\mathcal{H}}}{h_2} \geq (v_1 - 1) \cdot \binom{\Abs{\mathcal{H}}}{h_1} + 1$.
		For each $V_j \in \mathcal{N}$ let $\mathcal{X}_j = \Set{P^1 \cdot P^2 \in \mathcal{P} \mid P^1 \in \mathcal{X}^1_j \text{ and } P^2 \in \mathcal{X}^2_j}$.
		Since $\mathcal{X}^i_j \cup \mathcal{Y}^i_j = \mathcal{P}^i$ and $\Abs{\mathcal{Y}^i_j} < h_2$ for all $V_j \in \mathcal{N}$ and all $1 \leq i \leq 2$, we have that $\Abs{\mathcal{X}_j} \geq \Abs{\mathcal{H}} - 2 \cdot (h_2 - 1) \geq h_1$ for every $V_j \in \mathcal{N}$.
		
		There are at most $\binom{\Abs{\mathcal{H}}}{h_1}$ distinct linkages $\mathcal{H}' \subseteq \mathcal{P}$ of order $h_1$.
		By the pigeon-hole principle, there is some $\mathcal{V}' \subseteq \mathcal{V}$ of order $v_1$ for which some $\mathcal{H}' \subseteq \mathcal{P}$ of order $h_1$ exists such that $\mathcal{X}_j \supseteq \mathcal{H}'$ for all $V_j \in \mathcal{V}'$.

		Decompose \(\mathcal{H}'\) and \(\mathcal{V}'\) into
		\(\mathcal{F}^{1} \cdot \mathcal{F}^{2} \coloneqq \mathcal{H}'\) and
		\(\mathcal{W}^{1} \cdot \mathcal{W}^{2} \coloneqq \mathcal{V}'\)
		such that
		\(\mathcal{F}^2\) are maximal subpaths of \(\mathcal{P}^2\)
		which do not contain \(\Start{\mathcal{P}^2}\)
		and
		\(\mathcal{W}^1 \subseteq \mathcal{V}^1\).

		Because \(\mathcal{P}^2\)	is internally disjoint from \(\mathcal{U}^2\),
		we have that \(\mathcal{F}^2\) is disjoint from \(\mathcal{W}^2\).
		Further, by definition of \(\mathcal{X}^2_j\),
		every path in \(\mathcal{F}^2\) intersects every path in \(\mathcal{W}^1\).
		Finally, by definition of \(\mathcal{X}^1_j\),
		every path in \(\mathcal{F}^1\) intersects every path in \(\mathcal{W}^2\).
		Hence, $\Brace{\mathcal{H}', \mathcal{V}'}$ is a $2$-horizontal web with $\Abs{\mathcal{H}'} = h_1$, $\Abs{\mathcal{V}'} = v_1$, as witnessed by the decomposition above, and $\mathcal{H}'$ is weakly $\Abs{\mathcal{H}}$-minimal with respect to $\mathcal{V}'$, satisfying~\cref{item:c-horizontal-web:minimal}.
		
		\textbf{Case 2:} $\Abs{\mathcal{M}} \geq (v_2 - 1) \cdot 2 \binom{\Abs{\mathcal{H}}}{h_2} + 1$.
		
		By the pigeon-hole principle, there is some $i \in \Set{1,2}$ and some $\mathcal{H}' \subseteq \mathcal{P}^i$ of order $h_2$ for which there is a set $\mathcal{V}' \subseteq \mathcal{V}$ of order $v_2$ such that $\mathcal{Y}^i_j \supseteq \mathcal{H}'$ for all $V_j \in \mathcal{V}'$.
		Hence, $\mathcal{H}'$ is a linkage of order $h_2$ which is internally disjoint from $\mathcal{V}'$.
		
		If $i = 1$, then $\Start{\mathcal{H}'} \subseteq \Start{\mathcal{P}} = \Start{\mathcal{H}}$ and $\End{\mathcal{H}'} \subseteq \End{\mathcal{P}^1} \subseteq \V{\mathcal{H}^2}$.
		
		If $i = 2$, then $\Start{\mathcal{H}'} \subseteq \Start{\mathcal{P}^2} \subseteq \V{\mathcal{H}^2}$ and $\End{\mathcal{H}'} \subseteq \End{\mathcal{P}^2} = \End{\mathcal{H}}$.
		
		Hence,~\cref{item:c-horizontal-web:linkage} holds.
	\end{proof}

	A back-linkage $\mathcal{R}$ intersecting a path of well-linked sets cluster-by-cluster
	provides a similar structure as a horizontal web.
	Applying \cref{lemma:c-horizontal-web-minimal-or-disjoint-linkage} we can
	obtain a horizontal web $(\mathcal{H}, \mathcal{V})$ where $\mathcal{H}$ is weakly minimal with respect to $\mathcal{V}$, or
	we find some $\mathcal{H}' \subseteq \mathcal{H}$ which is disjoint from $\mathcal{V}' \subseteq \mathcal{V}$.
	Since the linkage $\mathcal{V}$ comes from the back-linkage $\mathcal{R}$ and $\mathcal{H}$ lies inside the path of well-linked sets,
	we can use \cref{lemma:well-linked-poss-with-disjoint-forward-linkage-implies-coss} to obtain
	a cycle of well-linked sets.

	We define
	\begin{align*}
		\boundDefAlign{lemma:coss-or-minimal-2-horizontal-web}{w}{h,w,\ell}
		\bound{lemma:coss-or-minimal-2-horizontal-web}{w}{h,w,\ell} & \coloneqq
		\bound{lemma:c-horizontal-web-minimal-or-disjoint-linkage}{h}{h, \bound{lemma:well-linked-poss-with-disjoint-forward-linkage-implies-coss}{q}{w,\ell}} + 2\bound{lemma:well-linked-poss-with-disjoint-forward-linkage-implies-coss}{r}{w, \ell},
		\\[0em]
		\boundDefAlign{lemma:coss-or-minimal-2-horizontal-web}{\ell}{w, \ell}
		\bound{lemma:coss-or-minimal-2-horizontal-web}{\ell}{w, \ell} & \coloneqq
		3 (\bound{lemma:well-linked-poss-with-disjoint-forward-linkage-implies-coss}{\ell'}{w, \ell} + 1) - 1,
		\\[0em]
		\boundDefAlign{lemma:coss-or-minimal-2-horizontal-web}{r}{h,w,\ell,v}
		\bound{lemma:coss-or-minimal-2-horizontal-web}{r}{h,w} & \coloneqq
		\bound{lemma:c-horizontal-web-minimal-or-disjoint-linkage}{v}
		{ \bound{lemma:c-horizontal-web-minimal-or-disjoint-linkage}{h}
			{ h
				, \bound{lemma:well-linked-poss-with-disjoint-forward-linkage-implies-coss}{q}{w, \ell}
			}
			, h
			, v
			, \bound{lemma:well-linked-poss-with-disjoint-forward-linkage-implies-coss}{q}{w, \ell}
			, 2\bound{lemma:well-linked-poss-with-disjoint-forward-linkage-implies-coss}{r}{w, \ell}
		},
		\\[0em]
		\boundDefAlign{lemma:coss-or-minimal-2-horizontal-web}{m}{h, w}
		\bound{lemma:coss-or-minimal-2-horizontal-web}{m}{h, w} & \coloneqq
		\bound{lemma:c-horizontal-web-minimal-or-disjoint-linkage}{h}{h, \bound{lemma:well-linked-poss-with-disjoint-forward-linkage-implies-coss}{q}{w, \ell}}.
	\end{align*}
	Observe that 
	\(\bound{lemma:coss-or-minimal-2-horizontal-web}{w}{h,w,\ell} \in \PowerTower{2}{\Polynomial{25}{h, w, \ell}}\),
	\(\bound{lemma:coss-or-minimal-2-horizontal-web}{\ell}{w,\ell} \in \PowerTower{3}{\Polynomial{25}{w, \ell}}\), as well as
	\(\bound{lemma:coss-or-minimal-2-horizontal-web}{r}{h,w,\ell,v} \in \PowerTower{3}{\Polynomial{26}{h,w,\ell,v}}\), and
	\(\bound{lemma:coss-or-minimal-2-horizontal-web}{m}{h,w,\ell} \in \PowerTower{2}{\Polynomial{25}{h, w, \ell}}\).
	
	\begin{lemma}
		\label{lemma:coss-or-minimal-2-horizontal-web}
		Let $w, \ell, h, v$ be integers, let $\Brace{\mathcal{S}, \mathscr{P}}$ be a strict path of well-linked sets of length $\ell' \geq \bound{lemma:coss-or-minimal-2-horizontal-web}{\ell}{w, \ell}$ and width $w' = \bound{lemma:coss-or-minimal-2-horizontal-web}{w}{h, w, \ell}$ with a back-linkage $\mathcal{R}$ of order $r \geq \bound{lemma:coss-or-minimal-2-horizontal-web}{r}{h,w,\ell,v}$ intersecting $\Brace{\mathcal{S}, \mathscr{P}}$ cluster-by-cluster.
		Then, $\ToDigraph{\mathcal{S} \cup \mathscr{P} \cup \mathcal{R}}$ contains one of the following:
		\begin{enamerate}{H}{item:coss-or-minimal-2-horizontal-web:coss}
			\item \label{item:coss-or-minimal-2-horizontal-web:coss}
                a cycle of well-linked sets of width $w$ and length $\ell$, or
			\item \label{item:coss-or-minimal-2-horizontal-web:c-horizontal-web}
                an $\bound{lemma:coss-or-minimal-2-horizontal-web}{m}{h, w, \ell}$-horizontally minimal $2$-horizontal web $\Brace{\mathcal{H}, \mathcal{V}}$ where $\mathcal{V} \subseteq \mathcal{R}$, $\Abs{\mathcal{H}} \geq h$ and $\Abs{\mathcal{V}} \geq v$.
		\end{enamerate}
	\end{lemma}
	\begin{proof}
		Let
		$h_1 = \bound{lemma:well-linked-poss-with-disjoint-forward-linkage-implies-coss}{\ell'}{w,\ell}$,
		$h_2 = \bound{lemma:well-linked-poss-with-disjoint-forward-linkage-implies-coss}{q}{w,\ell}$,
		$h_3 = \bound{lemma:c-horizontal-web-minimal-or-disjoint-linkage}{h}{h, h_2}$,
		$w_1 = \bound{lemma:well-linked-poss-with-disjoint-forward-linkage-implies-coss}{r}{w, \ell}$ and
		$v_1 = 2w_1$.
		
		Let $\Brace{ S_0, S_1, \dots, S_{\ell'}} \coloneqq \mathcal{S}$ and $\Brace{ \mathcal{P}_0, \mathcal{P}_1, \dots, \mathcal{P}_{\ell' - 1}} \coloneqq \mathscr{P}$.
		To simplify notation, set \(\mathcal{P}_{\ell'} \coloneqq \emptyset.\)
				Let $\mathcal{H}$ be an $A(S_0)$-$B(S_{\ell'})$-linkage of order $h_3$ in $\Brace{\mathcal{S}, \mathscr{P}}$.
		By~\cref{lem:poss-simple-routing}, such a linkage exists.
		
		Let $t_i = (i-1) (h_1 + 1)$ for each $i \in \Set{1,2,3,4}$.
		Decompose $\mathcal{H}$ into $\mathcal{H} = \mathcal{H}^1 \cdot \mathcal{H}^2 \cdot \mathcal{H}^3$, where $\mathcal{H}^i$ is the sublinkage of $\mathcal{H}$ contained in $\SubPOSS{\Brace{\mathcal{S}, \mathscr{P}}}{t_i}{ t_{i+1} - 1}$.
		Decompose $\mathcal{R}$ iteratively as follows.
		Let $X_0 = \Start{\mathcal{R}}$ and let $X_3 = \End{\mathcal{R}}$.
		For each $0 \leq i \leq 3$ let $Y_i$ be the vertices of $\mathcal{R}$ such that for each $R \in \mathcal{R}$ the last intersection of $R$ with $\ToDigraph{S_{t_{i+1} - 1} \cup \mathcal{P}_{t_{i + 1} - 1}}$ lies in $Y_i$.
		Let $X_i$ be the successors of the vertices of $Y_i$ in $\mathcal{R}$.
		For each $1 \leq i \leq 3$, let $\mathcal{R}^i$ be the $X_{i-1}$-$X_i$-sublinkage of order $\Abs{\mathcal{R}}$ in $\mathcal{R}$.
		To simplify notation, set $\mathcal{R}^0$ as the linkage containing only the vertices of $\Start{\mathcal{R}^1}$.
		
		Because $\mathcal{R}$ intersects $\Brace{\mathcal{S}, \mathscr{P}}$ cluster-by-cluster, we have $\ToDigraph{\mathcal{H}^i} \cap \ToDigraph{\mathcal{R}} \subseteq \ToDigraph{\mathcal{R}^{4 - i} \cup \allowbreak \mathcal{R}^{3 - i}}$ for all $1 \leq i \leq 3$.
		Hence, $\Brace{\mathcal{H}, \mathcal{R}}$ is a $3$-horizontal semi-web.
		Further, $\Abs{\mathcal{H}} = h_3 = \bound{lemma:c-horizontal-web-minimal-or-disjoint-linkage}{h}{h, h_2}$ and $\Abs{\mathcal{R}} \geq \bound{lemma:c-horizontal-web-minimal-or-disjoint-linkage}{v}{h_3, h, v, h_2, v_1}$.
		By~\cref{lemma:c-horizontal-web-minimal-or-disjoint-linkage}, we have two cases.
		
		\textbf{Case 1:}
		\Cref{item:c-horizontal-web:minimal} holds.
		That is, $\Brace{\mathcal{H}, \mathcal{R}}$ contains a $2$-horizontal web $\Brace{\mathcal{H}', \mathcal{V}'}$ where $\ToDigraph{\mathcal{H}'} \subseteq \ToDigraph{\mathcal{H} \cup \mathcal{R}}$, $\mathcal{V}' \subseteq \mathcal{R}$, $\Abs{\mathcal{H}'} \geq h$, $\Abs{\mathcal{V}'} \geq v$ and $\mathcal{H}'$ is weakly $\Abs{\mathcal{H}}$-minimal with respect to $\mathcal{V}'$.
		This satisfies~\cref{item:coss-or-minimal-2-horizontal-web:c-horizontal-web}.
		
		\textbf{Case 2:}
		\Cref{item:c-horizontal-web:linkage} holds.
		That is, there is some $\mathcal{H}' \subseteq \ToDigraph{\mathcal{H} \cup \mathcal{R}}$ of order $h_2$ and some $\mathcal{V}' \subseteq \mathcal{R}$ of order $v_1$ such that $\mathcal{H}'$ is internally disjoint from~$\mathcal{V}'$.
		Additionally, $\Start{\mathcal{H}'} \subseteq \Start{\mathcal{H}}$ and $\End{\mathcal{H}'} \subseteq \V{\mathcal{H}^2}$, or $\Start{\mathcal{H}'} \subseteq \V{\mathcal{H}^2}$ and $\End{\mathcal{H}'} \subseteq \End{\mathcal{H}}$.
		We assume, without loss of generality, that $\Start{\mathcal{H}'} \subseteq \Start{\mathcal{H}}$ holds.
		The other case follows analogously by considering $\mathcal{H}^3$ instead of $\mathcal{H}^1$.
		
		We show that every path in $\mathcal{H}'$ intersects $\ToDigraph{S_i \cup \mathcal{P}_i}$ for every $0 \leq i \leq h_1$.
		
		Assume towards a contradiction that this is not the case.
		As $\Start{\mathcal{H}'} \subseteq A(S_0)$ and $\End{\mathcal{H}'} \subseteq \V{S_{h_1 + d}}$ for some integer $d$, there is some $0 \leq j \leq h_1$ for which some $H \in \mathcal{H}'$ exists such that $H$ is an $A(S_0)$-$\V{S_{h_1 + d}}$-path which does not intersect any vertex of $\ToDigraph{\mathcal{S}_j \cup \mathscr{P}_j}$.
		
		Let $j_0 < j$ be the largest index smaller than $j$ such that $H$ intersects $\ToDigraph{S_{j_0} \cup \mathcal{P}_{j_0}}$.
		Similarly, let $j_1 > j$ be the smallest index larger than $j$ such that $H$ intersects $\ToDigraph{S_{j_1} \cup \mathcal{P}_{j_1}}$.
		Since there is no $\V{S_{j_0} \cup \mathcal{P}_{j_0}}$-$\V{S_{j_1} \cup \mathcal{P}_{j_1}}$-path which is disjoint from $\ToDigraph{S_j \cup \mathcal{P}_j}$ inside $\Brace{\mathcal{S}, \mathscr{P}}$, $H$ contains a $\V{S_{j_0} \cup \mathcal{P}_{j_0}}$-$\V{S_{j_1} \cup \mathcal{P}_{j_1}}$-path $H_x$ which is internally disjoint from $\Brace{\mathcal{S}, \mathscr{P}}$ as a subpath.
		Hence, $H_x$ is a subpath of $\mathcal{V}' \subseteq \mathcal{R}$.
		This, however, implies that $H_x$ is a jump of length $j_1 - j_0 > 1$, a contradiction to the initial assumption that $\mathcal{R}$ intersects $\Brace{\mathcal{S}, \mathscr{P}}$ cluster-by-cluster.
		
		Let $S_{h_1}' = S_{h_1} \cup \ToDigraph{\mathcal{P}_{h_1}}$.
		Clearly, the digraph $(\mathcal{S}' \coloneqq (S_0, S_1, \dots, S_{h_1 - 1}, S_{h_1}'), \mathscr{P}' \coloneqq (\mathcal{P}_0, \mathcal{P}_1, \dots,$ $\mathcal{P}_{h_1 - 1}))$ is a path of well-linked sets of width $w' \geq 2w_1$ and length $h_1$.
		By~\cref{obs:restricting_length_poss}, there is a partial back-linkage $\mathcal{R}'$ of order $w_1$ for $\SubPOSS{\Brace{\mathcal{S}, \mathscr{P}}}{0}{h_1}$ (and, hence, for $\Brace{\mathcal{S}', \mathcal{P}'}$ as well) such that $\ToDigraph{\mathcal{R}'} \cap \Brace{\mathcal{S}', \mathscr{P}'} \subseteq \ToDigraph{\mathcal{V}' \cup \Start{\mathcal{R}'} \cup \End{\mathcal{R}'}}$.
		As $\mathcal{V}'$ is weakly $m$-minimal with respect to $\Brace{\mathcal{S}', \mathscr{P}'}$, the back-linkage $\mathcal{R}'$ is also weakly $m$-minimal with respect to $\Brace{\mathcal{S}', \mathscr{P}'}$.
		
		As every path in $\mathcal{H}'$ intersects $\ToDigraph{S_j \cup \mathcal{P}_j}$ for every $0 \leq j \leq h_1$, by~\cref{lemma:well-linked-poss-with-disjoint-forward-linkage-implies-coss} the digraph $\ToDigraph{\mathcal{S}' \cup \mathscr{P}' \cup \mathcal{R}' \cup \mathcal{H}'}$ contains a cycle of well-linked sets of width $w$ and length $\ell$, implying~\cref{item:coss-or-minimal-2-horizontal-web:coss}.
	\end{proof}
	
	\section{Constructing a cycle of well-linked sets}
	\label{sec:cows-from-horizontal-web}

	We now handle the remaining case
	where we obtain a 2-horizontal web as in \cref{lemma:coss-or-minimal-2-horizontal-web}\cref{item:coss-or-minimal-2-horizontal-web:c-horizontal-web}.
	The idea is to construct a new path of well-linked sets (\cref{subsec:new-path-of-well-linked-sets}) on the \emph{first half} of $\mathcal{H}$ in the horizontal web and then use the \emph{last half} of $\mathcal{H}$ to complete the back-linkage.
	We then conclude with the proof of our main result in \cref{subsec:the-directed-grid-theorem}.

	\subsection{A new path of well-linked sets}
	\label{subsec:new-path-of-well-linked-sets}

	We start with two simple observations which allow us to reroute paths inside a web and
	can be proven by a simple application of \emph{Menger's Theorem} (\cref{thm:menger}).

	\begin{observation}
		\label{lemma:web_mixed_linkages}
		Let $D$ be a digraph and let $\Brace{\mathcal{P}, \mathcal{Q}}$ be a web where $\Abs{\mathcal{P}} = \Abs{\mathcal{Q}}$.
		Then $\Start{\mathcal{P}}$ is well-linked to $\End{\mathcal{Q}}$ in $\ToDigraph{\mathcal{P} \cup \mathcal{Q}}$.
	\end{observation}
	\begin{proof}
		Let $A \subseteq \Start{\mathcal{P}}$ and $B \subseteq \End{\mathcal{Q}}$ such that $\Abs{A} = \Abs{B}$.
		Since $\Brace{\mathcal{P}, \mathcal{Q}}$ is a web, there is no $A$-$B$-separator of size less than $\Abs{A}$, as such a separator must hit $\Abs{A}$ paths of $\mathcal{P}$.
		Hence, by~\cref{thm:menger} there is an $A$-$B$-linkage of size $\Abs{A}$ in $\ToDigraph{\mathcal{P} \cup \mathcal{Q}}$.
		Thus, $\Start{\mathcal{P}}$ is well-linked to $\End{\mathcal{Q}}$.
	\end{proof}
	
	\begin{observation}
		\label{lemma:well-linkedness_inside_horizontal_web}
		Let $\Brace{\mathcal{H}, \mathcal{V}}$ be a $2$-\HorizontalWeb.
		Define $\mathcal{H}^2 \coloneqq \Set{H_i^2 \mid H_i \in \mathcal{H}}$ and $\mathcal{V}^1 \coloneqq \Set{V_i^1 \mid V_i \in \mathcal{V}}$.
		Then, $\Start{\mathcal{H}^2}$ is well-linked to $\End{\mathcal{V}^1}$ inside $\ToDigraph{\mathcal{H}^2 \cup \mathcal{V}^1}$.
	\end{observation}
	\begin{proof}
		By definition $\Brace{\mathcal{H}^2, \mathcal{V}^1}$ is a web.
		Hence, by~\cref{lemma:web_mixed_linkages} $\Start{\mathcal{H}^2}$ is well-linked to $\End{\mathcal{V}^1}$ inside $\ToDigraph{\mathcal{H}^2 \cup \mathcal{V}^1}$.
	\end{proof}
	
	We also need the concepts of splits and segmentations together with the following \namecref{lemma:segmentation-or-2-split} from~\cite{kawarabayashi2015directed}.

	\begin{definition}[{\cite[Definitions 5.6 and 5.7]{kawarabayashi2022directed}}]
		\label{def:split-edge}
		Let $\PPP$ and $\QQQ^\star$ be linkages and let $\QQQ\subseteq \QQQ^\star$ be a sublinkage of order $q$. Let $r\geq 0$.
						        \begin{enamerate}{X}{last-item-split-edge}
			\item \label{def:split}
				An \emph{$(r, q')$-split of $(\PPP, \QQQ)$ (with respect to $\QQQ^\star$)} is a pair $(\PPP', \QQQ')$ of linkages of order $r = |\PPP'|$ and $q'=|\QQQ'|$ with $\QQQ'\subseteq \QQQ$ such that
				\begin{enumerate}[label=\itemsymbol]
					\item there is a path $P \in \PPP$ and arcs $e_1, \dots, e_{r-1}\in \A{P}\setminus \A{\QQQ^\star}$ such that $P = P_1e_1P_2\dots e_{r-1} P_r$ and $\PPP' \coloneqq (P_1, \dots, P_r)$,
					\item every $Q\in \QQQ'$ can be divided into subpaths $Q_1, \dots, Q_r$ such that $Q = Q_1e'_1\dots e'_{r-1}Q_r$ for suitable arcs $e'_1, \dots, e'_{r-1} \in E(Q)$, and
					\item $\emptyset \neq V(Q)\cap V(P_i) \subseteq V(Q_{r+1-i})$, for all $1\leq i \leq r$.
				\end{enumerate}
			
			\item \label{def:other-segmentation}
				A subset $\QQQ' \subseteq \QQQ$ of order $q'$ is a $q'$-\emph{segmentation of $P \in \mathcal{P}$ (with respect to $\QQQ^\star$)} if 
				\begin{enumerate}[label=\itemsymbol]
					\item there are arcs $e_1, \dots, e_{q'-1}\in \A{P}-\A{\QQQ^\star}$ with $P = P_1e_1 \dots P_{q'-1}e_{q'-1}P_{q'}$, for suitable subpaths $P_1, \dots, P_{q'}$, and
					\item $\QQQ'$ can be ordered as $(Q_1, \dots, Q_{q'})$ and $\V{Q_i}\cap \V{P} \subseteq \V{P_i}$.
				\end{enumerate}
			\item \label{def:segmentation}
				An \emph{$(r, q')$-segmentation (with respect to \(\mathcal{Q}^*\))} is a pair $(\PPP', \QQQ')$ where	$\PPP'$ is a linkage of order $r$ and $\QQQ'$ is a linkage of order $q'$ such that $\QQQ'$ is a $q'$-segmentation (with respect to \(\mathcal{Q}^*\)) of every path $P_i$ into segments $P^i_1e_1P^i_2\dots e_{q'-1}P^i_{q'}$.
			\item \label{def:ordered_segmentation}
				A segmentation $(\PPP', \QQQ')$ is \emph{ordered} if for all $P_i \in \PPP'$ the order $\Brace{Q_1,\dots,Q_{q'}}$ given by the $q'$-segmentation of $P_i$ is the same. 
				We say that $(\PPP', \QQQ')$ is an (ordered) $(r, q')$-segmentation of $(\PPP, \QQQ)$ if $\QQQ' \subseteq \QQQ$ and every path in $\PPP'$ is a subpath of a path in $\PPP$.\label{last-item-split-edge}
		\end{enamerate}
		An $(r, q)$-split $(\PPP, \QQQ)$ or an $(r, q)$-segmentation $(\PPP, \QQQ)$
		is well-linked if $\End{\mathcal{Q}}$ is well-linked to $\Start{\mathcal{Q}}$.
	\end{definition}

	Observe that every $(r, q)$-segmentation is an ordered $(r, q)$-web and
	that every $(r, q)$-split is a folded ordered \((r,q)\)-web.

	In order to construct a path of well-linked sets on the \emph{first half} of $\mathcal{H}$,
	we need to construct a split or a segmentation which is connected in an adequate way
	to the \emph{last half} of $\mathcal{H}$.
	Towards this end, we require the following statements from \cite{kawarabayashi2015directed, COSSI}.

	\begin{lemma}[{\cite[Corollary 5.12]{kawarabayashi2015directed}}]
		\label{lemma:segmentation-or-2-split}
		Let $H$ be a digraph and let $\mathcal{Q}^*$ be a linkage in $H$, and let $\mathcal{Q} \subseteq \mathcal{Q}^*$ be a linkage of order $q$.
		Let $P \subseteq H$ be a path intersecting every path in $\mathcal{Q}$.
		Let $c \geq 0$ be such that for every arc $e \in \A{P} \setminus \A{\mathcal{Q}^*}$ there are no $c$ pairwise vertex-disjoint paths in $H - e$ from $P_1$ to $P_2$ , where $P = P_1 \cdot e \cdot P_2$.
		For all $s, r \geq 0$, if $q \geq (r + c) \cdot s$, then
		\begin{enamerate}{R}{item:segmentation-or-2-split:2-split}
			\item \label{item:segmentation-or-2-split:segmentation}
                there is an $s$-segmentation $\mathcal{Q}' \subseteq \mathcal{Q}$ of $P$ with respect to $\mathcal{Q}^*$ or
			\item \label{item:segmentation-or-2-split:2-split}
                a $(2, r)$-split $((P_1 , P_2 ), \mathcal{Q}'' )$ of $(P, \mathcal{Q})$ with respect to $\mathcal{Q}^*$.
		\end{enamerate}
	\end{lemma}

	The following is an adaptation of the construction in the proof of \cite[Lemma 5.15]{kawarabayashi2015directed}, which is in turn based on some results of \cite{reed1996packing}.
	We obtain the same bounds as \cite{kawarabayashi2015directed}, but we require somewhat different properties from the segmentation obtained.
	In particular, we want the paths of the segmentation to be in the end of the linkage $\mathcal{H}^1$ in our horizontal web $(\mathcal{H}^1 \cdot \mathcal{H}^2 = \mathcal{H}, \mathcal{V})$.
	This will allow us to go from the last cluster of the path of well-linked sets to $\mathcal{H}^2$ without
	intersecting the new path of well-linked sets.

	The idea of the proof is to iteratively apply \cref{lemma:segmentation-or-2-split} to each path of $\mathcal{P}$.
	If we can split one path many times, the we obtain the desired $(y,q)$-split.
	Otherwise, we obtain many different paths of $\mathcal{P}$ which contain an $(x,q)$-segmentation.
	By taking the subpath at the end of path in $\mathcal{P}$ which is in the segmentation,
	we obtain the properties we need later.

	We define
	\begin{equation*}
		\boundDefAlign{lemma:web-to-split-or-segmentation-at-end}{q'}{q,c,z}
		\bound{lemma:web-to-split-or-segmentation-at-end}{q'}{q,c,z}  \coloneqq (q(c+1))^{2^{z}}(2^{2^{z} - 1})
	\end{equation*}
	and note that
	\(\bound{lemma:web-to-split-or-segmentation-at-end}{q'}{q,c,z} \in \PowerTower{2}{\Polynomial{2}{q, c, z}}\).
	
	\begin{lemma}
		\label{lemma:web-to-split-or-segmentation-at-end}
		Let $c,x,y,q,q' \geq 0$ and $p \geq x$ be integers.
		Let $\mathcal{W} = \Brace{\mathcal{P}, \mathcal{Q}}$ be a $\Brace{p,q'}$-web where $\mathcal{P}$ is weakly $c$-minimal with respect to $\mathcal{Q}$.
		If $q' \geq \bound{lemma:web-to-split-or-segmentation-at-end}{q'}{q, c, xy}$, then there is some $\mathcal{Q}' \subseteq \mathcal{Q}$ such that $\mathcal{W}$ contains one of the following	
		\begin{enamerate}{S}{item:web-to-split-or-segmentation-without-gaps:segmentation}
			\item \label{item:web-to-split-or-segmentation-without-gaps:split}
                a $\Brace{y, q}$-split $\Brace{\mathcal{P}', \mathcal{Q}'}$ of $\Brace{\mathcal{P}, \mathcal{Q}}$ or
			\item \label{item:web-to-split-or-segmentation-without-gaps:segmentation}
                an $\Brace{x, q}$-segmentation $\Brace{\mathcal{P}', \mathcal{Q}'}$ of $\Brace{\mathcal{P}, \mathcal{Q}}$ with $\End{\mathcal{P}'} \subseteq \End{\mathcal{P}}$.
		\end{enamerate}
	\end{lemma}
	\begin{proof}
		For all $0\leq i\leq xy$, we define values $q_i$ inductively as follows.
		We set $q_{xy} \coloneqq q$ and $q_{i-1} \coloneqq q_i \cdot (q_i+c)$.
		We first show that $q_{0} \leq \bound{lemma:web-to-split-or-segmentation-at-end}{q'}{q, c, xy}$.
		
		\begin{claim}
			\label{claim:bound-for-qi}
			$q_{i} \leq (q(c+1))^{2^{xy - i}}(2^{2^{xy - i} - 1})$ for all $0 \leq i \leq xy$
		\end{claim}
		\begin{claimproof}
			Clearly $q_{xy} = q \leq q(c+1)$.
			Assume the inequality holds from $xy$ to $i > 0$.
			By definition of $q_{i-1}$ we obtain
            \belowdisplayskip=-12pt
			\begin{align*}
				q_{i-1} & = q_i \cdot   (q_i+c)
				\\[0em]        & \leq (q(c+1))^{2^{xy - i}}(2^{2^{xy - i} - 1}) \cdot
				((q(c+1))^{2^{xy - i}}(2^{2^{xy - i} - 1}) + c)
				\\[0em]        & = ((q(c+1))^{2^{xy - i}}(2^{2^{xy - i} - 1}))^2 + 
				(q(c+1))^{2^{xy - i}}(2^{2^{xy - i} - 1})c
				\\[0em]        & = (q(c+1))^{2^{xy - (i - 1)}}(2^{2^{xy - (i - 1)} - 2}) + 
				(q(c+1))^{2^{xy - i}}(2^{2^{xy - i} - 1})c
				\\[0em]        & \leq (q(c+1))^{2^{xy - (i - 1)}}(2^{2^{xy - (i - 1)} - 2} + 2^{2^{xy - i} - 1})
				\\[0em]        & = (q(c+1))^{2^{xy - (i - 1)}}(2^{2^{xy - i} - 1} \cdot 2^{2^{xy - i} - 1} + 2^{2^{xy - i} - 1})
				\\[0em]        & = (q(c+1))^{2^{xy - (i - 1)}}(2^{2^{xy - i} - 1} \cdot (2^{2^{xy - i} - 1} + 1))
				\\[0em]        & \leq (q(c+1))^{2^{xy - (i - 1)}}(2^{2^{xy - i} - 1} \cdot 2^{2^{xy - i}})
				\\[0em]        & = (q(c+1))^{2^{xy - (i - 1)}}(2^{2^{xy - (i - 1)} - 1}).
			\end{align*}\qedhere
                                    		\end{claimproof}
		Hence, by~\cref{claim:bound-for-qi}, we have $q_{0} \leq q' \leq (q(c+1))^{2^{xy}}(2^{2^{xy} - 1}) = \bound{lemma:web-to-split-or-segmentation-at-end}{q'}{q, c, xy}$.
		
		\newcommand{\Seg}{\variablestyle{seg}}
		\newcommand{\Split}{\variablestyle{split}}
		
		For each $0 \leq i \leq xy$ we construct a tuple $\mathcal{M}_i \coloneqq (\mathcal{P}^i, \mathcal{Q}^i, \mathcal{S}_{\Seg}^i, \mathcal{S}_{\Split}^i)$ satisfying all of the following
		\begin{enamerate}{M}{item:web-to-split-or-segmentation-at-end:last}
			\item \label{item:web-to-split-or-segmentation-at-end:orders}
                $\mathcal{Q}^i \subseteq \mathcal{Q}^*$ is a linkage of order $q_i$ and $\mathcal{P}^i \subseteq \mathcal{P}$ is a linkage such that, if there is no $P \in \mathcal{S}_{\Split}^i \setminus \mathcal{S}^i_{\Seg}$ with $\End{P} \subseteq \End{\mathcal{P}}$, then $\Abs{\mathcal{P}^i \cup \mathcal{S}_{\Seg}^i} \geq x$,
			\item \label{item:web-to-split-or-segmentation-at-end:split}
                $(\mathcal{S}_{\Split}^i, \mathcal{Q}^i)$ is a $(y_i, q_i)$-split of $(\mathcal{P}, \mathcal{Q})$ and
			\item \label{item:web-to-split-or-segmentation-at-end:segmentation}
                $(\mathcal{S}_{\Seg}^i, \mathcal{Q}^i)$ is an $(x_i, q_i)$-segmentation of $\Brace{\mathcal{P}, \mathcal{Q}}$ where $\End{\mathcal{S}_{\Seg}^i} \subseteq \End{\mathcal{P}}$.
                \label{item:web-to-split-or-segmentation-at-end:last}
		\end{enamerate}
		
		Furthermore, for all $i$, $(\mathcal{P}^i, \mathcal{Q}^i)$ has linkedness $c$ and $\V{\mathcal{P}^i} \cap \V{\mathcal{S}_{\Seg}^i} = \emptyset$.
		Recall that, in particular, this means that the paths in $\mathcal{S}_{\Split}^i$ are the subpaths of a single path in $\mathcal{P}$ that is split by arcs $e \in \E{P} \setminus \E{\mathcal{Q}^*}$.
		
		We first set $\mathcal{P}^0 \coloneqq \mathcal{P}$, $\mathcal{Q}^0 \coloneqq \mathcal{Q}^*$, $\mathcal{S}_{\Seg}^0 \coloneqq \emptyset$, $\mathcal{S}_{\Split}^0 \coloneqq \emptyset$.
		Clearly, this satisfies the conditions~\cref{item:web-to-split-or-segmentation-at-end:orders}, \cref{item:web-to-split-or-segmentation-at-end:split} and \cref{item:web-to-split-or-segmentation-at-end:segmentation} defined above.
		
		We do the following on step $i + 1 \geq 1$.
		If $|\mathcal{S}_{\Split}^i| \geq y$ or if $|\mathcal{S}_{\Seg}^i| \geq x$, stop the construction.
		Otherwise, proceed as follows.
		
		We first set $\mathcal{S}_{\Seg}' \coloneqq \mathcal{S}_{\Seg}^i$.
		If there is no $P \in \mathcal{S}_{\Split}^i \setminus \mathcal{S}_{\Seg}^i$ such that $\End{P} \subseteq \End{\mathcal{P}}$, we choose a path $P \in \mathcal{P}^i$ and set $\mathcal{S}_{\Split}' = \{ P \}$ and $\mathcal{P}^{i+1} \coloneqq \mathcal{P}^i \setminus \{P\}$.
		Clearly, $\End{P} \subseteq \End{\mathcal{P}}$.
		
		Otherwise, if there is some $P \in \mathcal{S}_{\Split}^i \setminus \mathcal{S}_{\Seg}^i$ with $\End{P} \subseteq \End{\mathcal{P}}$, we set $\mathcal{S}_{\Split}' \coloneqq \mathcal{S}_{\Split}^i$ and $\mathcal{P}^{i+1} \coloneqq \mathcal{P}^i$.
		
		Now, let $P \in \mathcal{S}_{\Split}' \setminus \mathcal{S}_{\Seg}'$ with $\End{P} \subseteq \End{\mathcal{P}}$.
		We apply~\cref{lemma:segmentation-or-2-split} to $P, \mathcal{Q}^i$ and $\mathcal{Q}^*$ setting $r = s = q_{i+1}$ in the statement of the lemma.
		If~\cref{item:segmentation-or-2-split:segmentation} holds and there is a $q_{i+1}$-segmentation $\mathcal{Q}_1 \subseteq \mathcal{Q}^i$ of $P$ with respect to $\mathcal{Q}^*$, we set
		\begin{equation*}
    		\mathcal{Q}^{i+1} \coloneqq \mathcal{Q}_1, \qquad
            \mathcal{S}_{\Seg}^{i+1} \coloneqq \mathcal{S}_{\Seg}^i \cup \Set{ P } \quad \text{ and }\quad
            \mathcal{S}_{\Split}^{i+1} \coloneqq \mathcal{S}_{\Split}.
		\end{equation*}
		Otherwise,~\cref{item:segmentation-or-2-split:2-split} holds and there is a $\Brace{2, q_{i+1}}$-split $\Brace{ \Brace{P_1, P_2}, \mathcal{Q}_2}$  where $\mathcal{Q}_2 \subseteq \mathcal{Q}^i$.
		Then we set
		\begin{eqnarray*}
			\mathcal{Q}^{i+1} &\coloneqq& \mathcal{Q}_2,\\[0em]
			\mathcal{S}_{\Seg}^{i+1} &\coloneqq& \mathcal{S}_{\Seg}^i\quad\text{ and }\\[0em]
			\mathcal{S}_{\Split}^{i+1} &\coloneqq& (\mathcal{S}_{\Split}^i \setminus \{ P \}) \cup \{ P_1, P_2 \}.
		\end{eqnarray*}
		
		If there is no $P \in \Abs{\mathcal{S}_{\Split}^{i+1}}$ with $\End{P} \subseteq \End{\mathcal{P}}$, then we obtained a segmentation and so added a path to $\mathcal{S}^{i+1}_{\Seg}$.
		As~\cref{item:web-to-split-or-segmentation-at-end:orders} holds for $i$, we have that $\Abs{\mathcal{P}^{i + 1} \cup \mathcal{S}_{\Seg}^{i+1}} \geq x$ in this case.
		
		It is easily verified that the conditions~\cref{item:web-to-split-or-segmentation-at-end:orders}, \cref{item:web-to-split-or-segmentation-at-end:split} and \cref{item:web-to-split-or-segmentation-at-end:segmentation} are maintained for $\mathcal{P}^{i+1}$, $\mathcal{Q}^{i+1}$, $\mathcal{S}_{\Seg}^{i+1}$ and $\mathcal{S}_{\Split}^{i+1}$.
		In particular, the linkedness $c$ of $(\mathcal{P}^{i+1}, \mathcal{Q}^{i+1})$ is preserved as deleting or splitting paths cannot increase forward connectivity.
		This concludes the construction.
		
		Note that in the construction, after every $y$ steps, either~\cref{item:segmentation-or-2-split:2-split} holds after every application of~\cref{lemma:segmentation-or-2-split}, and so we find a set $\mathcal{S}_{\Split}^i$ of size $y$ or~\cref{item:segmentation-or-2-split:segmentation} holds in at least one of the $y$ steps, and so we add a path to $\mathcal{S}_{\Seg}^i$.
		Whenever this happens, we take a new path $P \in \mathcal{P}^{i}$ in the next iteration, which always exists because of~\cref{item:web-to-split-or-segmentation-at-end:orders}.
		
		Hence, in the construction above, in each step we either increase $x_i$ and add the path $P$ with $\End{P} \subseteq \End{\mathcal{P}}$ to $\mathcal{S}_{\Seg}^i$ or we increase $y_i$.
		After at most  $i \leq xy$ steps, either we have constructed a set $\mathcal{S}_{\Seg}^i$ of order $x$ or a set $\mathcal{S}_{\Split}^i$ of order $y$.
		
		Because~\cref{item:web-to-split-or-segmentation-at-end:split} holds, if we found a set $\mathcal{S}_{\Split}^i$ of order $y$, then we can choose any set $\mathcal{Q}'\subseteq \mathcal{Q}^i$ of order $q$ and $(\mathcal{S}_{\Split}^i, \mathcal{Q}')$ satisfies~\cref{item:web-to-split-or-segmentation-without-gaps:split}.
		
		If, instead, we get a set $\mathcal{S}_{\Seg} \coloneqq \mathcal{S}_{\Seg}^i$ of order $x' \geq x$, then, by~\cref{item:web-to-split-or-segmentation-at-end:segmentation},
		$(\mathcal{S}_{\Seg}, \mathcal{Q}^i)$ is an $(x, q)$-segmentation of $(\mathcal{P}, \mathcal{Q})$ such that $\End{\mathcal{S}_{\Seg}} \subseteq \End{\mathcal{P}}$, satisfying~\cref{item:web-to-split-or-segmentation-without-gaps:segmentation}.
		
		Finally, it is easily seen that if $\mathcal{W}$ is well-linked then so is $(\mathcal{S}_{\Split}^i, \mathcal{Q}')$ (in case~\cref{item:web-to-split-or-segmentation-without-gaps:split} holds) or $(\mathcal{S}_{\Seg}, \mathcal{Q}^i)$ (in case~\cref{item:web-to-split-or-segmentation-without-gaps:segmentation} holds).
	\end{proof}

	In order to use the results of \cref{sec:powls-framework},
	we use the split or the segmentation obtained from \cref{lemma:web-to-split-or-segmentation-at-end}
	to construct a folded ordered web or an ordered web.
	We also require the following observation.

	\begin{observation}[{\cite[Observation 9.6]{COSSI}}]
		\label{obs:split-to-folded-web}
		Let $\Brace{\mathcal{P}', \mathcal{Q}'}$ be a $\Brace{2p, q}$-split of $\Brace{\mathcal{P}, \mathcal{Q}}$.
		Then there is some $\mathcal{P}''$ containing only subpaths of $\mathcal{P}$ such that $\Brace{\mathcal{Q}', \mathcal{P}''}$ is a folded ordered $(q, p)$-web.
	\end{observation}

	We define
	\begin{align*}
		\boundDefAlign{cor:web-to-folded-or-ordered-web-without-gaps}{q}{q'', x}
		\bound{cor:web-to-folded-or-ordered-web-without-gaps}{q}{q'', x} & \coloneqq (q'')^{2^{2x - 1}} + 1,
		\\[0em]
		\boundDefAlign{cor:web-to-folded-or-ordered-web-without-gaps}{q'}{q, c, x, y}
		\bound{cor:web-to-folded-or-ordered-web-without-gaps}{q'}{q, c, x, y} & \coloneqq \bound{lemma:web-to-split-or-segmentation-at-end}{q'}{q,c,2 (2x - 1)(y - 1)y}, \text{ and}
		\\[0em]
		\boundDefAlign{cor:web-to-folded-or-ordered-web-without-gaps}{p}{x}
		\bound{cor:web-to-folded-or-ordered-web-without-gaps}{p}{x} & \coloneqq 2x - 1.
	\end{align*}
	Note that
	\(\bound{cor:web-to-folded-or-ordered-web-without-gaps}{q}{q'',x} \in \PowerTower{2}{\Polynomial{2}{q'', x}}\) and
	\(\bound{cor:web-to-folded-or-ordered-web-without-gaps}{q'}{q,c,x,y} \in \PowerTower{2}{\Polynomial{5}{q, c, x, y}}\).
	
	\begin{lemma}
		\label{cor:web-to-folded-or-ordered-web-without-gaps}
		Let $c,x,y,q'',q' \geq 0$, $q \geq \bound{cor:web-to-folded-or-ordered-web-without-gaps}{q}{q'', x}$ and $p \geq \bound{cor:web-to-folded-or-ordered-web-without-gaps}{p}{x}$ be integers.
		Let $\mathcal{W} = \Brace{\mathcal{P}, \mathcal{Q}}$ be a $\Brace{p,q'}$-web where $\mathcal{P}$ is weakly $c$-minimal with respect to $\mathcal{Q}$.
		If $q' \geq \bound{cor:web-to-folded-or-ordered-web-without-gaps}{q'}{q, c, x, y}$, then there is some $\mathcal{Q}' \subseteq \mathcal{Q}$ and some $\mathcal{P}'$ such that $\ToDigraph{\mathcal{P}'} \subseteq \ToDigraph{\mathcal{P}}$ and $W$ contains one of the following	
		\begin{enamerate}{O}{item:web-to-split-or-segmentation-without-gaps:segmentation}
			\item \label{item:web-to-folded-or-ordered-web-without-gaps:folded}
			a folded ordered $\Brace{q,y}$-web $\Brace{\mathcal{Q}', \mathcal{P}'}$, or
			\item \label{item:web-to-folded-or-ordered-web-without-gaps:ordered}
			an ordered $\Brace{x, q''}$-web $\Brace{\mathcal{P}', \mathcal{Q}'}$
			such that
			$\End{\mathcal{P}'} \subseteq \End{\mathcal{P}}$.
		\end{enamerate}
	\end{lemma}
	\begin{proof}
		Let	$x_1 = 2x - 1$.

		We apply~\cref{lemma:web-to-split-or-segmentation-at-end} to $\mathcal{W}$.
		If~\cref{item:web-to-split-or-segmentation-without-gaps:split} holds, then by~\cref{obs:split-to-folded-web} we obtain a folded ordered $(q,y)$-web and~\cref{item:web-to-folded-or-ordered-web-without-gaps:folded} holds.
		Otherwise,~\cref{item:web-to-split-or-segmentation-without-gaps:segmentation} holds and we obtain an $(x_1, q)$-segmentation $\Brace{\mathcal{P}^1, \mathcal{Q}'}$ of $(\mathcal{P}, \mathcal{Q})$ such that $\End{\mathcal{P}^1} \subseteq \End{\mathcal{P}}$.
		
		Recursively define $q_i$ by $q_{x_1} = q''$ and $q_i = \bound{thm:erdos_szekeres}{len}{q_{i + 1}, q_{i + 1}}$.
		We show that $q_i \leq (q'')^{2^{x_1 - i}} + 1$ for all $1 \leq i \leq x_1$.
		Clearly, $q_{x_1} = q'' \leq (q'')^{2^0} + 1$.
		By definition, for an arbitrary $1 \leq i \leq x_1$ we have
		\begin{align*}
			q_i & = (q_{i+1} - 1)^2 + 1
			\\[0em]  & \leq ((q'')^{2^{x_1 - i - 1}} + 1 - 1)^2 + 1
			\\[0em]  & = (q'')^{2^{x_1 - i}} + 1.
		\end{align*}
		Hence, $q \geq q_1$.
		
		We construct a set $\mathcal{Q}'' \subseteq \mathcal{Q}'$ as follows.
		Let $\Set{P^1_1, P^1_2, \dots, P^1_{x_1}} = \mathcal{P}^1$ be an arbitrary ordering of the paths in $\mathcal{P}^1$.
		Set $\mathcal{Q}_1 = \mathcal{Q}'$ and then iterate from $2$ to $x_1$, constructing a set $\mathcal{Q}_i$ of size $q_i$.
		
		On step $i \leq x_1$, consider the ordering $\preceq_{i-1}$ of the paths in $\mathcal{Q}_{i-1}$ according to their occurrence along $P^1_{i-1}$.
		By~\cref{thm:erdos_szekeres}, there is a $\mathcal{Q}_{i} \subseteq \mathcal{Q}_{i-1}$ of order at least $q_i$ such that $P^1_i$ intersects $\mathcal{Q}_i$ in order or in reverse with respect to $\preceq_{i-1}$.
		Since $\mathcal{Q}_i \subseteq \mathcal{Q}_{i-1}$, we have that each $P^1_j \in \mathcal{P}^1$ with $j \leq i$ also intersects $\mathcal{Q}_i$ in order or in reverse with respect to $\preceq_{i-1}$.
		
		After $x_1$ steps, we set $\mathcal{Q}'' \coloneqq \mathcal{Q}_{x_1}$.
		By construction, there is an ordering $\preceq$ of $\mathcal{Q}''$ such that each $P^1_i \in \mathcal{P}^1$ intersects $\mathcal{Q}''$ in order or in reverse with respect to $\preceq$.
		
		By the pigeon-hole principle, there is some $\mathcal{P}^2 \subseteq \mathcal{P}^1$ of order at least $x$ such that every path in $\mathcal{P}^2$ intersects the paths of $\mathcal{Q}''$ in the same order.
		Hence, $\Brace{\mathcal{P}^2, \mathcal{Q}''}$ is an ordered $(x, q'')$-web where $\End{\mathcal{P}^2} \subseteq \mathcal{P}$, satisfying~\cref{item:web-to-folded-or-ordered-web-without-gaps:ordered}.
	\end{proof}

	We consider two cases when constructing a path of well-linked sets or a path of order-linked sets
	from the \emph{first half} of $\mathcal{H}$ in our horizontal web $\Brace{\mathcal{H}^1 \cdot \mathcal{H}^2, \mathcal{V}^1 \cdot \mathcal{V}^2}$.

	In the first case, we obtain a folded ordered web $(\mathcal{V}', \mathcal{H}')$ from \cref{cor:web-to-folded-or-ordered-web-without-gaps}, where $\mathcal{V}' \subseteq \mathcal{V}^1$.
	In this case, we obtain a \emph{forward linkage} $\mathcal{L}$ which is contained in $\mathcal{V}^2$.
	This will allow us to use $\mathcal{V}^1$ to build a back-linkage which is disjoint from $\mathcal{L}$, and then use \cref{lemma:well-linked-poss-with-disjoint-forward-linkage-implies-coss} to obtain a cycle of well-linked sets.

	In the second case, we obtain an ordered web $(\mathcal{H}', \mathcal{V}')$ which ends on $\End{\mathcal{H}^1}$.
	We can then use the subpaths of $\mathcal{H}'$ after the last path of $\mathcal{V}'$ to construct a
	linkage to $\mathcal{H}^2$ which does not intersect the path of order-linked sets constructed from the ordered web.
	We will then need to use $\mathcal{V}^1$ in order to construct the back-linkage.
	As $\mathcal{V}^1$ may intersect our path of order-linked sets, we want to be able to restrict how these intersections happen.
	The idea will be to use paths of $\mathcal{V}^2$ leading to each cluster in order to
	obtain a large linkage which contradicts the assumption that $\mathcal{H}$ is weakly minimal with respect to $\mathcal{V}$.

	\Cref{lemma:1-horizontal-web-double-side-linkage} below provides the base cases for our argument of \cref{lemma:coss-inside-2-horizontal-web} later.
	The case \cref{item:1-horizontal-web-double-side-linkage:order-linked} below follows immediately from \cref{cor:web-to-folded-or-ordered-web-without-gaps}\cref{item:web-to-folded-or-ordered-web-without-gaps:folded}.
	For \cref{item:1-horizontal-web-double-side-linkage:double-order}, we will need to construct a routing temporal digraph and use the framework from \cref{sec:powls-framework} in order
	to construct the path of order-linked sets with the desired properties.
	
	We require the following lemma for our proof.

	\begin{lemma}[{\cite[Lemma 9.7]{COSSI}}]
		\label{lemma:folded-web-to-pows}
		There exist two functions \boundDef{lemma:folded-web-to-pows}{h}{w} \(\bound{lemma:folded-web-to-pows}{h}{w} \in \Oh(w^{11})\) and \boundDef{lemma:folded-web-to-pows}{v}{w, \ell} \(\bound{lemma:folded-web-to-pows}{v}{w,\ell} \in \PowerTower{1}{\Polynomial{2}{w, \ell}}\) such that every folded ordered $(h,v)$-web $\Brace{\mathcal{H}, \mathcal{V}}$ with $h \geq \bound{lemma:folded-web-to-pows}{h}{w}$ and $v \geq \bound{lemma:folded-web-to-pows}{v}{w, \ell}$ contains a path of well-linked sets $\Brace{\mathcal{S} = \Brace{ S_0, S_1, \dots, S_{\ell}}, \mathscr{P}}$ of width $w$ and length $\ell$.
		Additionally, there is a $\Start{\mathcal{H}}$-$\End{\mathcal{H}}$-linkage $\mathcal{L} = \mathcal{L}_1 \cdot \mathcal{L}_2 \cdot \mathcal{L}_3$ using only arcs of $\mathcal{H}$ such that $\mathcal{L}_2$ is an $A(S_0)$-$B(S_\ell)$-linkage of order $w$ inside $\Brace{\mathcal{S}, \mathscr{P}}$ and $\mathcal{L}_1$ and $\mathcal{L}_3$ are internally disjoint from $\Brace{\mathcal{S}, \mathscr{P}}$.
	\end{lemma}

	We now define 
	\begin{align*}
		\boundDefAlign{lemma:1-horizontal-web-double-side-linkage}{h}{w_2}
		\bound{lemma:1-horizontal-web-double-side-linkage}{h}{w_2} & \coloneqq 
		\bound{cor:web-to-folded-or-ordered-web-without-gaps}{p}{(w_{2})^{2} - 1},
		\\[0em]
		\Fkt{v'}{w_2, \ell_2} & \coloneqq 
		\bound{cor:web-to-folded-or-ordered-web-without-gaps}{q}{}((( w_{2} \ell_{2} - 1) {\binom{(w_{2})^{2} - 1}{w_{2}}} w_{2}! + 1)\\[0em]
		& \hspace{1.5cm} {} \cdot \bound{theorem:one-way connected temporal digraph contains P_k routing}{\Lifetime{}}{w_{2},(w_{2})^{2} - 1 },(w_{2})^{2} - 1 ),
		\\[0em]
		\boundDefAlign{lemma:1-horizontal-web-double-side-linkage}{v}{w_1, \ell_1, w_2, \ell_2, c }
		\bound{lemma:1-horizontal-web-double-side-linkage}{v}{w_1, \ell_1, w_2, \ell_2, c } & \coloneqq 
		\bound{cor:web-to-folded-or-ordered-web-without-gaps}{q'}{\bound{lemma:folded-web-to-pows}{h}{w_{1} } + \Fkt{v'}{w_{2},\ell_{2} },c,(w_{2})^{2} - 1,\bound{lemma:folded-web-to-pows}{v}{w_{1},\ell_{1} } }.
	\end{align*}
	Note that
	\(\bound{lemma:1-horizontal-web-double-side-linkage}{h}{w_2} \in \Oh((w_{2})^{2})\) and
	\(\bound{lemma:1-horizontal-web-double-side-linkage}{v}{w_1,\ell_1,w_2,\ell_2,c} \in \PowerTower{5}{\Polynomial{15}{w_1,\ell_1,w_2,\ell_2,c}}\).

	The case \cref{item:1-horizontal-web-double-side-linkage:order-linked} of \cref{lemma:1-horizontal-web-double-side-linkage} is illustrated in
	\cref{fig:1-horizontal-web-double-side-linkage:order-linked}.
	In this case, we obtain the clusters of the path of well-linked sets from the folded ordered web of \cref{cor:web-to-folded-or-ordered-web-without-gaps}\cref{item:web-to-folded-or-ordered-web-without-gaps:folded}.
	\begin{figure}[H]
		\centering
		\begin{tikzpicture}[yscale=0.6,xscale=0.7]
			 \node[vertex, fill = red]
	(v1) at (2.3, 3.45){};
\node[vertex, fill = blue]
	(v3) at (4.6, 5.75){};
\node[vertex, fill = blue]
	(v4) at (4.6, 0){};
\node[vertex, fill = blue]
	(v5) at (6.9, 5.75){};
\node[vertex, fill = blue]
	(v6) at (6.9, 0){};
\node[vertex, fill = blue]
	(v7) at (9.2, 0){};
\node[vertex, fill = blue]
	(v8) at (11.5, 0){};
\node[vertex, fill = blue]
	(v9) at (13.8, 0){};
\node[vertex, fill = blue]
	(v10) at (2.3, 0){};
\node[vertex, fill = blue]
	(v11) at (2.3, 5.75){};
\node[vertex, fill = yellow]
	(v12) at (0, 3.45){};
\node[vertex, fill = yellow]
	(v13) at (0, 2.3){};
\node[vertex, fill = yellow]
	(v14) at (0, 1.15){};
\node[vertex, fill = green]
	(v15) at (16.1, 3.45){};
\node[vertex, fill = green]
	(v16) at (16.1, 2.3){};
\node[vertex, fill = green]
	(v17) at (16.1, 1.15){};
\node[vertex, fill = red]
	(v19) at (13.8, 2.3){};
\node[vertex, fill = red]
	(v20) at (11.5, 2.3){};
\node[vertex, fill = red]
	(v21) at (9.2, 2.3){};
\node[vertex, fill = red]
	(v22) at (6.9, 2.3){};
\node[vertex, fill = red]
	(v23) at (4.6, 2.3){};
\node[vertex, fill = red]
	(v24) at (2.3, 2.3){};
\node[vertex, fill = blue]
	(v25) at (9.2, 5.75){};
\node[vertex, fill = blue]
	(v26) at (11.5, 5.75){};
\node[vertex, fill = blue]
	(v27) at (13.8, 5.75){};
\node[vertex, fill = red]
	(v28) at (2.3, 1.15){};
\node[vertex, fill = red]
	(v29) at (4.6, 1.15){};
\node[vertex, fill = red]
	(v30) at (4.6, 3.45){};
\node[]
	(v31) at (3.45, 5.75) [align = left]{{$S_0$}};
\node[]
	(v32) at (1.15, 4.6) [align = left]{\color{yellow}{$\mathcal{L}_1$}};
\node[]
	(v33) at (14.95, 4.6) [align = left]{\color{green}{$\mathcal{L}_3$}};
\node[]
	(v34) at (8.05, 4.6) [align = left]{\color{red}{$\mathcal{L}_2$}};
\node[line width = 0]
	(v35) at (5.75, 5.75) [align = left]{{$\mathcal{P}_0$}};
\node[line width = 0, draw = none]
	(v36) at (10.35, 5.75) [align = left]{{$\mathcal{P}_1$}};
\node[vertex, fill = red]
	(v37) at (6.9, 3.45){};
\node[vertex, fill = red]
	(v38) at (6.9, 1.15){};
\node[vertex, fill = red]
	(v39) at (9.2, 1.15){};
\node[vertex, fill = red]
	(v40) at (9.2, 3.45){};
\node[]
	(v41) at (8.05, 5.75) [align = left]{{$S_1$}};
\node[vertex, fill = red]
	(v42) at (11.5, 3.45){};
\node[vertex, fill = red]
	(v43) at (11.5, 1.15){};
\node[vertex, fill = red]
	(v44) at (13.8, 1.15){};
\node[vertex, fill = red]
	(v45) at (13.8, 3.45){};
\node[]
	(v46) at (12.65, 5.75) [align = left]{{$S_2$}};
\path[-latex, line width = 1.2, draw = red]
	(v29) to (v38);
\path[-latex, line width = 1.2, draw = red]
	(v40) to (v42);
\path[-latex, line width = 1.2, draw = red]
	(v39) to (v43);
\path[-latex, line width = 1.2, draw = red]
	(v30) to (v37);
\path[latex-, line width = 1.2, draw = blue, dashed]
	(v30) to (v23);
\path[latex-, line width = 1.2, draw = blue, dashed]
	(v23) to (v29);
\path[-latex, line width = 1.2, draw = blue, dashed]
	(v38) to (v22);
\path[-latex, line width = 1.2, draw = blue, dashed]
	(v22) to (v37);
\path[-latex, line width = 1.2, draw = blue, dashed]
	(v40) to (v21);
\path[-latex, line width = 1.2, draw = blue, dashed]
	(v21) to (v39);
\path[-latex, line width = 1.2, draw = red]
	(v1) to (v30);
\path[-latex, line width = 1.2, draw = red]
	(v37) to (v40);
\path[-latex, line width = 1.2, draw = red]
	(v42) to (v45);
\path[-latex, line width = 1.2, draw = red]
	(v20) to (v19);
\path[-latex, line width = 1.2, draw = red]
	(v43) to (v44);
\path[-latex, line width = 1.2, draw = red]
	(v21) to (v20);
\path[-latex, line width = 1.2, draw = red]
	(v38) to (v39);
\path[-latex, line width = 1.2, draw = red]
	(v22) to (v21);
\path[-latex, line width = 1.2, draw = red]
	(v23) to (v22);
\path[-latex, line width = 1.2, draw = red]
	(v24) to (v23);
\path[-latex, line width = 1.2, draw = red]
	(v28) to (v29);
\path[-latex, line width = 1.2, draw = blue, dashed]
	(v11) to (v1);
\path[-latex, line width = 1.2, draw = blue, dashed]
	(v1) to (v24);
\path[-latex, line width = 1.2, draw = blue, dashed]
	(v24) to (v28);
\path[-latex, line width = 1.2, draw = blue, dashed]
	(v30) to (v3);
\path[-latex, line width = 1.2, draw = blue, dashed]
	(v10) to (v4);
\path[-latex, line width = 1.2, draw = blue, dashed]
	(v28) to (v10);
\path[-latex, line width = 1.2, draw = blue, dashed]
	(v4) to (v29);
\path[-latex, line width = 1.2, draw = yellow]
	(v12) to (v1);
\path[-latex, line width = 1.2, draw = yellow]
	(v13) to (v24);
\path[-latex, line width = 1.2, draw = yellow]
	(v14) to (v28);
\path[-latex, line width = 1.2, draw = blue, dashed]
	(v39) to (v7);
\path[-latex, line width = 1.2, draw = blue, dashed]
	(v7) to (v6);
\path[-latex, line width = 1.2, draw = blue, dashed]
	(v6) to (v38);
\path[-latex, line width = 1.2, draw = blue, dashed]
	(v37) to (v5);
\path[-latex, line width = 1.2, draw = blue, dashed]
	(v25) to (v40);
\path[-latex, line width = 1.2, draw = blue, dashed]
	(v43) to (v8);
\path[-latex, line width = 1.2, draw = blue, dashed]
	(v8) to (v9);
\path[-latex, line width = 1.2, draw = blue, dashed]
	(v9) to (v44);
\path[-latex, line width = 1.2, draw = blue, dashed]
	(v26) to (v42);
\path[-latex, line width = 1.2, draw = blue, dashed]
	(v42) to (v20);
\path[-latex, line width = 1.2, draw = blue, dashed]
	(v20) to (v43);
\path[-latex, line width = 1.2, draw = blue, dashed]
	(v44) to (v19);
\path[-latex, line width = 1.2, draw = blue, dashed]
	(v19) to (v45);
\path[-latex, line width = 1.2, draw = blue, dashed]
	(v45) to (v27);
\path[-latex, line width = 1.2, draw = green]
	(v45) to (v15);
\path[-latex, line width = 1.2, draw = green]
	(v19) to (v16);
\path[-latex, line width = 1.2, draw = green]
	(v44) to (v17);

		\end{tikzpicture}
			\caption{Illustration of \cref{lemma:1-horizontal-web-double-side-linkage}\cref{item:1-horizontal-web-double-side-linkage:order-linked}.
			The clusters in blue (dashed) come from a folded ordered web.}
			\label{fig:1-horizontal-web-double-side-linkage:order-linked}
	\end{figure}
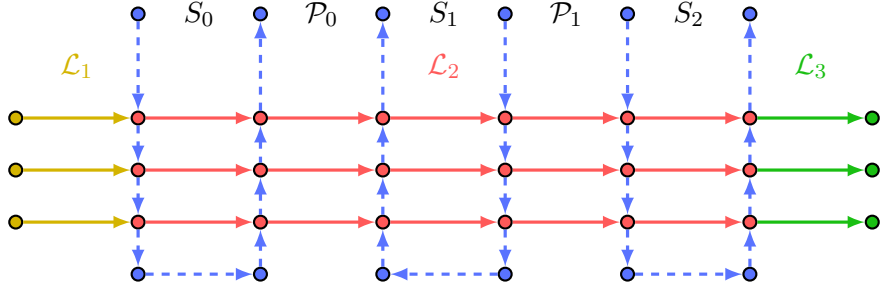

	In \cref{fig:1-horizontal-web-double-side-linkage:double-order} we illustrate
	the case \cref{item:1-horizontal-web-double-side-linkage:double-order} of \cref{lemma:1-horizontal-web-double-side-linkage}.
	In this case, we obtain the clusters of the path of order-linked sets from the ordered web of \cref{cor:web-to-folded-or-ordered-web-without-gaps}\cref{item:web-to-folded-or-ordered-web-without-gaps:ordered},
	and so each path of $\mathcal{V}$ used intersects at most one cluster.
	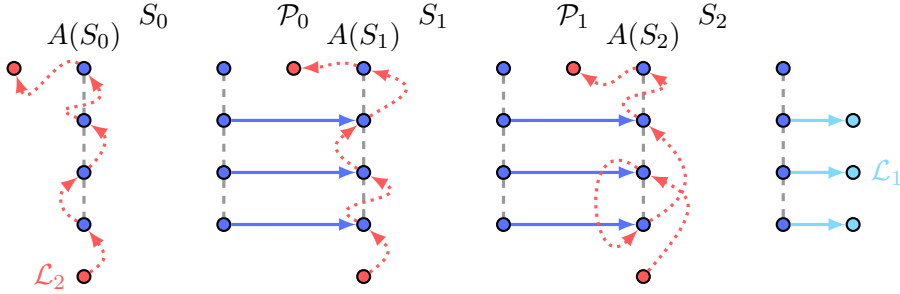
\begin{figure}[H]
		\centering
		\begin{tikzpicture}[yscale=0.6,xscale=0.8]
			 \node[vertex, fill = blue]
	(v1) at (1.15, 3.45){};
\node[vertex, fill = red]
	(v2) at (4.6, 4.6){};
\node[vertex, fill = blue]
	(v3) at (3.45, 4.6){};
\node[vertex, label = above:{{$A(S_1)$}}, fill = blue]
	(v4) at (5.75, 4.6){};
\node[vertex, label = above:{{$A(S_0)$}}, fill = blue]
	(v5) at (1.15, 4.6){};
\node[vertex, label = left:{\color{red}{$\mathcal{L}_2$}}, fill = red]
	(v6) at (1.15, 0){};
\node[vertex, fill = red]
	(v7) at (0, 4.6){};
\node[vertex, fill = red]
	(v8) at (5.75, 0){};
\node[vertex, fill = red]
	(v9) at (10.35, 0){};
\node[vertex, fill = cyan]
	(v10) at (13.8, 3.45){};
\node[vertex, label = right:{\color{cyan}{$\mathcal{L}_1$}}, fill = cyan]
	(v11) at (13.8, 2.3){};
\node[vertex, fill = cyan]
	(v12) at (13.8, 1.15){};
\node[vertex, fill = red]
	(v13) at (9.2, 4.6){};
\node[vertex, fill = blue]
	(v14) at (12.65, 2.3){};
\node[vertex, fill = blue]
	(v15) at (10.35, 2.3){};
\node[vertex, fill = blue]
	(v16) at (8.05, 2.3){};
\node[vertex, fill = blue]
	(v17) at (5.75, 2.3){};
\node[vertex, fill = blue]
	(v18) at (3.45, 2.3){};
\node[vertex, fill = blue]
	(v19) at (1.15, 2.3){};
\node[vertex, fill = blue]
	(v20) at (8.05, 4.6){};
\node[vertex, label = above:{{$A(S_2)$}}, fill = blue]
	(v21) at (10.35, 4.6){};
\node[vertex, fill = blue]
	(v22) at (12.65, 4.6){};
\node[vertex, fill = blue]
	(v23) at (1.15, 1.15){};
\node[vertex, fill = blue]
	(v24) at (3.45, 1.15){};
\node[vertex, fill = blue]
	(v25) at (3.45, 3.45){};
\node[]
	(v26) at (2.3, 5.75) [align = left]{{$S_0$}};
\node[]
	(v32) at (4.6, 5.75) [align = left]{{$\mathcal{P}_0$}};
\node[]
	(v33) at (9.2, 5.75) [align = left]{{$\mathcal{P}_1$}};
\node[vertex, fill = blue]
	(v34) at (5.75, 3.45){};
\node[vertex, fill = blue]
	(v35) at (5.75, 1.15){};
\node[vertex, fill = blue]
	(v36) at (8.05, 1.15){};
\node[vertex, fill = blue]
	(v37) at (8.05, 3.45){};
\node[]
	(v38) at (6.9, 5.75) [align = left]{{$S_1$}};
\node[vertex, fill = blue]
	(v39) at (10.35, 3.45){};
\node[vertex, fill = blue]
	(v40) at (10.35, 1.15){};
\node[vertex, fill = blue]
	(v41) at (12.65, 1.15){};
\node[vertex, fill = blue]
	(v42) at (12.65, 3.45){};
\node[]
	(v43) at (11.5, 5.75) [align = left]{{$S_2$}};
\path[-latex, line width = 1.2, draw = blue]
	(v24) to (v35);
\path[-latex, line width = 1.2, draw = blue]
	(v37) to (v39);
\path[-latex, line width = 1.2, draw = blue]
	(v36) to (v40);
\path[-latex, line width = 1.2, draw = blue]
	(v25) to (v34);
\path[line width = 1.2, draw = gray, dashed]
	(v25) to (v18);
\path[line width = 1.2, draw = gray, dashed]
	(v18) to (v24);
\path[line width = 1.2, draw = gray, dashed]
	(v35) to (v17);
\path[line width = 1.2, draw = gray, dashed]
	(v17) to (v34);
\path[line width = 1.2, draw = gray, dashed]
	(v37) to (v16);
\path[line width = 1.2, draw = gray, dashed]
	(v16) to (v36);
\path[-latex, line width = 1.2, draw = blue]
	(v16) to (v15);
\path[-latex, line width = 1.2, draw = blue]
	(v18) to (v17);
\path[line width = 1.2, draw = gray, dashed]
	(v5) to (v1);
\path[line width = 1.2, draw = gray, dashed]
	(v1) to (v19);
\path[line width = 1.2, draw = gray, dashed]
	(v19) to (v23);
\path[line width = 1.2, draw = gray, dashed]
	(v25) to (v3);
\path[line width = 1.2, draw = gray, dashed]
	(v34) to (v4);
\path[line width = 1.2, draw = gray, dashed]
	(v20) to (v37);
\path[line width = 1.2, draw = gray, dashed]
	(v21) to (v39);
\path[line width = 1.2, draw = gray, dashed]
	(v39) to (v15);
\path[line width = 1.2, draw = gray, dashed]
	(v15) to (v40);
\path[line width = 1.2, draw = gray, dashed]
	(v41) to (v14);
\path[line width = 1.2, draw = gray, dashed]
	(v14) to (v42);
\path[line width = 1.2, draw = gray, dashed]
	(v42) to (v22);
\path[-latex, line width = 1.2, draw = cyan]
	(v42) to (v10);
\path[-latex, line width = 1.2, draw = cyan]
	(v14) to (v11);
\path[-latex, line width = 1.2, draw = cyan]
	(v41) to (v12);
\path[-latex, line width = 1.2, draw = red, dotted]
	(v6) .. controls (1.494, 0.326) and (1.56, 0.423) .. (v23);
\path[-latex, line width = 1.2, draw = red, dotted]
	(v23) .. controls (0.693, 1.452) and (0.729, 1.868) .. (v19);
\path[-latex, line width = 1.2, draw = red, dotted]
	(v19) .. controls (1.657, 2.962) and (1.442, 3.193) .. (v1);
\path[-latex, line width = 1.2, draw = red, dotted]
	(v1) .. controls (0.334, 3.874) and (1.688, 3.6) .. (v5);
\path[-latex, line width = 1.2, draw = red, dotted]
	(v5) .. controls (0.744, 5.045) and (0.242, 3.612) .. (v7);
\path[-latex, line width = 1.2, draw = red, dotted]
	(v8) .. controls (6.344, 0.556) and (6.102, 0.709) .. (v35);
\path[-latex, line width = 1.2, draw = red, dotted]
	(v35) .. controls (4.928, 1.702) and (6.657, 1.637) .. (v17);
\path[-latex, line width = 1.2, draw = red, dotted]
	(v17) .. controls (5.095, 2.829) and (5.324, 3.014) .. (v34);
\path[-latex, line width = 1.2, draw = red, dotted]
	(v34) .. controls (6.228, 3.873) and (6.847, 4.158) .. (v4);
\path[-latex, line width = 1.2, draw = red, dotted]
	(v4) .. controls (5.35, 4.805) and (5.14, 4.63) .. (v2);
\path[-latex, line width = 1.2, draw = red, dotted]
	(v9) .. controls (11.29, 1.432) and (11.231, 1.763) .. (v15);
\path[-latex, line width = 1.2, draw = red, dotted]
	(v15) .. controls (9.268, 3.359) and (9.466, -0.047) .. (v40);
\path[-latex, line width = 1.2, draw = red, dotted]
	(v40) .. controls (11.279, 1.842) and (11.14, 2.456) .. (v39);
\path[-latex, line width = 1.2, draw = red, dotted]
	(v39) .. controls (9.463, 4.128) and (11.026, 3.995) .. (v21);
\path[-latex, line width = 1.2, draw = red, dotted]
	(v21) .. controls (10.005, 4.611) and (9.845, 4.034) .. (v13);

		\end{tikzpicture}
			\caption{Illustration of \cref{lemma:1-horizontal-web-double-side-linkage}\cref{item:1-horizontal-web-double-side-linkage:double-order}.
			The clusters come from an ordered web, which is then used to prove the red (dotted) linkage intersects
			the $A$ sets.}
			\label{fig:1-horizontal-web-double-side-linkage:double-order}
	\end{figure}

	\begin{lemma}
		\label{lemma:1-horizontal-web-double-side-linkage}
		Let $\Brace{\mathcal{H}, \mathcal{V}}$ be an $(h,v)$-web where $\mathcal{H}$ is weakly $c$-minimal with respect to $\mathcal{V}$.
		If $\Abs{\mathcal{H}} \geq \bound{lemma:1-horizontal-web-double-side-linkage}{h}{w_2}$ and $\Abs{\mathcal{V}} \geq \bound{lemma:1-horizontal-web-double-side-linkage}{v}{w_1, \ell_1, w_2, \ell_2, c}$, then one of the following is true:
		\begin{enamerate}{E}{item:1-horizontal-web-double-side-linkage:double-order}
			\item
			\label{item:1-horizontal-web-double-side-linkage:order-linked}
                $\Brace{\mathcal{H}, \mathcal{V}}$ contains a path of well-linked sets $\Brace{\mathcal{S} = \Brace{ S_0, S_1, \dots, S_{\ell_1}}, \mathscr{P}}$ of width $w_1$ and length $\ell_1$.
                Additionally, there is a $\Start{\mathcal{V}}$-$\End{\mathcal{V}}$-linkage $\mathcal{L} = \mathcal{L}_1 \cdot \mathcal{L}_2 \cdot \mathcal{L}_3$ of order $w_1$ using only arcs of $\mathcal{V}$ such that $\mathcal{L}_2$ is an $A(S_0)$-$B(S_{\ell_1})$-linkage and both $\mathcal{L}_1$ and $\mathcal{L}_3$ are internally disjoint from $\Brace{\mathcal{S}, \mathscr{P}}$.
			\item
			\label{item:1-horizontal-web-double-side-linkage:double-order}
                There is some $\mathcal{V}' \subseteq \mathcal{V}$ such that $\Brace{\mathcal{H}, \mathcal{V}'}$ contains a uniform path of $w_2$-order-linked sets $\Brace{\mathcal{S} = \Brace{ S_0, S_1, \dots, S_{\ell_2}},	\mathscr{P}}$ of width $w_2$ and length $\ell_2$ for which there are linkages $\mathcal{L}_1, \mathcal{L}_2$ such that
    			\begin{enamerate}[resume]{Y}{subitem:1-horizontal-web-double-side-linkage:L2b}
    				\item \label{subitem:1-horizontal-web-double-side-linkage:L1b}
                        $\mathcal{L}_1$ is a $B(S_{\ell_2})$-$\End{\mathcal{H}}$-linkage of order $w_2$ inside $\mathcal{H}$ which is internally disjoint from $\mathcal{V}'$ and from $\Brace{\mathcal{S}, \mathscr{P}}$, and
    				\item \label{subitem:1-horizontal-web-double-side-linkage:L2b}
                        $\mathcal{L}_2 \subseteq \mathcal{V}'$ is a linkage of order $\ell_2 + 1$ where for each $L_{2,j} \in \mathcal{L}_2$ there is some $0 \leq i \leq \ell_2$ such that $A(S_i) \subseteq \V{L_{2,j}}$ and $\V{L_{2,j}} \cap \V{\mathcal{S}, \mathscr{P}} \subseteq \V{S_i}$.
    			\end{enamerate}
		\end{enamerate}
	\end{lemma}
	\begin{proof}
		We define
		$h_2 = (w_2)^2 - 1$,
		$h_1 = \bound{lemma:folded-web-to-pows}{v}{w_1, \ell_1}$,
		$\ell_4 = w_2 \ell_2$,
		$\ell_3 = \bound{theorem:one-way connected temporal digraph contains P_k routing}{\Lifetime{}}{w_2, h_2}$,
		$t_1 = (\ell_4 - 1)\binom{h_2}{w_2} w_2! + 1$,
		$v_2 = t_1 \ell_3$,
		$v_1 = \bound{lemma:folded-web-to-pows}{h}{w_1} + \bound{cor:web-to-folded-or-ordered-web-without-gaps}{q}{v_2}$.
		
		By~\cref{cor:web-to-folded-or-ordered-web-without-gaps}, there is some $\mathcal{H}^1 \subseteq \mathcal{H}$ such that one of the following cases hold:
		
		\textbf{Case 1:}
		\Cref{item:web-to-folded-or-ordered-web-without-gaps:folded} holds.
		
		That is, there is a sublinkage $\mathcal{V}^1 \subseteq \mathcal{V}$ for which $\Brace{\mathcal{V}^1, \mathcal{H}^1}$ is a folded ordered $\Brace{v_1, h_1}$-web.
		As $v_1 \geq \bound{lemma:folded-web-to-pows}{h}{w_1}$, by~\cref{lemma:folded-web-to-pows} $\Brace{\mathcal{V}^1, \mathcal{H}^1}$ contains a path of well-linked sets $\Brace{\mathcal{S},\mathscr{P}}$ of width $w_1$ and length $\ell_1$.
		Additionally, there is a $\Start{\mathcal{V}^1}$-$\End{\mathcal{V}^1}$-linkage $\mathcal{L} = \mathcal{L}_1 \cdot \mathcal{L}_2 \cdot \mathcal{L}_3$ of order $w_1$ using only arcs of $\mathcal{V}^1$ such that $\mathcal{L}_2$ is an $A(S_0)$-$B(S_{\ell_1})$-linkage and both $\mathcal{L}_1$ and $\mathcal{L}_3$ are internally disjoint from $\Brace{\mathcal{S}, \mathscr{P}}$.
		This immediately satisfies~\cref{item:1-horizontal-web-double-side-linkage:order-linked}.
		
		\textbf{Case 2:}
		\Cref{item:web-to-folded-or-ordered-web-without-gaps:ordered} holds.
		
		That is, $\Brace{\mathcal{H}^1, \mathcal{V}}$ contains an ordered $\Brace{h_2, v_2}$-web $\Brace{\mathcal{H}^2, \mathcal{V}^1}$ such that $\End{\mathcal{H}^2} \subseteq \End{\mathcal{H}^1}$.
		
		Decompose $\mathcal{H}^2$ into $\mathcal{H}^2 = \mathcal{Q} \cdot \mathcal{L}_1'$ such that $\mathcal{L}_1'$ is internally disjoint from $\mathcal{V}^1$ and $\Start{\mathcal{L}_1'} \subseteq \V{\mathcal{V}^1}$.
		Since $\mathcal{L}_1'$ is internally disjoint from $\mathcal{V}^1$, we have that $\Brace{\mathcal{Q}, \mathcal{V}^1}$ is also an ordered $(h_2, v_2)$-web.
		
		Let $\Brace{V^1_1, V^1_2, \ldots, V^1_{v_2}} \coloneqq \mathcal{V}^1$ be an ordering of $\mathcal{V}^1$ witnessing that $\Brace{\mathcal{Q}, \mathcal{V}^1}$ is an ordered web.
		For each $1 \leq i \leq t_1$ let $T_i$ be the routing temporal digraph of $\mathcal{Q}$ through $G_i \coloneqq (V^1_{(i-1) \ell_3 + 1 }, V^1_{(i-1) \ell_3 + 2}$, $\ldots, V^1_{i \ell_3})$.
		As every path in $\mathcal{V}^1$ intersects every path in $\mathcal{Q}$, we have that $\Layer{T_i}{j}$ is unilateral for all $1 \leq i \leq t_1$ and all $1 \leq j \leq \ell_3$.
		Since $\Lifetime{T_i} = \ell_3$, by~\cref{theorem:one-way connected temporal digraph contains P_k routing} we have that every $T_i$ contains a $\Pk{w_2}$-routing $\varphi_i$ over some $\mathcal{Q}_i \subseteq \mathcal{Q}$.
		
		As there are at most $\binom{h_2}{w_2} w_2!$ distinct $\varphi_i$, by the pigeon-hole principle there is some $\mathcal{I} \subseteq \Set{1, \ldots, t_1}$ of size $\ell_4$ such that $\varphi \coloneqq \varphi_i = \varphi_j$ and $\mathcal{Q}' \coloneqq \mathcal{Q}_i = \mathcal{Q}_j$ hold for all $i, j \in \mathcal{I}$.
		
		For each $i \in \mathcal{I}$ let $\mathcal{R}_i$ be the maximal $\V{V^1_{(i - 1)\ell_3 + 1}}$-$\V{V^1_{i\ell_3}}$-linkage of order $w_2$ inside $\mathcal{Q}'$ and let $T_i'$ be the routing temporal digraph of $\mathcal{R}_i$ through $G_i$.
		Note that $\varphi$ induces a $\Pk{w_2}$-routing $\psi$ over $\mathcal{R}_i$ in $T_i'$.
		
		By~\cref{lemma:P_k-routing-implies-1-order-linked} we have that $\Start{\mathcal{R}_i}$ is 1-order-linked to $\End{\mathcal{R}_{i}}$ inside $\ToDigraph{\mathcal{R}_i} \cup \ToDigraph{G_i}$ for every $i \in \mathcal{I}$.
		For every two consecutive $i < j \in \mathcal{I}$ (that is, there is no $k \in \mathcal{I}$ with $i < k < j$) let $\mathcal{R}_i'$ be the $\End{\mathcal{R}_i}$-$\Start{\mathcal{R}_j}$-linkage of order $w_2$ in $\mathcal{Q}'$.
		We define
		\begin{align*}
			\mathscr{P} & \coloneqq \Brace{\mathcal{R}_{i}' \mid i \in \mathcal{I} \setminus \max(\mathcal{I})},\\[0em]
			\Brace{ \mathcal{P}_0, \mathcal{P}_1, \dots, \mathcal{P}_{w_2\ell_2 - 1}} & \coloneqq \mathscr{P}, \\[0em]
			\mathcal{S} & \coloneqq \Brace{\ToDigraph{\mathcal{R}_i} \cup \ToDigraph{G_i} \mid i \in \mathcal{I}},
			\\[0em]
			\Brace{ S_0, S_1, \dots, S_{w_2 \ell_2}} & \coloneqq \mathcal{S},
		\end{align*}
		whereas the order of the elements of $\mathcal{S}$ and $\mathscr{P}$ is given by the order of $i \in \mathcal{I}$.
		Finally, we set $A(S_j) = \Start{\mathcal{R}_i}$ and $B(S_j) = \End{\mathcal{R}_i}$ for every $0 \leq j \leq \ell_2$, where $\mathcal{R}_i$ is the sublinkage of $\mathcal{Q}'$ which is inside $S_j$.
		
		By choice of $\mathcal{R}_i$, $\Brace{\mathcal{S}, \mathscr{P}}$ is a uniform path of 1-order-linked sets of width $w_2$ and length $w_2 \ell_2$.
		By~\cref{lemma:merging path of order-linked sets}, there is a uniform path of $w_2$-order-linked sets \((\mathcal{S}' = \Brace{ S'_0, S'_1, \dots, S'_{\ell_2}}, \mathscr{P}' = \Brace{ \mathcal{P}'_0, \mathcal{P}'_1, \dots, \mathcal{P}'_{\ell_2 - 1}})\) of length $\ell_2$ and width $w_2$ inside $\Brace{\mathcal{S}, \mathscr{P}}$.
		Additionally, for every $0 \leq i \leq \ell_2$ we have $S_i' \subseteq \SubPOSS{(\mathcal{S}, \mathscr{P})}{i w_2}{(i+1) w_2  - 1}$, $A(S_i') \subseteq A(S_{i w_2})$ and $B(S_i') \subseteq B(S_{(i + 1) w_2 - 1})$, and for $0 \leq i < \ell_2$ we have $\mathcal{P}'_i \subseteq \mathcal{P}_{(i + 1)(w_2 - 1)}$.
		
		Let $\mathcal{L}_1 \subseteq \mathcal{L}_1'$ be the paths of $\mathcal{L}_1'$ satisfying $\Start{\mathcal{L}_1} = B(S_{\ell_2}')$.
		Since $\mathcal{L}_1'$ is internally disjoint from $\mathcal{V}^1$ by construction and $\End{\mathcal{L}_1} \subseteq \End{\mathcal{H}}$, we have that $\mathcal{L}_1$ is a $B(S'_{\ell_2})$-$\End{\mathcal{H}}$-linkage of order $w_2$ which is internally disjoint from $\mathcal{V}^1$, satisfying~\cref{subitem:1-horizontal-web-double-side-linkage:L1b}.
		
		For each $0 \leq i \leq \ell_2$ let $L_i \in \mathcal{V}^1$ be the path of $\mathcal{V}^1$ which intersects $A(S_i)$.
		By construction of $\mathcal{S}$, each path in $\mathcal{V}^1$ intersects at most one $A(S_j)$.
		Since $S_i' \subseteq \SubPOSS{(\mathcal{S}, \mathscr{P})}{i w_2}{(i+1) w_2 - 1}$, we have that $L_i$ intersects exactly one $A(S_i')$ as well.
		Let $\mathcal{L}_2 = \Set{L_{i} \mid 0 \leq i \leq \ell_2}$.
		
		By construction we have $\V{\mathcal{V}^1} \cap \V{\mathcal{P}'_i} \subseteq \Start{\mathcal{P}'_i} \cup \End{\mathcal{P}'_i}$ for every $\mathcal{P}_i' \in \mathscr{P}'$.
		Further, $L_{j}$ intersects only the cluster $S_i'$.
		Hence, $\mathcal{L}_2 \subseteq \mathcal{V}^1$ is a linkage of order $\ell_2 + 1$ and for each $L_{j} \in \mathcal{L}_2$ there is some $0 \leq i \leq \ell_2$ such that $A(S_i') \subseteq \V{L_{j}}$ and $\V{L_{j}} \cap \V{\mathcal{S}', \mathscr{P}'} \subseteq \V{S_i'}$, satisfying~\cref{subitem:1-horizontal-web-double-side-linkage:L2b} and so~\cref{item:1-horizontal-web-double-side-linkage:double-order} as well.
	\end{proof}
	
	To build our cycle of well-linked sets, we will first use a construction from \cite[Lemma 6.12]{kawarabayashi2015directed} in order to \emph{reserve} some paths of $\mathcal{H}^1$ in our
	2-horizontal web $(\mathcal{H}^1 \cdot \mathcal{H}^2, \mathcal{V}^1, \mathcal{V}^2)$.
	These paths will later be used to construct a back-linkage for the path of well-linked sets obtained
	from \cref{lemma:1-horizontal-web-double-side-linkage}\cref{item:1-horizontal-web-double-side-linkage:order-linked}.

	Further, we also want some path $H_e \in \mathcal{H}$ which can be used
	when we instead obtain a path of order-linked sets from \cref{lemma:1-horizontal-web-double-side-linkage}\cref{item:1-horizontal-web-double-side-linkage:double-order}.
	We can then use the linkage $\mathcal{L}_2$ from \cref{item:1-horizontal-web-double-side-linkage:double-order} in order to construct a linkage from $H_e$ to the path of order-linked sets.
	Using the clusters, we reach the paths of $\mathcal{V}^1$ which intersect our
	path of order-linked sets.
	In this way, we can obtain a large linkage avoiding an arc of $H_e$ if
	$\mathcal{V}^1$ intersects too many clusters.

	On the other hand, if $\mathcal{V}^1$ is disjoint from most clusters,
	then we can obtain a back-linkage as in the precondition of \cref{lemma:back-linkage avoids many clusters}.
	This can then be used to construct the cycle of well-linked sets as desired.
	
	We define 
	\begin{align*}
		\Fkt{\ell'}{w, c } & \coloneqq 4 w + \left(\ell - 1\right) \left(c + 1\right) - 2,
		\\[0em]
		\Fkt{w'}{w, \ell} & \coloneqq 
		\bound{lemma:well-linked-poss-with-disjoint-forward-linkage-implies-coss}{q}{w,\ell } + 2 \bound{lemma:well-linked-poss-with-disjoint-forward-linkage-implies-coss}{r}{w,\ell },
		\\[0em]
		\Fkt{v'}{w } & \coloneqq 
		\bound{lemma:1-horizontal-web-double-side-linkage}{h}{4 w },
		\\[0em]
		\Fkt{v''}{w, \ell, c } & \coloneqq 
		\bound{lemma:1-horizontal-web-double-side-linkage}{v}{\Fkt{w'}{w,\ell},\bound{lemma:well-linked-poss-with-disjoint-forward-linkage-implies-coss}{\ell'}{w, \ell},\bound{lemma:well-linked-poss-with-disjoint-forward-linkage-implies-coss}{r}{w,\ell },\Fkt{\ell'}{w, c },c } - 1,
		\\[0em]
		\Fkt{v'''}{w, \ell, c } & \coloneqq 
		  \Fkt{v''}{w,\ell,c } \left(\Fkt{v'}{w } + 1\right),
		\\[0em]
		\boundDefAlign{lemma:coss-inside-2-horizontal-web}{h}{w, \ell }
		\bound{lemma:coss-inside-2-horizontal-web}{h}{w, \ell } & \coloneqq 
		2 \bound{lemma:well-linked-poss-with-disjoint-forward-linkage-implies-coss}{r}{w,\ell } + \Fkt{v'}{w} + 1,
		\\[0em]
		\boundDefAlign{lemma:coss-inside-2-horizontal-web}{v}{w, \ell, c }
		\bound{lemma:coss-inside-2-horizontal-web}{v}{w, \ell, c } & \coloneqq 
		\Fkt{v'''}{w,\ell,c } {\binom{\bound{lemma:coss-inside-2-horizontal-web}{h}{w,\ell }}{2 \bound{lemma:well-linked-poss-with-disjoint-forward-linkage-implies-coss}{r}{w,\ell }}} + 1 + \bound{lemma:coss-inside-2-horizontal-web}{h}{w, \ell} c.
	\end{align*}
	Observe that
	\(\bound{lemma:coss-inside-2-horizontal-web}{h}{w,\ell} \in \Oh(w^{2} \ell^{2})\) and
	\(\bound{lemma:coss-inside-2-horizontal-web}{v}{w,\ell,c} \in \PowerTower{8}{\Polynomial{25}{w,\ell,c}}\).
	
	\begin{lemma}
		\label{lemma:coss-inside-2-horizontal-web}
		Let $\Brace{\mathcal{H}, \mathcal{V}}$ be a $2$-horizontal web where $\mathcal{H}$ is weakly $c$-minimal with respect to $\mathcal{V}$.
		If $\Abs{\mathcal{H}} \geq \bound{lemma:coss-inside-2-horizontal-web}{h}{w, \ell}$ and $\Abs{\mathcal{V}} \geq \bound{lemma:coss-inside-2-horizontal-web}{v}{w, \ell, c}$, then $\ToDigraph{\Brace{\mathcal{H}, \mathcal{V}}}$ contains a cycle of well-linked sets of length $\ell$ and width $w$.
	\end{lemma}
	\begin{proof}
		We define
    		\begin{gather*}
        		z_1 \coloneqq c + 1, \qquad
        		q_1 \coloneqq \bound{lemma:well-linked-poss-with-disjoint-forward-linkage-implies-coss}{q}{w, \ell}, \qquad
        		\ell_4 \coloneqq \ell - 1, \\[0em]
        		w_4 \coloneqq 2w, \quad
        		w_3 \coloneqq 2w_4, \quad
        		w_2 \coloneqq \bound{lemma:well-linked-poss-with-disjoint-forward-linkage-implies-coss}{r}{w, \ell}, \quad
        		w_1 \coloneqq q_1 + 2w_2, \\[0em]
        		\ell_3 \coloneqq 2(w_4 - 1), \quad
        		\ell_2 \coloneqq \ell_3 + z_1\ell_4, 
        		\ell_1 \coloneqq \bound{lemma:well-linked-poss-with-disjoint-forward-linkage-implies-coss}{\ell'}{w,\ell}, \\[0em]
        		m_1 \coloneqq 2w_2, \quad
        		m_2 \coloneqq \bound{lemma:1-horizontal-web-double-side-linkage}{h}{w_3}, \qquad
        		h_1 \coloneqq m_1 + m_2 + 1, \\[0em]
        		v_2 \coloneqq \bound{lemma:1-horizontal-web-double-side-linkage}{v}{w_1, \ell_1, w_3, \ell_2, c}, \quad
        		v_1 \coloneqq (v_2 - 1) (m_2 + 1) \binom{h_1}{m_1} + 1.
    		\end{gather*}
		Observe that $\bound{lemma:coss-inside-2-horizontal-web}{h}{w, \ell} = m_1 + m_2 + 1$ and $\bound{lemma:coss-inside-2-horizontal-web}{v}{w, \ell, c} = v_1 + \bound{lemma:coss-inside-2-horizontal-web}{h}{w,\ell}c$.

		Decompose $\mathcal{H}$ into $\mathcal{H} = \mathcal{H}^1 \cdot \mathcal{H}^2$ and decompose $\mathcal{V}$ into $\mathcal{V}^1 \cdot \mathcal{V}^2$ such that $\Start{\mathcal{V}^2} \subseteq \V{\mathcal{H}^2}$ and $\mathcal{V}^2$ is internally disjoint from $\mathcal{H}^2$.
		Hence, $\V{\mathcal{V}^2} \cap \V{\mathcal{H}} \subseteq \V{\mathcal{H}^1} \cup \Start{\mathcal{V}^2}$.
		By~\cref{def:horizontal-web}, such a decomposition exists.
		For each $H_i \in \mathcal{H}$ we write $H_i^1$ for the subpath of $H_i$ in $\mathcal{H}^1$ and $H_i^2$ for the subpath of $H_i$ in $\mathcal{H}^2$.
		
		Let $\mathcal{V}' \subseteq \mathcal{V}$ be the paths $V_j \in \mathcal{V}$ for which there is some $H_i \in \mathcal{H}$ such that $V_j$ contains a subpath $V_j'$ with $\Start{V_j'} \in \V{H_i^1}$ and $\End{V_j'} \in \V{H^2_i}$.
		Let $\mathcal{V}^* = \mathcal{V} \setminus \mathcal{V}'$.
		Since $\Brace{\mathcal{H}, \mathcal{V}}$ is $2$-horizontal web where \(\mathcal{H}\) is weakly $c$-minimal with respect to \(\mathcal{V}\), for each $H_i \in \mathcal{H}$ there are at most $c$ paths in $\mathcal{V}'$ which contain a subpath as above.
		Hence, \(\Abs{\mathcal{V}'} \leq \Abs{\mathcal{H}}c\) and thus $\Abs{\mathcal{V}^*} \geq v_1$.
		Further, $\mathcal{V}^*$ satisfies the following by construction.
		
		\begin{enamerate}{V}{claim:V*-no-forward-jumps}
			\item \label{claim:V*-no-forward-jumps}
			Let $Q_i^1 \cdot Q_i^2 \in \mathcal{V}^*$ be an arbitrary decomposition of a path in $\mathcal{V}^*$.
			If $Q_i^2$ intersects some $H^2_j \in \mathcal{H}^2$, then $Q_i^1$ is disjoint from $H^1_j$.
		\end{enamerate}
		
		Let $\mathcal{V}^3 \subseteq \mathcal{V}^1$ and $\mathcal{V}^4 \subseteq \mathcal{V}^2$ be the subpaths of $\mathcal{V}^1$ and $\mathcal{V}^2$ such that $\V{\mathcal{V}^3} \subseteq \V{\mathcal{V}^*}$ and $\V{\mathcal{V}^4} \subseteq \V{\mathcal{V}^*}$.
		Note that $\mathcal{V}^* = \mathcal{V}^3 \cdot \mathcal{V}^4$.
		
		For each subpath $V_i$ of $\mathcal{V}^*$ which contains some path of $\mathcal{V}^4$ as a subpath, define a linear ordering $\preceq_{V_i}$ on $\mathcal{H}^1$ such that $H_a \preceq_{V_i} H_b$ if $V_i$ does not intersect $H_b$ before the first intersection of $V_i$ with $H_a$.
		As every $V^4_i \in \mathcal{V}^4$ intersects every $H_a \in \mathcal{H}^1$, every $\preceq_{V_i}$ is a linear ordering.
		Define $\mathcal{M}(V_i)$ as the set of $m_1 + m_2$ maximal elements of $\preceq_{V_i}$ and $\mathcal{N}(V_i)$ as the set of $m_1$ maximal elements of $\preceq_{V_i}$.

		For each $1 \leq i \leq \Abs{\mathcal{V}^*}$, decompose $V_i \in \mathcal{V}^3 \cdot \mathcal{V}^4$ and construct a set $\mathcal{M}_i \subseteq \mathcal{H}^1$ 		iteratively as follows.
        				Start with the split $V_i = Q^1_i \cdot Q^2_i$ such that $Q^1_i \in \mathcal{V}^3$ and $Q^2_i \in \mathcal{V}^4$.
		Also, set $\mathcal{H}' = \mathcal{H}^2$ and $\mathcal{M}_i = \emptyset$.
		During the construction, we update the values of \(\mathcal{M}_i\), \(\mathcal{H}'\) \(Q^1_i\), \(Q^2_i\) (and two auxiliary paths \(Q^3_i\) and \(Q^4_i\)) multiple times.
		Repeat the following steps until stopping.
                
    \begin{enumerate}			\item Let $H_j \in \mathcal{H}$ be such that $\Start{Q^2_i} \in \V{H^2_j}$.
			\item \label{item:coss-inside-2-horizontal-web:stop}
                If $H^1_j \not\in \mathcal{M}(Q_i^2)$, stop the construction.
			\item Otherwise, set $\mathcal{H}' \coloneqq \mathcal{H}' \setminus \Set{H^2_j}$ and let $Q^3_i \cdot Q^4_i \coloneqq Q^1_i$ such that $\Start{Q^4_i} \subseteq \V{\mathcal{H}'}$ and $Q^4_i$ is internally disjoint from $\mathcal{H}'$.
			\item Set $Q^1_i \coloneqq Q^3_i$, $Q^2_i \coloneqq Q^4_i \cdot Q^2_i$ and $\mathcal{M}_i \coloneqq \mathcal{M}_i \cup \Set{H^1_j}$.
                \label{item:coss-inside-2-horizontal-web:last}
		\end{enumerate}
		
		By~\cref{claim:V*-no-forward-jumps}, whenever we add some $H^1_j$ to $\mathcal{M}_i$ in the construction above, then $H^1_j \in \mathcal{M}(Q^4_i \cdot \mathcal{Q}^2_i)$ as well.
		Hence, $\mathcal{M}_i \subseteq \mathcal{M}(Q^2_i)$.
		The construction above stops at step~\cref{item:coss-inside-2-horizontal-web:stop} for every $V_i \in \mathcal{V}^*$ after at most $\Abs{\mathcal{M}(V_i)}$ iterations because $\Abs{\mathcal{M}(V_i)} \leq \Abs{\mathcal{H}}$, every path in $\mathcal{V}^3$ intersects every path in $\mathcal{H}^2$, and $\Abs{\mathcal{M}_i}$ increases after each iteration.
		
		Since $\Abs{\mathcal{V}^*} \geq (v_2 - 1) \binom{h_1}{m_1} \binom{h_1 - m_1}{m_2} + 1$, by the pigeon-hole principle, there is some $\mathcal{Q}^* \subseteq \mathcal{V}^*$ of order $v_2$ such that $\mathcal{M}(Q^*_i) = \mathcal{M}(Q^*_j)$ and $\mathcal{N}' \coloneqq \mathcal{N}(Q^*_i) = \mathcal{N}(Q^*_j)$ for every $Q^*_i, Q^*_j \in \mathcal{Q}^*$.
		We set $\mathcal{M}' = \mathcal{M}(Q^*_i) \setminus \mathcal{N}(Q^*_i)$ for some $Q^*_i \in \mathcal{Q}^*$.
		Decompose $\mathcal{Q}^*$ into $\mathcal{Q}^* = \mathcal{Q}^1 \cdot \mathcal{Q}^2 \cdot \mathcal{Q}^M \cdot \mathcal{Q}^N$ such that $\mathcal{Q}^1 = \Set{Q^1_i \mid Q_i \in \mathcal{Q}^*}$, $\End{\mathcal{Q}^2} \subseteq \V{\mathcal{M}'}$, $\End{\mathcal{Q}^M} \subseteq \V{\mathcal{N}'}$, $\mathcal{Q}^1 \cdot \mathcal{Q}^2$ is internally disjoint from $\mathcal{M}' \cup \mathcal{N}'$, and $\mathcal{Q}^M$ is internally disjoint from $\mathcal{N}'$.
		By choice of $\mathcal{M}'$ and $\mathcal{N}'$, such a decomposition exists.
		
												By construction of $\mathcal{H}'$, for each $Q^1_i \in \mathcal{Q}^1$ we have that $Q^2_i$ intersects $H^1_j$, where $\End{Q^1_i} \subseteq \V{H^2_j}$ and $Q^1_i \cdot Q^2_i \in \mathcal{Q}^1 \cdot \mathcal{Q}^2$.
		Finally, $\Brace{\mathcal{M}', \mathcal{Q}^M}$ is a 1-horizontal web where \(\mathcal{M}'\) is weakly $c$-minimal with respect to \(\mathcal{Q}^M\), $\Abs{\mathcal{M}'} = m_2$ and $\Abs{\mathcal{Q}^M} = v_2$.
		From~\cref{lemma:1-horizontal-web-double-side-linkage}, we obtain two cases.
		
		\noindent\textbf{Case 1:}~\cref{item:1-horizontal-web-double-side-linkage:order-linked} holds.
		
		That is, $\Brace{\mathcal{M}', \mathcal{Q}^M}$ contains a path of well-linked sets $\Brace{\mathcal{S} = \Brace{ S_0, S_1, \dots, S_{\ell_1} }, \mathscr{P}}$ of width $w_1$ and length $\ell_1$.
		Additionally, there is a $\Start{\mathcal{Q}^M}$-$\End{\mathcal{Q}^M}$-linkage $\mathcal{L} = \mathcal{L}_1 \cdot \mathcal{L}_2 \cdot \mathcal{L}_3$ of order $w_1 \geq q_1$ using only arcs of $\mathcal{Q}^M$ such that $\mathcal{L}_2$ is an $A(S_0)$-$B(S_\ell)$-linkage and both $\mathcal{L}_1$ and $\mathcal{L}_3$ are internally disjoint from $\Brace{\mathcal{S}, \mathscr{P}}$.
		
		We construct a $B(S_{\ell_1})$-$A(S_0)$-linkage $\mathcal{R}$ of order $w_2$ which is internally disjoint from $\mathcal{L}_2$ as follows.
		Take an $\End{\mathcal{L}_2}$-$\End{\mathcal{N}'}$-linkage $\mathcal{X}_1$ in $\ToDigraph{\mathcal{N}' \cup \mathcal{Q}^N \cup \mathcal{L}_3}$.
		Take an $\End{\mathcal{X}_1}$-$\Start{\mathcal{L}_1}$-linkage $\mathcal{X}_2$ in $\ToDigraph{\mathcal{H}^2 \cup \mathcal{Q}^1}$.
		Since both $\Brace{\mathcal{M}',\mathcal{Q}^M}$ and $\Brace{\mathcal{H}^2, \mathcal{Q}^1}$ are webs and $\End{\mathcal{L}_3} \subseteq \End{\mathcal{Q}^M} = \Start{\mathcal{Q}^N}$, by~\cref{lemma:web_mixed_linkages} the linkages $\mathcal{X}_1$ and $\mathcal{X}_2$ exist.
		
		As~\cref{item:1-horizontal-web-double-side-linkage:order-linked} holds, the linkages $\mathcal{L}_1$, $\mathcal{L}_2$ and $\mathcal{L}_3$ are pairwise internally disjoint.
		Since $\mathcal{L}_2$ is contained in $\Brace{\mathcal{M}', \mathcal{Q}^M}$, we have that $\mathcal{L}_2$ is internally disjoint from $\mathcal{M}'$ and, hence, from $\mathcal{X}_1$.
		Further, as $\mathcal{L}_2$ only uses arcs of $\mathcal{Q}^M$, we have that $\mathcal{L}_2$ is internally disjoint from $\mathcal{X}_2$.
		Hence, $\mathcal{R}' = \mathcal{X}_1 \cdot \mathcal{X}_2 \cdot \mathcal{L}_1'$ is internally disjoint from $\mathcal{L}_2$, where $\mathcal{L}_1' \subseteq \mathcal{L}_1$ are the paths with $\Start{\mathcal{L}_1'} = \End{\mathcal{X}^2}$.
		
		Because $\mathcal{L}_1, \mathcal{L}_2$ and $\mathcal{L}_3$ only use arcs of $\mathcal{Q}^M$ and $\mathcal{Q}^M$ is internally disjoint from $\mathcal{N}'$ and from $\mathcal{H}^2$, we have that $\mathcal{L}_1$ and $\mathcal{L}_2$ are internally disjoint from $\mathcal{X}_1$ and from $\mathcal{X}_2$.
		Hence, $\mathcal{R}'$ is a half-integral $B(S_{\ell_1})$-$A(S_{0})$-linkage, as $\End{\mathcal{L}_1} = \Start{\mathcal{L}_2} \supseteq A(S_0)$ and $\Start{\mathcal{L}_2} = \End{\mathcal{L}_1} \supseteq B(S_{\ell_1})$.
		By~\cref{lemma:half_integral_to_integral_linkage}, there is a $B(S_{\ell_1})$-$A(S_{0})$-linkage $\mathcal{R}''$ of order $w_2$ inside $\ToDigraph{\mathcal{R}'}$.
		Hence, by~\cref{lemma:well-linked-poss-with-disjoint-forward-linkage-implies-coss}, $\ToDigraph{\Brace{\mathcal{S}, \mathscr{P}} \cup \mathcal{R}''}$ contains a cycle of well-linked sets of width $w$ and length $\ell$.
		
		\noindent\textbf{Case 2:}~\cref{item:1-horizontal-web-double-side-linkage:double-order} holds.
		
		That is, there is some $\mathcal{Q}'' \subseteq \mathcal{Q}^M$ such that $\Brace{\mathcal{M}', \mathcal{Q}''}$ contains a uniform path of $w_3$-order-linked sets $\Brace{\mathcal{S} = \Brace{ S_0, S_1, \dots, S_{\ell_2}}, \mathscr{P} = \Brace{\mathcal{P}_0, \mathcal{P}_1, \dots, \mathcal{P}_{\ell_2 - 1}}}$ of width $w_3$ and length $\ell_2$ for which there are linkages $\mathcal{L}_1$ and $\mathcal{L}_2$ satisfying~\cref{subitem:1-horizontal-web-double-side-linkage:L1b} and \cref{subitem:1-horizontal-web-double-side-linkage:L2b}.
		
		Let $\mathcal{L}'_2 \subseteq \mathcal{L}_2$ be the paths of $\mathcal{L}_2$ satisfying $\V{\mathcal{L}'_2} \cap \bigcup_{i=0}^{\ell_3}{\V{S_{2i}}} \neq \emptyset$, let $\mathcal{L}_3'$ be the paths of $\mathcal{Q}^2$ such that $\End{\mathcal{L}_3'} = \Start{\mathcal{L}_2'}$.
		Finally, let $\mathcal{Q}^4 \subseteq \mathcal{Q}^2$ be the paths satisfying $\End{\mathcal{Q}^4} = \Start{\mathcal{L}_2}$ and let $\mathcal{Q}^3 \subseteq \mathcal{Q}^1$ be the paths satisfying $\End{\mathcal{Q}^3} = \Start{\mathcal{L}_3'}$.
		
		\begin{claim}
			\label{claim:Q5-disjoint-from-POWS}
			There are $i,j$ with $j - i > \ell_4$ and $i \geq \ell_3 + 1$ for which some $\mathcal{Q}^5 \subseteq \mathcal{Q}^3$ of order $w_5$ exists such that $\mathcal{Q}^5$ is internally disjoint from $\SubPOSS{\Brace{\mathcal{S}, \mathscr{P}}}{i}{j}$.
		\end{claim}
		\begin{claimproof}
			Assume towards a contradiction that for every $\ell_3 + 1 \leq i \leq j \leq \ell_2$ with $j - i > \ell_4$ and
			every $\mathcal{Q}^5 \subseteq \mathcal{Q}^3$ of order at least $w_5$ there is a path $Q_x^5 \in \mathcal{Q}^5$ which intersects some vertex of $\SubPOSS{\Brace{\mathcal{S}, \mathscr{P}}}{i}{j}$.
			
			For each $1 \leq k \leq z_1$ we construct sets $\mathcal{S}^a_k, \mathcal{S}^b_k \subseteq \mathcal{S}$, $\mathcal{O}^1_k \subseteq \mathcal{Q}^3$ and bijections $f_{a,k} : \mathcal{O}^1_k \rightarrow \mathcal{S}^a_k$ and $f_{b,k} : \mathcal{O}^1_k \rightarrow \mathcal{S}^b_k$ as follows.
			
			Start with empty $\mathcal{S}^a_0, \mathcal{S}^b_0$, $\mathcal{O}^1_0$, $f_{a,0}$ and $f_{b,0}$.
			Iterate from $1$ to $z_1$.
			On step $k$, choose some $Q_j^1 \in \mathcal{Q}^3 \setminus \mathcal{O}^1_{k - 1}$ such that $Q^1_j$ intersects some $S_i \in \mathcal{S}$ in $\SubPOSS{\Brace{\mathcal{S}, \mathscr{P}}}{\ell_3 + 1 + w k}{\allowbreak\ell_3 + 1 + w (k + 1) - 1}$, and then set
			$\mathcal{O}^1_k = \mathcal{O}^1_{k-1} \cup \Set{Q^1_j}$,
			$\mathcal{S}^a_k = \mathcal{S}^a_{k-1} \cup \Set{S_{i-1}}$ and
			$\mathcal{S}^b_k = \mathcal{S}^b_{k-1} \cup \Set{S_i}$.
			Further, define $f_{a,k}$ and $f_{b,k}$ as the functions satisfying
			$f_{a,k}(Q^i_x) = f_{a,k-1}(Q^i_x)$ for all $Q^i_x \in \mathcal{S}^a_{k-1}$,
			$f_{b,k}(Q^i_x) = f_{b,k-1}(Q^i_x)$ for all $Q^i_x \in \mathcal{S}^b_{k-1}$,
			$f_{a,k}(Q^1_j) = S_{i-1}$, and $f_{b,k}(Q^1_j) = S_{i}$.
			
			Because $\Abs{\mathcal{O}^1_{k-1}} = k - 1$, we have $\Abs{\mathcal{Q}^3 \setminus \mathcal{O}^1_{k-1}} \geq w_5$.
			Hence, in every step $k$, there is some $Q^1_j \in \mathcal{Q}^3 \setminus \mathcal{O}^1_{k-1}$ which intersects $\SubPOSS{\Brace{\mathcal{S}, \mathscr{P}}}{\ell_3 + 1 + w (k - 1)}{\ell_3 + 1 + w k - 1}$.
			Further, $\Brace{\mathcal{S}, \mathscr{P}}$ has length $\ell_2 = \ell_3 + (c+1) \ell_4$.
			Thus, we can construct such sets $\mathcal{S}^a_k, \mathcal{S}^b_k$ and $\mathcal{O}^1_k$.
			Let $\mathcal{S}^a = \mathcal{S}^a_{z_1}$,
			$\mathcal{S}^b = \mathcal{S}^b_{z_1}$,
			$\mathcal{O}^1 = \mathcal{O}^1_{z_1}$,
			$f_a = f_{a, z_1}$, and 
			$f_b = f_{b, z_1}$.
			
			Let $X = \V{\mathcal{O}^1} \cap \V{\SubPOSS{\Brace{\mathcal{S}, \mathscr{P}}}{\ell_3 + 1}{\ell_2}}$.
			We construct an $X$-$\End{\mathcal{Q}^3}$-linkage $\mathcal{Z}$ of order $z_1$ as follows.
			For each $O_j^1 \in \mathcal{O}^1$ choose some $x \in X \cap f_b(O_j^1)$ and add the $x$-$\End{\mathcal{O}^1} \subseteq \End{\mathcal{Q}^3}$ subpath of $O_j^1$ to $\mathcal{Z}$.
			Note that $\Abs{X} = \Abs{\mathcal{Z}} \geq z_1$.
			By choice of $P^2_e$, $\End{\mathcal{Z}} \subseteq \V{P^2_e}$.
			Let $a$ be the last arc of $P^1_e$.
			
			Construct a $\V{P^1_e}$-$\V{P^2_e}$-linkage $\mathcal{F}$ of order $c$ avoiding $a$ as follows.
			For each $O_i^1 \in \mathcal{O}^1$ let $S_j = f_b(O_i^1)$ and let $\Brace{a_{j,1}, a_{j,2}, \ldots , a_{j,w_3}} \coloneqq A(S_{j})$ and $\Brace{a_{j-1,1}, a_{j-1,2}, \ldots , a_{j-1,w_3}} \coloneqq A(S_{j-1})$ be ordered according to the orders witnessing that $A(S_{j})$ is $w_3$-order-linked to $B(S_j)$ and $A(S_{j-1})$ is $w_3$-order-linked to $B(S_{j-1})$.
			Let $L_i^2 \in \mathcal{L}_2$ be the path with $\V{L_i^2} \cap \V{\Brace{\mathcal{S}, \mathscr{P}}} \subseteq f_a(O^1_i) = S_{j-1}$ and let $F_i$ be a $\V{P^1_e}$-$a_{j-1,1}$-path in $\mathcal{Q}^4 \cdot \mathcal{L}_2$.
			
			The path of $w_3$-order-linked sets $\Brace{\mathcal{S}, \mathscr{P}}$ is contained within $\ToDigraph{\mathcal{M}' \cup \mathcal{Q}^M}$.
			By~\cref{subitem:1-horizontal-web-double-side-linkage:L2b}, $\mathcal{L}_2 \subseteq \mathcal{Q}^M$ holds.
			Further, $\mathcal{Q}^4$ is contained inside $\mathcal{Q}^2$.
			By choice of $\mathcal{Q}^4$, every path in $\mathcal{Q}^4$ intersects $P^1_e$.
			For each $O^1_i \in \mathcal{O}^1$ there is some $L_2^i \in \mathcal{L}_2$ such that $A(S_j) \subseteq \V{L_2^i}$, where $S_j = f_a(O_i^1)$.
			Hence, there is some $Q^4_i \in \mathcal{Q}^4$ such that $Q^4_i \cdot L_2^i$ contains a $\V{P^1_e}$-$a_{j-1,1}$-path as desired.
			Thus, the linkage $\mathcal{F}_1$ above exists.
			
			Construct an $\End{\mathcal{F}_1}$-$\Start{\mathcal{Z}'}$-linkage $\mathcal{F}_2$ as follows.
			For each $O^1_i \in \mathcal{O}^1$, let $S_j = f_b(O^1_i)$ and let $F_{4,i}$ be an $A(S_j)$-$x_i$-path in $S_j$, where $x_i \in \Start{\mathcal{Z}} \cap \V{O^1_i}$.
			Let $\Set{a_{j,k}} = \Start{F_{4,i}}$.
			Let $F_{3,i}$ be an $a_{j-1,1}$-$b_{j-1,k}$-path in $S_{j-1}$.
			As $\Set{b_{j-1,k}}$ is a 1-shift of $\Set{a_{j-1,1}}$ and $A(S_{j-1})$ is $\ell_3$-order-linked to $B(S_{j-1})$ in $S_{j-1}$, such a path $F_{3,i}$ exists.
			Now set $F_{2,i} = F_{3,i} \cdot P_{j-1,k} \cdot F_{4,i}$, where $P_{j-1,k} \in \mathcal{P}_{j-1}$ is the $b_{j-1,k}$-$a_{j,k}$-path in $\mathcal{P}_{j-1}$.
			
			Since $\Brace{\mathcal{S}, \mathscr{P}}$ is a uniform path of $w_3$-order-linked sets, each $F_{2,i}$ is a path.
			Let $\mathcal{F}_3 = \Set{F_{3,i} \mid 0 \leq i \leq c}$ and $\mathcal{F}_4 = \Set{F_{4,i} \mid 0 \leq i \leq c}$.
			As each $S_j$ contains at most one path of $\mathcal{F}_3 \cup \mathcal{F}_4$, we have that $\mathcal{F}_3$ and $\mathcal{F}_4$ are two disjoint linkages inside $\Brace{\mathcal{S}, \mathscr{P}}$.
			Hence, $\mathcal{F}_2 = \Set{F_{2,i} \mid 0 \leq i \leq c}$ is an $\End{\mathcal{F}_1}$-$\Start{\mathcal{Z}}$-linkage of order $c+1$ as desired.
			Finally, let $\mathcal{F}_5$ be the $\Start{\mathcal{F}_1}$-$\End{\mathcal{F}_2}$-linkage contained inside $\ToDigraph{\mathcal{F}_1 \cup \mathcal{F}_2}$.
			Since each path in one linkage intersects exactly one path in the other, we have that $\Abs{\mathcal{F}_5} = \Abs{\mathcal{F}_1}$.
			
			Construct an $\End{\mathcal{F}_5}$-$\V{P^2_e}$-linkage $\mathcal{F}_3$ of order $z_1$ by following the corresponding paths of $\mathcal{Z}$ until the first intersection with $P^2_e$.
			This is possible by choice of $\mathcal{Z}$.
			
			If there is some path in $\mathcal{F} \coloneqq \mathcal{F}_5 \cdot \mathcal{F}_3$ using $a$, we delete this path from $\mathcal{F}$.
			Hence, we obtain a $\V{P^1_e}$-$\V{P^2_e}$-linkage $\mathcal{F}$ of order at least $c$ inside $\ToDigraph{\mathcal{H}^1 \cup \mathcal{Q}^2 \cup \mathcal{Q}^M} - a$, contradicting the initial assumption that $\Brace{\mathcal{H}, \mathcal{V}}$ is a 2-horizontal web where \(\mathcal{H}\) is weakly $c$-minimal with respect to \(\mathcal{V}\).
		\end{claimproof}
		
		By~\cref{claim:Q5-disjoint-from-POWS}, there is some $\mathcal{Q}^5 \subseteq \mathcal{Q}^3$ of order $w_3$ and some $\ell_3 + 1 \leq i < j \leq \ell_2$ such that $\mathcal{Q}^5$ is internally disjoint from $\SubPOSS{\Brace{\mathcal{S}, \mathscr{P}}}{i}{j}$ and $j - i \geq \ell_4 - 1$.
		
		By~\cref{proposition:order-linked to path of well-linked sets}, the path of $w_3$-order-linked sets $\SubPOSS{\Brace{\mathcal{S}, \mathscr{P}}}{i}{j}$ contains a path of well-linked sets $\Brace{\mathcal{S}' = \Brace{S'_0, S'_1, \dots, S'_{\ell_4}}, \mathscr{P}' = \Brace{ \mathcal{P}'_0, \mathcal{P}'_1, \dots, \mathcal{P}'_{\ell_4 - 1}}}$ of width $w_3$ and length $\ell_4$ such that $A(S_0') \subseteq A(S_i)$ and $B(S_{\ell_4}') \subseteq B(S_{j})$.
		
		Construct a $B(S'_{\ell_4})$-$A(S'_0)$-linkage $\mathcal{R}$ of order $w_3$ as follows.
		
		By~\cref{lemma:linkage-inside-pools}, there is a  $B(S'_{\ell_4})$-$B(S_{\ell_2})$-linkage $\mathcal{Z}_5$	of order $w_3$ inside $\Brace{\mathcal{S}, \mathscr{P}}$.
		
		Let $\mathcal{L}_1' \subseteq \mathcal{L}_1$ be the linkage satisfying $\Start{\mathcal{L}_1'} = \End{\mathcal{Z}_5}$ and let $\mathcal{L}_3'' \subseteq \mathcal{L}_3'$ be the linkage satisfying	$\Start{\mathcal{L}_3''} = \End{\mathcal{Q}^5}$.
		Take an $\End{\mathcal{L}_1'}$-$\Start{\mathcal{L}_3''}$-linkage $\mathcal{X}_1$ of order $w_3$ in $\ToDigraph{\mathcal{H}^2 \cup \mathcal{Q}^5}$.
		Because	$\Brace{\mathcal{H}^2, \mathcal{Q}^5}$ is a web, and because $\End{\mathcal{L}_1'} \subseteq \End{\mathcal{H}^1} = \Start{\mathcal{H}^2}$ and $\Start{\mathcal{L}_3''} = \End{\mathcal{Q}^5}$ hold, by~\cref{lemma:web_mixed_linkages} such a linkage $\mathcal{X}_1$ exists.
		
		For each $i \in \Set{0, \ldots, w_3 - 1}$ let $X_{2,i}$ be a path inside $\ToDigraph{\mathcal{L}_3'' \cup \mathcal{L}_2'}$ which starts on $\Start{\mathcal{X}_1}$ and ends on $a_{2i,i} \in A(S_{2i})$, where $\Brace{a_{2i,0}, a_{2i,1}, \dots, a_{2i,w_3 - 1}} \coloneqq A(S_{2i})$ is sorted according to the order witnessing that $A(S_{2i})$ is $w_3$-order-linked to $B(S_{2i})$ inside $S_{2i}$.
		Let $\mathcal{X}_2 = \Set{X_{2,0}, X_{2,1}, \ldots, X_{2,w_5 - 1}}$.
		By choice of $\mathcal{L}_2'$ and	of $\mathcal{L}_3''$ and because~\cref{subitem:1-horizontal-web-double-side-linkage:L2b} holds, such a linkage $\mathcal{X}_2$ exists.
		
		Construct an $\End{\mathcal{X}_2}$-$A(S_{2(w_3 - 1)})$-linkage $\mathcal{X}_3$ inside $\Brace{\mathcal{S}, \mathscr{P}}$ as follows.
		First construct, for each $0 \leq i \leq w_3 - 1$, an $A(S_{2(i - 1)})$-$A(S_{2i})$-linkage $\mathcal{X}^i_3$ of order $i + 1$.
		Start with $\mathcal{X}_3^0 \coloneqq \Set{a_{0,0}} \subseteq A(S_0)$.
		
		On step $i \geq 1$,	let $\Brace{b_{2(i - 1),0}, b_{2(i - 1),1}, \dots, b_{2(i - 1), w_3 - 1}} \coloneqq B(S_{2(i - 1)})$ be the ordering of the set $B(S_{2(i - 1)})$ witnessing that $A(S_{2(i - 1)})$ is $w_3$-order linked to $B(S_{2(i - 1)})$.
		Let $\mathcal{Y}_i$ be an $\End{\mathcal{X}_3^{i-1}}$-$B_i$-linkage of order $i$ in $S_{2(i - 1)}$, where $B_i = \Set{b_{2(i - 1),j} \in B(S_{2(i - 1)}) \mid 1 \leq j \leq i}$.
		Since $A(S_{2(i - 1)})$ is $w_3$-order-linked to $B(S_{2(i - 1)})$ in $S_{2(i - 1)}$ and $\End{\mathcal{X}_3^{i-1}}$ contains the minimal $i$ elements of the corresponding ordering,	such a linkage $\mathcal{Y}_i$ exists.
		
		Let $\mathcal{Z}_i$ be a $B_i$-$A_i$-linkage of order $i + 1$ in $\mathcal{P}_{2(i-1)}$ such that $\End{\mathcal{Y}_i} \subseteq \Start{\mathcal{Z}_i}$,	where $A_i = \Set{a_{2i,j} \mid 1 \leq j \leq i + 1}$.
		Since $\Brace{\mathcal{S},\mathscr{P}}$ is a uniform path of $w_3$-order-linked sets, such a linkage $\mathcal{Z}_i$ exists.
		Set $\mathcal{X}_3^{i} = \mathcal{X}_3^{i-1} \cdot \mathcal{Y}_i \cdot \mathcal{Z}_i$.
		Since $\Abs{\mathcal{Y}_i \cdot \mathcal{Z}_i} = i + 1$, $\mathcal{X}_3^i$ is a $\Start{\mathcal{X}_3^i}$-$A_i$-linkage of order $i + 1$.
        (Recall that, by definition of the concatenation operation~$\cdot$, the additional path in $\mathcal{Z}_i$ which does not have a corresponding endpoint in $\mathcal{Y}_{i}$ is simply added to the result of the concatenation.)
		
		After iterating all the steps above, we obtain an $\End{\mathcal{X}_2}$-$A(S_{2(w_3 - 1)})$-linkage $\mathcal{X}_3 \coloneqq \mathcal{X}_3^{w_3 - 1}$ of order $w_3$ as desired.
		By~\cref{lemma:increase-order-linkedness}, $A(S_{2(w_3 - 1)})$ is $w_3$-order-linked to $A(S_{i}) \supseteq A(S_0')$ in	$\SubPOSS{\Brace{\mathcal{S}, \mathscr{P}}}{2(w_3 - 1)}{i}$.
		As $\Start{\mathcal{X}_3}$ contains the minimal $w_3$ elements of $A(S_{2(w_3 - 1)})$, the set $A(S_0')$ is an $w_3$-shift of $A(S_{2(w_3 - 1)})$.
		Hence, there is an $\End{\mathcal{X}_3}$-$A(S_0')$-linkage $\mathcal{X}_4$ of order $w_3$ in $\SubPOSS{\Brace{\mathcal{S}, \mathscr{P}}}{2(w_3 - 1)}{i}$.
		
		The concatenation $\mathcal{X}_2 \cdot \mathcal{X}_3 \cdot \mathcal{X}_4$ produces a half-integral $\Start{\mathcal{L}_3''}$-$A(S_0')$-linkage of order $w_3$.
		By~\cref{lemma:half_integral_to_integral_linkage} there is a $\Start{\mathcal{X}_2}$-$A(S_{0}')$-linkage $\mathcal{X}_5$ of order $w_4$ inside $\ToDigraph{\mathcal{X}_2 \cup \mathcal{X}_3}$.
		
		Let $\mathcal{X}_6 \subseteq \mathcal{Z}_5 \cdot \mathcal{L}_1' \cdot \mathcal{X}_1$ be the linkage of order $w_4$ with $\End{\mathcal{X}_6} = \Start{\mathcal{X}_5}$.
		We claim that $\mathcal{X}_6 \cdot \mathcal{X}_5$ is a half-integral $B(S_{\ell_4}')$-$A(S_0')$-linkage of order $w_4$.
		
		Assume towards a contradiction that there is some $v \in \V{\mathcal{Z}_5 \cdot \mathcal{L}_1'} \cap \V{\mathcal{X}_1} \cap \V{\mathcal{X}_5}$.
		Since $\mathcal{Z}_5 \cdot \mathcal{L}_1'$ is contained in $\ToDigraph{\mathcal{M}' \cup \mathcal{Q}^M}$ and $\mathcal{X}_1$ is contained inside $\ToDigraph{\mathcal{Q}^5 \cup \mathcal{H}^2}$, we have that $v \in \V{\mathcal{Q}^5} \cap \V{\mathcal{M}'}$.
		Furthermore, $v$ is not in $\SubPOSS{\Brace{\mathcal{S}, \mathscr{P}}}{0}{2(w_5 - 1)}$ as $\mathcal{Z}_5 \cdot \mathcal{L}_1'$ is disjoint from $\SubPOSS{\Brace{\mathcal{S}, \mathscr{P}}}{0}{2(w_3 - 1)}$ by construction.
		As $v \in \V{\mathcal{X}_5}$ and $\mathcal{X}_3 \cdot \mathcal{X}_4$ is contained inside the path of order-linked sets $\SubPOSS{\Brace{\mathcal{S}, \mathscr{P}}}{0}{2(w_3 - 1)}$,	we have that $v \in \V{\mathcal{X}_2} \subseteq \V{\mathcal{Q}^2}$ as well.
		This however implies that $v \in \V{\mathcal{Q}^2} \cap \V{\mathcal{Q}^5} = \Start{\mathcal{Q}^2}$.
		However, $\Start{\mathcal{Q}^2} \cap \V{\mathcal{M}'} = \emptyset$ by choice of $\mathcal{Q}^2$, a contradiction to the previous observation that $v \in \V{\mathcal{Q}^5} \cap \V{\mathcal{M}'}$.
		Hence, by~\cref{lemma:half_integral_to_integral_linkage}, $\ToDigraph{\mathcal{L}_1 \cdot \mathcal{X}_1 \cdot \mathcal{X}_6}$ contains a $B(S_{\ell_4}')$-$A(S_0')$-linkage $\mathcal{R}$ of order $w$.
		
		We show that $\V{\mathcal{R}} \cap \V{\SubPOSS{\Brace{\mathcal{S}, \mathscr{P}}}{i}{j}} \subseteq B(S_{\ell_4}') \cup A(S_0')$.
		
		Because~\cref{subitem:1-horizontal-web-double-side-linkage:L1b} holds, we have that $\mathcal{L}_1'$ is internally disjoint from $\Brace{\mathcal{S}, \mathscr{P}}$.
		By construction we have that $\V{\mathcal{X}_1} \cap \V{\Brace{\mathcal{M}', \mathcal{Q}^2}} \subseteq \V{\mathcal{Q}^5}$.
		By choice of $\mathcal{Q}^5$ we have that $\V{\mathcal{Q}^5} \cap \V{\SubPOSS{\Brace{\mathcal{S}, \mathscr{P}}}{i}{j}} = \emptyset$.
		Hence, $\mathcal{X}_1$ is disjoint from	$\SubPOSS{\Brace{\mathcal{S}, \mathscr{P}}}{i}{j}$.
		
		The linkage $\mathcal{X}_2$ is contained in	$\ToDigraph{\mathcal{L}''_3 \cdot \mathcal{L}'_2}$ and is thus disjoint from $\SubPOSS{\Brace{\mathcal{S}, \mathscr{P}}}{i}{j}$ because~\cref{subitem:1-horizontal-web-double-side-linkage:L2b} holds and $i > 2(w_3 - 1)$.
		
		The linkage $\mathcal{X}_5$ is contained in	$\SubPOSS{\Brace{\mathcal{S}, \mathscr{P}}}{0}{i}$ and is thus internally disjoint from $\SubPOSS{\Brace{\mathcal{S}, \mathscr{P}}}{i}{j}$.
		Hence, $\mathcal{X}_5$ is also internally disjoint from	$\SubPOSS{\Brace{\mathcal{S}, \mathscr{P}}}{i}{j}$.
		This implies that $V(\SubPOSS{\Brace{\mathcal{S}, \mathscr{P}}}{i}{j}) \cap \V{\mathcal{R}} \subseteq B(S_{\ell_4}') \cup A(S_0')$, as desired.
		Hence, $\Brace{\mathcal{S}', \Brace{ \mathcal{P}'_0, \mathcal{P}'_1, \dots, \mathcal{P}'_{\ell_4 - 1}, \mathcal{R}} }$ is a cycle of well-linked sets of width $w$ and length $\ell$, as desired.
	\end{proof}
	
\subsection{The Directed Grid Theorem}
\label{subsec:the-directed-grid-theorem}
	We are now ready to prove our main theorems.
	We state our main result in terms of cylindrical grids, walls  and cycles of well-linked sets, as each may be useful in a different setting.
	
	We define
	\begin{align*}
	\Fkt{m'}{w, \ell} & \coloneqq 
			\bound{lemma:coss-or-minimal-2-horizontal-web}{m}{\bound{lemma:coss-inside-2-horizontal-web}{h}{w,\ell }, w, \ell},
	\\[0em]
		\boundDefAlign{thm:POSS_plus_back-linkage_to_COSS}{w'}{w, \ell }
	\bound{thm:POSS_plus_back-linkage_to_COSS}{w'}{w, \ell } & \coloneqq 
			\bound{lemma:coss-or-back-linkage-intersects-cluster-by-cluster}{w'}{w,\bound{lemma:coss-or-minimal-2-horizontal-web}{w}{\bound{lemma:coss-inside-2-horizontal-web}{h}{w,\ell },w, \ell } },
	\\[0em]
		\boundDefAlign{thm:POSS_plus_back-linkage_to_COSS}{r}{w, \ell }
	\bound{thm:POSS_plus_back-linkage_to_COSS}{r}{w, \ell } & \coloneqq 
			\bound{lemma:coss-or-minimal-2-horizontal-web}{r}{\bound{lemma:coss-inside-2-horizontal-web}{h}{w,\ell },w,\ell, \bound{lemma:coss-inside-2-horizontal-web}{v}{w,\ell, \bound{lemma:coss-or-minimal-2-horizontal-web}{m}{\bound{lemma:coss-inside-2-horizontal-web}{h}{w,\ell }, w, \ell}} },
	\\[0em]
		\boundDefAlign{thm:POSS_plus_back-linkage_to_COSS}{\ell'}{w, \ell }
	\bound{thm:POSS_plus_back-linkage_to_COSS}{\ell'}{w, \ell} & \coloneqq 
			\bound{lemma:coss-or-back-linkage-intersects-cluster-by-cluster}{\ell'}{w, \ell, \bound{lemma:coss-or-minimal-2-horizontal-web}{\ell}{w, \ell}, \bound{thm:POSS_plus_back-linkage_to_COSS}{r}{w, \ell}}.
	\end{align*}
	Observe that
	\begin{align*}
		\bound{thm:POSS_plus_back-linkage_to_COSS}{w'}{w,\ell} & \in \PowerTower{2}{\Polynomial{97}{\ell, w}},
		\\[0em]
		\bound{thm:POSS_plus_back-linkage_to_COSS}{r}{w,\ell} & \in \PowerTower{13}{\Polynomial{97}{\ell, w}} \text{ and }
		\\[0em]
		\bound{thm:POSS_plus_back-linkage_to_COSS}{\ell'}{w,\ell} & \in \PowerTower{14}{\Polynomial{97}{\ell, w}}.
	\end{align*}

	\thmPOWLtoCOWS*
\begin{proof}
		Assume, without loss of generality, that $\mathcal{R}$ is weakly $r$-minimal with respect to	$\Brace{\mathcal{S}, \mathscr{P}}$ and that \(r = \bound{thm:POSS_plus_back-linkage_to_COSS}{r}{w,\ell}\).
		If this is not the case, we just choose a $\Start{\mathcal{R}}$-$\End{\mathcal{R}}$-linkage of order $\bound{thm:POSS_plus_back-linkage_to_COSS}{r}{w,\ell} \leq \Abs{\mathcal{R}}$	which is $\Brace{\mathcal{S}, \mathscr{P}}$-minimal.
		By~\cref{obs:H-minimal-implies-weakly-minimal}, such a linkage is also weakly $\bound{thm:POSS_plus_back-linkage_to_COSS}{r}{w,\ell}$-minimal with respect to $\Brace{\mathcal{S}, \mathscr{P}}$.

		We define
		$h = \bound{lemma:coss-inside-2-horizontal-web}{h}{w,\ell }$,
		$w_1 = \bound{lemma:coss-or-minimal-2-horizontal-web}{w}{h, w, \ell}$,
		$m = \bound{lemma:coss-or-minimal-2-horizontal-web}{m}{h, w, \ell}$,
		$v = \bound{lemma:coss-inside-2-horizontal-web}{v}{w,\ell,m}$
		and
		$\ell_1 = \bound{lemma:coss-or-minimal-2-horizontal-web}{\ell}{w,r }$.
		Observe that $w' \geq \bound{lemma:coss-or-back-linkage-intersects-cluster-by-cluster}{w'}{w, w_1}$, $\ell' \geq \bound{lemma:coss-or-back-linkage-intersects-cluster-by-cluster}{\ell'}{w, \ell, \ell_1, r}$ and $r \geq \bound{lemma:coss-or-minimal-2-horizontal-web}{r}{h,w,v}$.

		Applying~\cref{lemma:coss-or-back-linkage-intersects-cluster-by-cluster} to $(\mathcal{S}, \mathscr{P})$ and $\mathcal{R}$ yields two cases.
		If~\cref{item:coss-or-back-linkage-cluster-by-cluster:coss} holds, then we obtain a cycle of well-linked sets of width $w$ and length $\ell$ as desired.
		Otherwise,~\cref{item:coss-or-back-linkage-cluster-by-cluster:back-linkage} holds, and $\ToDigraph{\Brace{\mathcal{S},\mathscr{P}} \cup \mathcal{R}}$ contains a path of well-linked sets $\Brace{\mathcal{S}', \mathscr{P}'}$ of width $w_1$ and length $\ell_1$ with a back-linkage $\mathcal{R}'$ of order $w_1$ intersecting $\Brace{\mathcal{S}', \mathscr{P}'}$ cluster-by-cluster such that	$\mathcal{R}' \subseteq \mathcal{R}$.
		Note that $\mathcal{R}'$ is also weakly $r$-minimal with respect to	$\Brace{\mathcal{S}', \mathscr{P}'}$.

		Applying~\cref{lemma:coss-or-minimal-2-horizontal-web} to $\Brace{\mathcal{S}', \mathscr{P}'}$ and $\mathcal{R}'$ yields two further cases.
		If~\cref{item:coss-or-minimal-2-horizontal-web:coss} holds,	then we obtain a cycle of well-linked sets of width $w$ and length $\ell$ as desired.
		Otherwise,~\cref{item:coss-or-minimal-2-horizontal-web:c-horizontal-web} holds, and we obtain a $2$-horizontal $(h,v)$-web $\Brace{\mathcal{H}, \mathcal{V}}$ such that $\mathcal{H}$ is weakly $m$-minimal with respect to $\mathcal{V}$.

		By~\cref{lemma:coss-inside-2-horizontal-web}, $\Brace{\mathcal{H}, \mathcal{V}}$ contains a cycle of well-linked sets of width $w$ and length $\ell$, as desired.
\end{proof}

	We define 
	\(
        \bound{thm:high_dtw_to_COSS}{dtw}{w, \ell} \coloneqq 
		\bound{thm:high_dtw_to_POSS_plus_back-linkage}{t}{\bound{thm:POSS_plus_back-linkage_to_COSS}{w'}{w, \ell} + \bound{thm:POSS_plus_back-linkage_to_COSS}{r}{w, \ell}, \bound{thm:POSS_plus_back-linkage_to_COSS}{\ell'}{w, \ell}},
    \)
	   \boundDef{thm:high_dtw_to_COSS}{dtw}{w,\ell} and note that
    \(
        \bound{thm:high_dtw_to_COSS}{dtw}{w,\ell} \in \PowerTower{21}{\Polynomial{97}{w, \ell}}
    \).
	The following theorem is our main result, stated in terms of cycles of well-linked sets.
	
	\begin{restatable}{theorem}{thmHighDTWToCOSS}
		\label{thm:high_dtw_to_COSS}
		Let $w, \ell$ be integers.
		Every digraph $D$ with $\dtw{D} \geq \bound{thm:high_dtw_to_COSS}{dtw}{w, \ell}$ contains a cycle of well-linked sets $\Brace{\mathcal{S},\mathscr{P}}$ of width $w$ and length $\ell$.
	\end{restatable}
	\begin{proof}
		Let
		$r_1 = \bound{thm:POSS_plus_back-linkage_to_COSS}{r}{w, \ell}$,
		$w_1 = \bound{thm:POSS_plus_back-linkage_to_COSS}{w'}{w, \ell} + r_1$ and
		$\ell_1 = \bound{thm:POSS_plus_back-linkage_to_COSS}{\ell'}{w, \ell}$.

		By~\cref{thm:high_dtw_to_POSS_plus_back-linkage}, $D$ contains a path of well-linked sets $\Brace{\mathcal{S} = \Brace{ S_0, S_1, \dots, S_{\ell_1}}, \mathscr{P}}$ of width $w_1$ and length $\ell_1$ where $B(S_{\ell_1})$ is well-linked to $A(S_0)$ in $D$.
		Hence, there is a $B(S_{\ell_1})$-$A(S_0)$ linkage $\mathcal{R}$ of order $r_1$ in $D$.
		By~\cref{thm:POSS_plus_back-linkage_to_COSS}, $\ToDigraph{\mathcal{S} \cup \mathscr{P} \cup \mathcal{R}}$ contains a cycle of well-linked sets $\Brace{\mathcal{S}', \mathscr{P}'}$ of width $w$ and length $\ell$.
	\end{proof}
	
	We close this section by stating our main result in terms of cylindrical grids and walls. 
	Define
    \(\bound{thm:high-treewidth-implies-cylindrical-grid}{dtw}{k} \coloneqq \bound{thm:high_dtw_to_COSS}{dtw}{\bound{thm:coss contains grid}{w}{k}, \bound{thm:coss contains grid}{\ell}{k}}.\)
	Note that
	\(\bound{thm:high-treewidth-implies-cylindrical-grid}{dtw}{k} \in \PowerTower{22}{\Polynomial{9}{k}}\).
	
	\thmCylindricalGrid*
	\begin{proof}
		By~\cref{thm:high_dtw_to_COSS},	$D$ contains a cycle of well-linked sets $\Brace{\mathcal{S}, \mathscr{P}}$ of width $\bound{thm:coss contains grid}{w}{k}$ and length $\bound{thm:coss contains grid}{\ell}{k}$.
		By~\cref{thm:coss contains grid}, $\Brace{\mathcal{S}, \mathscr{P}}$ contains a cylindrical grid of order $k$.
	\end{proof}

	\begin{theorem}
		\label{statement:high-dtw-to-wall}
		Every digraph \(D\) of \(\DTreewidth{D} \geq \bound{statement:high-dtw-to-wall}{dtw}{k} \coloneqq \bound{thm:high-treewidth-implies-cylindrical-grid}{dtw}{2k - 2}\)
		contains the cylindrical wall of order \(k\) as a topological minor.
	\end{theorem}
	\begin{proof}
		By \cref{thm:high-treewidth-implies-cylindrical-grid}, \(D\) contains a cylindrical grid of order 
		\(2k - 2\) as a butterfly minor.
		By \cref{statement:grid-contains-wall},
		\(D\) contains the wall of order \(k\) as a butterfly minor.
		By \cref{statement:degree-3-topological-butterfly-minor},
		\(D\) contains the wall of order \(k\) as a topological minor.
	\end{proof}

\section{Conclusion}
\label{sec:conclusion}

While we significantly improved the bounds for the Directed Grid Theorem, the functions are still very large and likely far from optimal.
Determining the best bounds is an interesting open question.

Our modular approach allows us to break the question above down into smaller parts, facilitating further improvements to the Directed Grid Theorem.
In particular, the concept of 2-horizontal webs played a central role in our proof, and studying this object more closely may lead to new, valuable insights.
For example, consider the construction given in \cref{fig:horizontal-web-inside-grid} and the following simple observation.

\begin{figure}[H]
		\centering
		\includegraphics{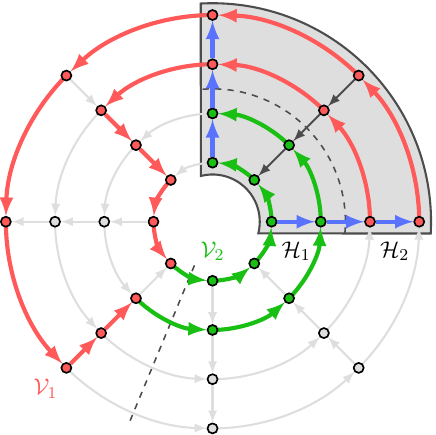}
		\caption{An illustration of the construction in the proof of \cref{statement:horizontal-web-inside-grid}.
		The area with a gray background corresponds to where the linkage \(\mathcal{H}\) is,
		with the dashed line marking where \(\mathcal{H}_1\) ends and where \(\mathcal{H}_2\) starts.
		The dashed gray line outside of the gray area delimits where \(\mathcal{V}_1\) ends and where
		\(\mathcal{V}_2\) starts.
		The light gray arcs and vertices are not used by the horizontal web.
			\label{fig:horizontal-web-inside-grid}}
\end{figure}

\begin{observation}
	\label{statement:horizontal-web-inside-grid}
		Every cylindrical grid of order \(2k\) contains a 2-horizontal web \((\mathcal{H}, \mathcal{V})\) where \(\mathcal{H}\) is minimal with respect to \(\mathcal{V}\) and \(\Abs{\mathcal{H}} = k =\Abs{\mathcal{V}}\).
\end{observation}
\begin{proof}
	Let \(C_{1}, C_{2}, \ldots, C_{2k}\) be the cycles and let \(P_{1}, P_{2}, \ldots, P_{4k}\) be the paths of a cylindrical grid of order \(2k\).

	Let \(\mathcal{H} = \{P_1, P_3, \ldots, P_{2k - 1}\}\) and let \(\mathcal{H}_1 \cdot \mathcal{H}_2 = \mathcal{H}\) such that \(\mathcal{H}_2\) contains the subpaths of \(\mathcal{H}\) intersecting \(C_{1}, C_{2}, \ldots, C_{k}\) and \(\mathcal{H}_1\) contains the subpaths of \(\mathcal{H}\) intersecting \(C_{k + 1}, C_{k + 2}, \ldots, C_{2k}\).

	Let \(\mathcal{V}_1\) be the subpaths of \(C_{1}, C_{2}, \ldots, C_{k}\) from \(\V{P_1}\) to \(\V{P_{2k}}\).
	Note that \(\mathcal{V}_1\) intersects all paths of \(\mathcal{H}_2\).

	Construct \(\mathcal{V}_2\) as follows.
	Let \(\mathcal{R}\) be the linkage from \(\End{\mathcal{V}_1}\) to \(\bigcup_{i = k + 1}^{2k}\V{C_i} \cap \V{P_{4k}}\).
	This linkage exists since we can use each of the \(k\) paths \(P_{2k}, P_{2k + 2}, \ldots, P_{4k - 2}\) in order to reroute one path from a vertex of \(C_{1}, C_{2}, \ldots, C_{k}\) to a vertex of \(C_{k + 1}, C_{k + 2}, \ldots, C_{2k}\).
	We construct \(\mathcal{V}_2\) by first taking \(\mathcal{R}\) and then the subpaths of \(C_{k + 1}, C_{k + 2}, \ldots, C_{2k}\) from \(P_{4k - 2}\) to \(P_{2k - 1}\).
	Note that \(\mathcal{V}_2\) intersects all paths of \(\mathcal{H}_1\) without intersecting any path of \(\mathcal{H}_2\).
	Let \(\mathcal{V} = \mathcal{V}_1 \cdot \mathcal{V}_2\).
	Hence, \((\mathcal{H}, \mathcal{V})\) is a 2-horizontal web where \(\mathcal{H}\) is minimal with respect to \(\mathcal{V}\).
	Further, \(\Abs{\mathcal{H}} = \Abs{\mathcal{V}} = k\).
\end{proof}

In~\cref{lemma:coss-inside-2-horizontal-web}, we essentially prove that the converse holds for some elementary function.
Do we obtain a polynomial bound for the other direction as well?
Alternatively, one can start with the weaker question, asking for a cycle of well-linked sets instead.

\begin{question}
	Are there polynomials \(h,v\) such that every 2-horizontal web \((\mathcal{H}, \mathcal{V})\) where \(\mathcal{H}\) is weakly \(c\)-minimal with respect to \(\mathcal{V}\) with \(\Abs{\mathcal{H}} \geq h(w, \ell, c)\) and \(\Abs{\mathcal{V}} \geq v(w, \ell, c)\) contains a cycle of well-linked sets of length \(\ell\) and width \(w\)?
\end{question}

Similarly, it is not difficult to see that a 2-horizontal web has a large directed \treewidth.
To see this, recall the definition of \emph{haven} from~\cite{johnson2001directed}.
A haven of order \(w\) in \(D\) is a function \(\beta\) assigning to every set \(Z \subseteq \V{D}\) with \(\Abs{Z} < w\) the vertex-set of a strong component of \(D \setminus Z\) in such a way that if \(Z' \subseteq Z \subseteq \V{D}\), then \(\Func{\beta}{Z} \subseteq \Func{\beta}{Z'}\).
Johnson, Robertson, Seymour and Thomas~\cite{johnson2001directed} prove that, if a digraph contains a haven of order \(w\), then its directed \treewidth is at least \(w - 1\).

\begin{observation}
	Every 2-horizontal web \((\mathcal{H}, \mathcal{V})\) contains a haven of order \(\min(h,v)\), where \(h = \Abs{\mathcal{{H}}}\) and \(v = \Abs{\mathcal{V}}\).
	In particular, \((\mathcal{H}, \mathcal{V})\) has directed \treewidth at least \(\min(h,v) - 1\).
\end{observation}
\begin{proof}
	Let \(w = \min(h,v)\).
	Let \(Z \subseteq \V{D}\) be some set of size less than \(w\).
	Let \(\mathcal{H}' \subseteq \mathcal{H}\) and \(\mathcal{V}' \subseteq \mathcal{V}\) such that neither \(\mathcal{H}'\) nor \(\mathcal{V}'\) contain any vertex of \(Z\).
	Let \(\mathcal{V}''\) be the sublinkage of \(\mathcal{V}'\) where the paths start and end inside \(\V{\mathcal{H}'}\).
	Similarly, let \(\mathcal{H}''\) be the sublinkage of \(\mathcal{H}'\) where the paths start and end inside \(\V{\mathcal{V}''}\).
	Observe that \((\mathcal{H}'', \mathcal{V}'')\)	is also a 2-horizontal web.

	Let \(\mathcal{H}_1 \cdot \mathcal{H}_2\) be the partition of \(\mathcal{H}''\) and let \(\mathcal{V}_1 \cdot \mathcal{V}_2\) be the partition of \(\mathcal{V}''\) witnessing that \((\mathcal{H}'', \mathcal{V}'')\) is a 2-horizontal web.

	As \(\mathcal{V}_1\) intersects all paths of \(\mathcal{H}_2\) and \(\mathcal{V}_2\) intersects all paths of \(\mathcal{H}_1\), it is easy to see that the digraph \(D'\) induced by \(\mathcal{H}'' \cup \mathcal{V}''\) is strongly connected.
	Hence, we assign \(\Func{\beta}{Z} = D'\).
	Clearly, \(\beta\) is a haven of order \(w\).
\end{proof}

\Cref{lemma:coss-or-minimal-2-horizontal-web,thm:high_dtw_to_POSS_plus_back-linkage} essentially imply that the converse statement holds with an elementary function between the directed \treewidth and the size of the 2-horizontal web.
Can we also obtain a polynomial bound?

\begin{question}
	Is there a polynomial \(w\) such that every digraph \(D\) with \(\DTreewidth{D} \geq w(h,v,c)\) contains a  2-horizontal web \((\mathcal{H}, \mathcal{V})\) where \(\mathcal{H}\) is weakly \(c\)-minimal with respect to \(\mathcal{V}\) with \(\Abs{\mathcal{H}} \geq h\) and \(\Abs{\mathcal{V}} \geq v\)?
\end{question}

If there is a polynomial bound for the Directed Grid Theorem, the answer to both questions above is \say{yes}.
Conversely, obtaining good functions for either question (even if not polynomial) would immediately lead to better bounds for the Directed Grid Theorem by applying our framework.
Hence, if one intends to improve the bounds of the Directed Grid Theorem further, it is only natural to consider both questions above, as they are necessary and sufficient steps towards this end.

\bibliographystyle{alphaurl}
\bibliography{literature.bib}

\end{document}